%% file: main.tex
\documentclass[a4paper,USenglish,cleveref, autoref, thm-restate,nolineno]{socg-lipics-v2021}

\pdfoutput=1
\hideLIPIcs

\input{macros}

\title{Generalized Graph Packing Problems Parameterized by Treewidth} 

\author{Bar\i\c{s} Can Esmer}
  {CISPA Helmholtz Center for Information Security, Germany}
  {baris-can.esmer@cispa.de}
  {https://orcid.org/0000-0001-5694-1465}{The author is part of Saarbrücken Graduate School of Computer Science, Germany.}

\author{D\'aniel Marx}
  {CISPA Helmholtz Center for Information Security, Germany}
  {marx@cispa.de}
  {https://orcid.org/0000-0002-5686-8314}{}

\authorrunning{B.\,C.~Esmer, D.~Marx}  

\Copyright{Bar\i\c{s} Can Esmer, D\'aniel Marx}



\begin{document}

\maketitle

\begin{abstract}
	\input{abstract}
\end{abstract}

\section{Introduction}\label{sec:intro}
\input{introduction}

\subsection{Our Results}\label{sec:results}
\input{results}

\section{Technical Overview}\label{sec:tech_overview}
\input{overview}

\section{Gadgets and Relations}
\label{section:gadgets}
\input{gadgets}

\section{Lower Bounds for Clique Partitioning Problems}
\label{section:lb_clique}
\input{lower_bounds_for_clique_packing_problems}

\section{Lower Bound for $\graphpartprob{H}$ Where $H$ Is Not a Block-Graph}
\label{section:lb_arbitrary}
\input{lower_bounds_arb_graph}

\section{Algorithm for $\singlecliquepackprob{c,d}$ Parameterized by Treewidth}
\label{section:algo_clique}
\input{algorithm_clique}

\section{Algorithm for $\graphpackprob{H}$ Where $H$ Is an Arbitrary Graph }
\label{section:algo_arb}
\input{algorithm_general_graph}

\bibliographystyle{plain}
\bibliography{bib}

\end{document}

%% file: macros.tex
\usepackage{amsmath, amssymb, amstext, amsthm, mathtools, bm, dsfont}
\usepackage{cancel}
\usepackage{xspace}
\usepackage{stmaryrd}
\usepackage[numbers,square,comma,longnamesfirst,sort]{natbib}
\usepackage[procnumbered, linesnumbered,ruled]{algorithm2e}

\SetCommentSty{mycommfont}
\usepackage{framed}
\usepackage[framemethod=TikZ]{mdframed}

\newcommand{\coherent}[1]{#1\textrm{-internally coherent}}
\newcommand{\regular}[1]{#1\textrm{-regular}}
\newcommand{\internal}[1]{\mathrm{int}\left( #1 \right)}

\newcommand{\Compl}[2]{\mathrm{Compl}_{#1}\bigl(#2\bigr)}

\numberwithin{equation}{section}
\numberwithin{figure}{section}

\newcommand{\stack}[1]{\operatorname{stack}\left(#1\right)}

\newcommand{\packing}[1]{#1\textrm{-packing}}
\newcommand{\partitioning}[1]{#1\textrm{-partition}}

\newcommand{\multipacking}[2]{\big(#1, #2\big)\textrm{-multi-packing}}
\newcommand{\multipartitioning}[2]{\big(#1, #2\big)\textrm{-multi-partition}}

\newcommand{\singlepacking}[2]{\left(#1, #2\right)\textrm{-single-packing}}
\newcommand{\singlepartitioning}[2]{\left(#1, #2\right)\textrm{-single-partition}}

\newcommand{\multicliquepacking}[2]{\multipacking{#1}{K_{#2}}}

\newcommand{\singlecliquepacking}[2]{\singlepacking{#1}{K_{#2}}}

\newcommand{\partialpacking}[1]{\textrm{partial }#1{-Packing}}

\newcommand{\graphpackprob}[1]{\textnormal{\texttt{Generalized Graph Packing}}(#1)}
\newcommand{\multicliquepackprob}[1]{\textnormal{\texttt{Multi-Clique Packing}}(#1)}
\newcommand{\singlecliquepackprob}[1]{\textnormal{\texttt{Single-Clique Packing}}(#1)}

\newcommand{\graphpartprob}[1]{\textnormal{\texttt{Generalized Graph Partitioning}}(#1)}
\newcommand{\multicliquepartprob}[1]{\textnormal{\texttt{Multi-Clique Partitioning}}(#1)}
\newcommand{\singlecliquepartprob}[1]{\textnormal{\texttt{Single-Clique Partitioning}}(#1)}

\newcommand{\kkclique}{k \times k\textsc{ Permutation Clique}}
\newcommand{\kkindset}{k \times k\textsc{ Permutation Independent Set}}

\newcommand{\arb}[1]{arb-$\left( #1 \right) $}
\newcommand{\dist}[1]{dist-$\left( #1 \right) $}
\newcommand{\strict}[1]{strict-$\left( #1 \right) $}

\newcommand{\generaltype}[1]{\mathrm{TYPE}_{#1}}
\newcommand{\cliquetype}[1]{\mathrm{TYPE}_{#1}}

\newcommand{\partialpackings}[2]{\mathrm{PP}(#1, #2)}
\newcommand{\numfullcopies}[1]{\nu(#1)}

\crefname{claim}{Claim}{Claims}
\crefname{enumi}{Item}{Items} 
\Crefname{enumi}{Item}{Items} 
\crefname{equation}{}{} 

\newcommand{\from}{\colon}

\newcommand{\im}{\mathrm{im}}
\newcommand{\dom}[1]{\mathrm{dom}\left( #1 \right)}
\DeclarePairedDelimiter{\abs}{\lvert}{\rvert}
\DeclarePairedDelimiter{\iverson}{\llbracket}{\rrbracket}

\newcommand{\border}[2]{\mathrm{Border}^{#2}_{#1}}
\newcommand{\partition}{\mathbf{part}}
\newcommand{\upvertex}{\uparrow}
\newcommand{\downvertex}{\downarrow}

\newcommand{\covernum}[1]{\#\mathrm{Cover}\left( #1 \right)}

\newcommand{\aleft}[1]{a_{\mathrm{left}}^{(#1)}}
\newcommand{\aright}[1]{a_{\mathrm{right}}^{(#1)}}
\newcommand{\aup}[1]{a_{\mathrm{top}}^{(#1)}}
\newcommand{\bleft}[1]{b_{\mathrm{left}}^{(#1)}}
\newcommand{\bright}[1]{b_{\mathrm{right}}^{(#1)}}

\newcommand{\coverrel}[2]{\mathrm{COVER}_{#1}^{#2}}
\newcommand{\eqrel}[2]{\mathrm{EQ}_{#1}^{#2}}

\newcommand{\cneqrel}[1]{\mathrm{NEQ}^{#1}}
\newcommand{\sumrel}[2]{\mathrm{SUM}_{#1}^{#2}}

\newcommand{\copyrel}{\mathrm{COPY}}

\newcommand{\neqgadget}{\mathrm{NEQ}}
\newcommand{\w}{\mathbf{w}}
\newcommand{\pw}{\mathbf{pw}}

\newcommand{\blocks}{\mathcal{B}_H}

\newcommand{\relheavy}[2]{\mathcal{S}^{#1}_{#2} }

\newcommand{\srelheavy}[1]{\mathcal{S}_{#1} }

\newcommand{\ppseth}{$\pw$\textsc{-SETH}}
\newcommand{\seth}{\textsc{SETH}}
\newcommand{\ethh}{\textsc{ETH}}
\newcommand{\sat}{\textsc{SAT}}

\newcommand{\csp}{\textsc{CSP}}
\newcommand{\kcsp}[1]{$#1$-\textsc{CSP}}

\newcommand{\restr}{\mathbin{\downharpoonright}}
\newcommand{\incr}{\mathbin{\upharpoonright}}
\newcommand{\funcrem}[2]{{#1}\restr_{-#2}}
\newcommand{\funcrest}[2]{{#1}\restr_{#2}}
\newcommand{\funcadd}[3]{{#1}\incr_{[#2 \to #3]}}

\newcommand{\genupdintr}{\mathrm{UpdIntr}}
\newcommand{\genupdforgetone}{\mathrm{UpdForget1}}
\newcommand{\genupdforgettwo}{\mathrm{UpdForget2}}

\newcommand{\cliqueupdforget}{\mathrm{UpdForget}}

\let\eps\epsilon

\newcounter{openquestion}

\newcommand{\defproblem}[3]{
  \vspace{1mm}
\noindent\fbox{
  \begin{minipage}{0.96\textwidth}
  \begin{tabular*}{\textwidth}{@{\extracolsep{\fill}}lr} #1 \\ \end{tabular*}
  {\bf{Input:}} #2  \\
  {\bf{Question:}} #3
  \end{minipage}
  }
  \vspace{1mm}
}

\newcommand{\defproblemoutput}[3]{
  \vspace{1mm}
\noindent\fbox{
  \begin{minipage}{0.96\textwidth}
  \begin{tabular*}{\textwidth}{@{\extracolsep{\fill}}lr} #1 \\ \end{tabular*}
  {\bf{Input:}} #2  \\
  {\bf{Output:}} #3
  \end{minipage}
  }
  \vspace{1mm}
}

%% file: abstract.tex
\textsc{$H$-Packing} is the problem of finding a maximum number of vertex-disjoint copies of $H$ in a given graph $G$. \textsc{$H$-Partition} is the special case of finding a set of vertex-disjoint copies that cover each vertex of $G$ exactly once. Our goal is to study these problems and some generalizations on bounded-treewidth graphs. The case of $H$ being a triangle is well understood: given a tree decomposition of $G$ having treewidth $\textup{tw}$, the \textup{$K_3$-Packing} problem can be solved in time $2^\textup{tw}\cdot n^{O(1)}$, while Lokshtanov et al.~[{\it ACM Transactions on Algorithms} 2018] showed, under the Strong Exponential-Time Hypothesis (SETH), that there is no $(2-\epsilon)^\textup{tw}\cdot n^{O(1)}$ algorithm for any $\epsilon>0$ even for \textup{$K_3$-Partition}. Similar results can be obtained for any other clique $K_d$ for $d\ge 3$. We provide generalizations in two directions:
\begin{itemize}
\item We consider a generalization of the problem where every vertex can be used at most $c$ times for some $c\ge 1$. When $H$ is any clique $K_d$ with $d\ge 3$, then we give upper and lower bounds showing that the optimal running time increases to $(c+1)^\textup{tw}\cdot n^{O(1)}$. We consider two variants depending on whether a copy of $H$ can be used multiple times in the packing.
\item If $H$ is not a clique, then the dependence of the running time on treewidth may not be even single exponential. Specifically, we show that if $H$ is any fixed graph where not every 2-connected component is a clique, then there is no $2^{o(\textup{tw}\log \textup{tw})}\cdot n^{O(1)}$ algorithm for \textsc{$H$-Partition}, assuming the Exponential-Time Hypothesis (ETH).
\end{itemize}

%% file: introduction.tex
Parameterized complexity theory
has proven instrumental in systematically understanding the computational
complexity of various combinatorial problems under different parameterizations.
Parameterization by treewidth implies tractability for a large number of fundamental algorithmic problems. A prominent line of research has
emerged around classifying the complexity of classical NP-hard graph problems
under this parameterization framework \cite{lokshtanovKnownAlgorithmsGraphs2018,lampisPrimalPathwidthSETH2025}.

In their influential work, Lokshtanov et al.
\cite{lokshtanovKnownAlgorithmsGraphs2018} studied six classical combinatorial problems
for which parameterized algorithms are known where the running-time dependence on treewidth is single exponential. They showed that, under Strong Exponential
Time Hypothesis (SETH), the running times
are optimal in the base of the exponent. Following this work, significant efforts
have been devoted to generalizing and extending these results.
In particular, five out of six problems in \cite{lokshtanovKnownAlgorithmsGraphs2018}
have been put into a wider context and generalized to an infinite family of problems:  $q$-Coloring 
 was generalized into $H$-homomorphism problems \cite{canesmerFundamentalProblemsBoundedTreewidth2024,DBLP:journals/siamcomp/OkrasaR21,DBLP:conf/esa/OkrasaPR20}, Independent Set (equivalent to
Vertex Cover), MaxCut and Odd Cycle Transversal \cite{canesmerListHomomorphismsDeleting2024} into $H$-homomorphism deletion problem, and Dominating Set into general $(\sigma,\rho)$-dominating set problems \cite{fockeTightComplexityBounds2025}. Curticapean and Marx \cite{curticapean_tight_2016} showed tight lower bounds for the problem of counting perfect matchings, which was later extended into general factor problems. However, from the initial six problems studied by Lokshtanov et al.
\cite{lokshtanovKnownAlgorithmsGraphs2018}, the triangle packing problem has remained unaddressed by prior generalization efforts.

The $2^{tw} \cdot n^{\mathcal{O}\left(1\right)}$ running time for triangle packing stated in \cite{lokshtanovKnownAlgorithmsGraphs2018} is actually valid for any clique packing problem,
instead of triangle with just three vertices. In this paper, we investigate further generalizations of
the triangle packing problem. This line of research naturally gives rise to two
conceptually distinct directions:

\begin{enumerate}
	\item \textbf{Allowing vertices to be covered multiple times:} The
		motivation for this generalization stems naturally from closely
		related problems such as fractional packing, where vertices can
		inherently contribute fractionally or multiple times to
		different packings. 
		To illustrate, \cref{fig:triangle_partition} shows a graph that
		does not admit a triangle partition, yet allows a collection of
		triangles that cover each vertex exactly twice.
		Motivated by this observation, we formally define and explore a
		generalized packing problem, where each vertex of $G$ can be
		covered at most $c$ times, for any fixed integer $c \geq 1$.
		Unlike $c = 1$ case, this generalization demands careful
		consideration of the definition, specifically on whether 
		a copy of $H$ may appear multiple times in the solution.
		
	\item \textbf{Considering packings of arbitrary graphs:} The second
		natural direction involves generalizing triangle packing by
		replacing triangles with an arbitrary fixed graph $H$. This
		problem has already been studied in the literature under the
		name of $H$-partition in
		\cite{hellGeneralizedMatchingProblems1981}. Specifically, in
		\cite{hellGeneralizedMatchingProblems1981} the authors prove
		the NP-Hardness of the $H$-partition problem where $H$ has a
		connected component with at least three vertices. In this paper
		we let $H$ be a connected graph with at least three vertices
		and study the complexity of the $H$-partition
		problem parameterized by treewidth.
\end{enumerate}

\begin{figure}[htpb]
  \centering
  \begin{subfigure}[t]{0.48\textwidth}
    \centering
    \includegraphics[page=1, width=\textwidth]{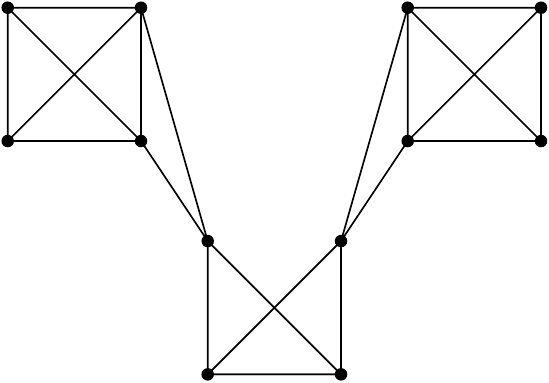}
    \caption{The graph under consideration.}
    \label{fig:triangle_partition_1}
  \end{subfigure}
  \hfill
  \begin{subfigure}[t]{0.48\textwidth}
    \centering
   \includegraphics[page=2, width=\textwidth]{figures/triangle_partition.pdf}
    \caption{Each color represents a triangle in the collection. Observe that each vertex is
    covered exactly twice.}
    \label{fig:triangle_partition_2}
  \end{subfigure}
  \caption{A graph that does not admit a triangle partition, but allows a collection of triangles such that each vertex is covered exactly twice.}
  \label{fig:triangle_partition}
\end{figure}

%% file: results.tex
Let $H$ be a fixed graph and $c \geq 1$ be an integer.
Given a graph $G$, a subgraph $Z$ of $G$ is called a copy of $H$
if $Z$ is isomorphic to $H$.
Moreover, if $Z$ is a copy of $H$ in $G$,
we say that $Z$ covers $v$ if $v \in V(Z)$.

\begin{definition}\label{definition:multi_packing}
	We say that a multiset (respectively, set) $\mathcal{S} = \Bigl\{ \bigl( V_1, E_1 \bigr) , \ldots, \bigl(V_k, E_k\bigr)\Bigr\}$ of subgraphs of $G$
	is a $\multipacking{c}{H}$ (respectively, $\singlepacking{c}{H}$) in $G$ if
	\begin{enumerate}
		\item each $(V_i, E_i)$ is isomorphic to $H$ for $1 \leq i \leq k$
		\item each vertex $v$ of $G$ is covered at most $c$ times by the subgraphs in $\mathcal{S}$.
	\end{enumerate}
	The collection $\mathcal{S}$ is called a $\multipartitioning{c}{H}$ (respectively, $\singlepartitioning{c}{H}$) of $G$ if each vertex $v \in V(G)$
	is covered exactly $c$ times.
\end{definition}

Observe that when $c$ is equal to $1$, two copies of $H$ are not allowed to have
any common vertices. Therefore, in this case, the notions of $\singlepacking{c}{H}$
and $\multipacking{c}{H}$ coincide, and we write $\packing{H}$ to simplify the notation.
We define $\partitioning{H}$ in a similar way.
The first problem we introduce, $\graphpackprob{H}$, asks the maximum size of an $\packing{H}$ in $G$.

\defproblemoutput{$\graphpackprob{H}$} 
{A graph $G$}
{The size of a largest $\packing{H}$ in $G$.}

For a fixed $H$, we show that $\graphpackprob{H}$
can be solved in time $2^{\mathcal{O}\left(w \cdot \log w\right)} \cdot n^{\mathcal{O}\left(1\right)}$
for graphs of treewidth $w$.

\begin{restatable}{theorem}{graphpackingarbitraryalgo}\label{theorem:graph_packing_arbitrary_algo}
	Let $H$ be an arbitrary graph such that it contains at least 3 vertices. Then, $\graphpackprob{H}$
	can be solved in time $2^{\mathcal{O}\left(w \cdot \log(w)\right)} \cdot n^{\mathcal{O}\left(1\right)}$
	for all $n$-vertex graphs $G$ where $w$ is the treewidth of $G$.
\end{restatable}

Then, we define partitioning problems in which each vertex of $G$ must be
covered the same number of times. For an arbitrary graph $H$, we let $c = 1$
which results in
vertex disjoint copies of $H$.

\defproblem{$\graphpartprob{H}$}
{A graph $G$}
{Is there an $\packing{H}$ of $G$ such that each vertex is covered exactly once?}

Observe that $\graphpartprob{H}$ is a special case of the problem $\graphpackprob{H}$,
which implies an algorithm with running time $2^{\mathcal{O}\left(w \cdot \log w\right)} \cdot n^{\mathcal{O}\left(1\right)}$
for the $\graphpartprob{H}$ problem on graphs of treewidth $w$.
We show that this running time cannot be improved for many choices of $H$.

\begin{theorem}\label{theorem:graph_packing_arbitrary_lower_bound}
	Let $H$ be an arbitrary graph with at least 3 vertices such that
	$H$ is not a block graph.
	Then, there is no
	algorithm for $\graphpartprob{H}$ problem,
	that solves all instances $G$
	in time $2^{o\left(w \cdot \log w\right)} \cdot n^{\mathcal{O}\left(1\right)}$
	where $w$ is the treewidth of $G$,
	unless \ethh\, fails.
\end{theorem}

When $H$ is a clique, we consider the more general problem in which a vertex
$v \in V(G)$ can be covered more than once, at most $c$ times for some $c \geq 1$.
Therefore, we distinguish between two variants
of the problem: one where each clique can be selected at most once,
and another without restriction.
The $\multicliquepackprob{c,d}$ problem asks the maximum size of a $\multicliquepacking{c}{d}$ in $G$.

\defproblemoutput{$\multicliquepackprob{c,d}$} 
{A graph $G$}
{The size of a largest $\multicliquepacking{c}{d}$ in $G$.}

Similar to the $\multicliquepackprob{c,d}$ problem, 
the $\singlecliquepackprob{c,d}$ problem asks the maximum size of a $\singlecliquepacking{c}{d}$ in $G$.

\defproblemoutput{$\singlecliquepackprob{c,d}$} 
{A graph $G$}
{The size of a largest $\singlecliquepacking{c}{d}$ in $G$.}

In \cref{section:algo_clique} we show that $\singlecliquepackprob{c,d}$ admits a single-exponential time
algorithm where the same algorithmic result also applies to $\multicliquepackprob{c,d}$
with slight modifications.

\begin{restatable}{theorem}{singlecliquepackingalgo}\label{theorem:single_clique_packing_algo}
	Let $c \geq 1$ and $d \geq 3$ be integers. Then, $\singlecliquepackprob{c,d}$
	can be solved in time $(c+1)^{w} \cdot n^{\mathcal{O}\left(1\right)}$
	for all $n$-vertex graphs $G$ given together with a tree decomposition of width
	at most $w$.
\end{restatable}

Moreover, we define partitioning problems where $H$ is a clique.

\defproblem{$\multicliquepartprob{c,d}$}
{A graph $G$}
{Is there a $\multipartitioning{c}{K_d}$ of $G$?}

\defproblem{$\singlecliquepartprob{c,d}$}
{A graph $G$}
{Is there a $\singlepartitioning{c}{K_d}$ of $G$?}

Similarly, we prove that this running time is optimal to polynomial factors in
the size of the input graph.
\begin{theorem}\label{theorem:multi_clique_packing_lower_bound}
	Let $c \geq 1$ and $d \geq 3$ be integers. If there exists an $\varepsilon > 0$ such that
	$\multicliquepartprob{c,d}$ can be solved in time $(c+1 - \varepsilon)^{w} \cdot n^{\mathcal{O}\left(1\right)}$
	for all $n$-vertex graphs $G$ given together with a path decomposition of width
	at most $w$, then the \ppseth\ fails.
\end{theorem}
Finally, we show that the lower bound result in \cref{theorem:single_clique_packing_lower_bound}
can also be transferred similarly.
\begin{theorem}\label{theorem:single_clique_packing_lower_bound}
	Let $c \geq 1$ and $d \geq 3$ be integers. If there exists an $\varepsilon > 0$ such that
	$\singlecliquepartprob{c,d}$ can be solved in time $(c+1 - \varepsilon)^{w} \cdot n^{\mathcal{O}\left(1\right)}$
	for all $n$-vertex graphs $G$ given together with a path decomposition of width
	at most $w$, then the \ppseth\ fails.
\end{theorem}

%% file: overview.tex
In this section, we will give an overview of the techniques and ideas presented in the paper.

\subsection{Preliminaries}\label{sec:prelim}
\input{preliminaries}

\subsection{Gadgets}
In this paper, we derive our lower bounds via reductions from well‐studied
base problems whose intractability is established under standard complexity
hypotheses. Each reduction makes use of compact,
purpose‐built gadgets — small graphs whose behavior can be precisely engineered.
Embedding these gadgets into our constructions yields intuitive, transparent
proofs that highlight the underlying ideas without excessive technical
overhead.

More formally, we define a gadget $G$ as a graph with designated portal vertices $\{p_1,
\ldots, p_{\ell}\} \subseteq V(G)$, where the vertices $\internal{G} \coloneqq
\left( G \setminus \{p_1, \ldots, p_{\ell}\} \right) $ are called internal
vertices. We say that a graph $E$ is an extension of a gadget $G$ if $E$
contains $G$ as an induced subgraph, where each internal vertex $v$ of $G$
satisfies
\begin{equation*}
	N_E(v) \subseteq V(G).
\end{equation*}
In other words, in an extension $E$,
only the portal vertices of $G$ can have neighbors outside of $G$.
\begin{figure}[htpb]
	\centering
	\includegraphics[scale=1]{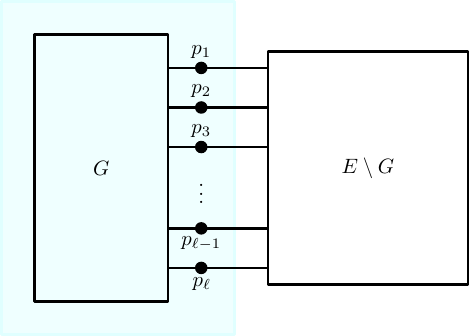}
	\caption{Illustration of an extension of gadget $G$ with $\ell$ portal
		vertices: the blue region on the left highlights $G$ along with its
		portals, while the rectangle on the right depicts the remainder
		$E \setminus G$.}
	\label{fig:gadget_intro}
\end{figure}

We also require gadgets to behave in a structured way. Whenever a copy of $H$
intersects the gadget, it must lie entirely within the gadget when connected
to a larger graph. We formalize this in the following definition.

\begin{definition}
	Let $G$ be a gadget and $P_G$ be the set of its portal vertices.
	For a graph $H$ and $c \geq 1$,
	we say that $G$ is $\coherent{(c,H)}$ if,
	for any extension $E$ of $G$ and any $\multipartitioning{c}{H}$ / $\singlepartitioning{c}{H}$
	$\mathcal{Z}$ of $E$, the following holds:
	If there exists $Z \in \mathcal{Z}$ that contains an internal
	vertex of $G$, then all vertices in $Z$ must belong to $G$. 
	Formally,
	\begin{equation*}
		\Bigl(V(Z) \cap \bigl( V(G) \setminus P_G \bigr) \Bigr) \neq \emptyset \implies V(Z) \subseteq V(G).
        \end{equation*}
\end{definition}

Next, we define the relation realized by a gadget.
\begin{definition}\label{definition:gadget_realizing}
	Let $H$ be a graph, $c \geq 1$ and 
	$G$ be a gadget with portal vertices $P_G = \{p_1, \ldots, p_\ell\}$. 
	We say $G$ \dist{c,H}-realizes (respectively, \arb{c,H}-realizes) a relation $R \subseteq
	\{0,\ldots,c\}^{\ell}$ if the following holds:
	\begin{equation*}
		r \in R \iff
		\parbox[t][][t]{12cm}{There exists a $\singlepacking{c}{H}$ (respectively, $\multipacking{c}{H}$) 
			 $\mathcal{Z}$ of $G$
			such that
			$\mathcal{Z}$ covers each
			internal vertex of
			$G$ exactly $c$ times, and each portal
			vertex $p_i$ exactly $r_i$
		times for $1 \leq i \leq \ell$.} 
	\end{equation*}
	We say that $G$ \strict{c,H}-realizes $R$ if it both \dist{c,H}
	and \arb{c,H}-realizes $R$.
\end{definition}

\begin{remark}
	When $c = 1$, the notions of \dist{c,H}-realization and \arb{c,H}-realization
	coincide. In this case, we simply say that the relation $R$ is $H$-realized
	by $G$, instead of
	using the term ``\strict{1,H}''.
\end{remark}
We say a relation $R \subseteq \{0,\ldots,c\} ^{\ell}$
is $\regular{(x,d)}$ for $d \geq 1$ and  $0 \leq x \leq d - 1$ if for each $r \in R$, the weight of $r$ is equivalent to $x \mod d$, i.e. 
$\w(r) \coloneqq \left( \sum_{i \in [\ell]} r_i \right)  =  x \mod d$.
Moreover, a relation $R$ has weight $X$ if $\Bigl( \max_{r \in R} \w(r) \Bigr)  = X$.
Recall that a gadget is a small, engineered graph used to enforce specific behaviors in a reduction.
We streamline the lower bound constructions by building general‑purpose gadgets in \cref{section:gadgets} that can realize any relation.
This way, we avoid repetitive constructions and keep the focus on the core ideas.
\begin{restatable}{lemma}{gadgetarbanyrelation}\label{lemma:gadget_arb_any_relation}
	Let $H$ be an arbitrary graph, $\ell \geq 1$ be an integer and $R
	\subseteq \{0, 1\}^{\ell}$ be a relation that is
	$\regular{(x, \abs{H})}$ for some $0 \leq x \leq \abs{H} -1$. Then,
	there exists a $\coherent{(1,H)}$ gadget $G$ that $H$-realizes the
	relation $R$ and the size of $G$ is bounded by some function of
	$\ell$. Moreover, for relations with constant weight,
	it holds that $\pw(G) = \mathcal{O}\left(\ell\right)$.
\end{restatable}
When $c \geq 1$, we require the relation to have more structure, i.e., the relation
should be $\regular{(0,\abs{H})}$.
However, this restriction can be easily handled in our lower bound constructions.
\begin{restatable}{lemma}{gadgetcliqueregrelation}\label{lemma:gadget_clique_reg_relation}
	Let $c \geq 1$, $H$ be a clique, $\ell \geq 1$ be a constant and $R \subseteq \{0, \ldots, c\}^{\ell}$ be a $\regular{(0,\abs{H})}$ relation.
	Then, there exists a $\coherent{(c,H)}$ gadget $G$ that \strict{c,H}-realizes the relation $R$.
	Moreover, the size of $G$ is bounded by some function of $\ell$.			
\end{restatable}

\subsection{Single-Exponential Lower Bound}
In this section, we give an overview of how the above-described gadgets
are used to prove \cref{theorem:single_clique_packing_lower_bound}.
Classically, when proving conditional lower bounds based on \seth,
one needs to find a reduction from \sat. However, this usually involves
carrying out repetitive, unnecessary work that is not specific to the problem
one is working on. 
One of the strengths of the framework introduced by Lampis \cite{lampisPrimalPathwidthSETH2025}
is removing the need for such repetitive constructions.

The primal graph of a \csp\ instance $\psi$ is a graph $G$ that has a vertex for each variable
of $\psi$, and there is an edge between $x,y \in V(G)$ if $x$ and $y$ appear together
in a constraint. The pathwidth of the \csp\ instance $\psi$ is defined to be the pathwidth of its primal graph.
Similarly, by path decomposition of $\psi$, we mean a path decomposition of the primal
graph of $\psi$. The following lemma from \cite{lampisPrimalPathwidthSETH2025}
implies that we can assume the path decomposition to be nice.

\begin{lemma}[Lemma~2.1 in \cite{lampisPrimalPathwidthSETH2025}, restated for CSPs]
	There is a linear-time algorithm that takes as input a \csp\ formula $\psi$
	with $n$ variables and $m$ constraints and a path decomposition of its primal
	graph of width $p$ and outputs a nice path decomposition $B_1, \ldots, B_t$ of
	$\psi$ containing at most $t = \mathcal{O}\left(p \cdot m\right)$ bags,
	as well as an injective function $b$ from the set of constraints of $\psi$
	to $[t]$ such that for each constraint $c$, $B_{b(c)}$ contains
	all the variables of $c$.
\end{lemma}
The following conjecture from \cite{lampisPrimalPathwidthSETH2025}, called \ppseth,
will form the basis of our hardness results.

\begin{conjecture}[Conjecture 1.1 from \cite{lampisPrimalPathwidthSETH2025}.]
	\label{conjecture:pwseth}
	For all $\eps>0$ we have the following: there exists no algorithm which
	takes as input a $3-\textsc{SAT}$ instance $\phi$ on $n$ variables and a path
	decomposition of its primal graph of width $\pw$ and correctly decides
	if $\phi$ is satisfiable in time $(2-\eps)^{\pw}n^{O(1)}$.
\end{conjecture}

Moreover, we have the following result from \cite{lampisPrimalPathwidthSETH2025},
which proves the equivalence of falsifying \ppseth\ and finding a faster algorithm
for \textsc{2-CSP}.

\begin{theorem}[Theorem 3.2 from \cite{lampisPrimalPathwidthSETH2025}, shortened]
\label{theorem:pwseth_equiv}
For each $B\ge 3$ the following statements are equivalent:

\begin{enumerate}
\item The \ppseth\ is false.

\item There exist $\eps>0, b>0$ and an algorithm that takes as input a
\kcsp{2} instance $\psi$ on alphabet $[B]$, together with a path decomposition
of $\psi$, and decides if $\psi$ is
satisfiable in time $(B-\eps)^{\pw} \cdot \abs{\psi}^{b}$.
\end{enumerate}
\end{theorem}
\begin{figure}[htpb]
	\centering
	\includegraphics[scale=0.9]{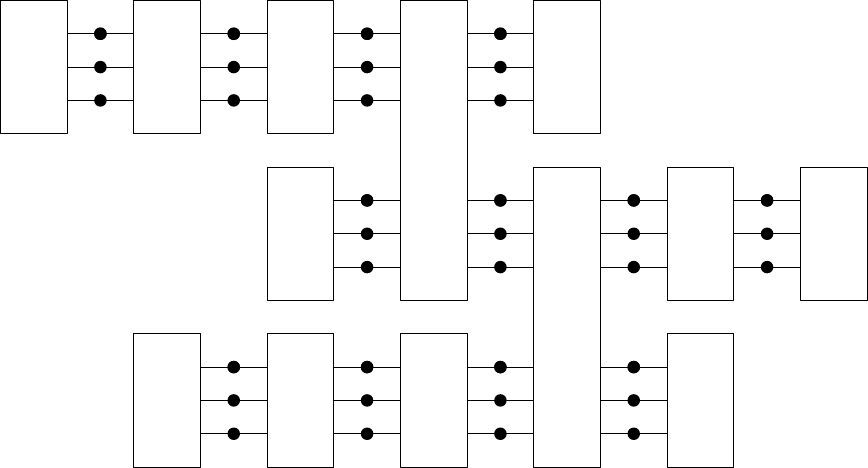}
	\caption{A high-level description of the lower bound construction used
		to prove \cref{theorem:multi_clique_packing_lower_bound}. The
		coverage counts of all vertices in a column represent an
		assignment to the variables of the \kcsp{2} instance. The
		rectangles depict the gadgets that enforce the constraints;
		observe that they interact locally with the vertices, which
	allows us to bound the pathwidth of the constructed instance.}
	\label{fig:single_exponential_lb}
\end{figure}
Motivated by \cref{theorem:pwseth_equiv}, \cref{section:lb_clique} presents a
reduction from the \kcsp{2} problem to the \singlecliquepartprob{c,d} problem.
The core idea is to encode each variable’s $B$ possible assignments using a
set of vertices whose coverage count in a $\singlecliquepacking{c}{d}$ represents
the chosen value. We then introduce gadgets each enforcing a single constraint of the \kcsp{2} instance via carefully specified
relations so that only coverage patterns corresponding to satisfying
assignments are possible. By designing these gadgets to interact locally, we
ensure the pathwidth of the resulting instance increases by at most a constant,
completing the reduction.

\subsection{Slightly Superexponential Lower Bound}
We now give a high-level description of the proof of \cref{theorem:graph_packing_arbitrary_lower_bound}.
Let $H$ be a graph with at least $3$ vertices that is not a block graph.
Next, we will describe the lower bound for the $\graphpartprob{H}$ problem
which is stated in \cref{theorem:graph_packing_arbitrary_lower_bound}.
Recall that $\graphpartprob{H}$ is solvable in slightly super-exponential
time for any fixed graph $H$,
i.e., there is an algorithm with running time $\mathcal{O}^*\left(2^{O(tw \cdot \log tw)}\right)$.\footnote{The $O^*(f(k))$ notation suppresses polynomial factors in the input size $n$; that is, $O^*(f(k)) = O(f(k) \cdot n^c)$ for some constant $c$.}
We show that this running time is optimal under \ethh, i.e. there exists no algorithm for $\graphpartprob{H}$ with running time
$\mathcal{O}^*\left(2^{o(tw \cdot \log(tw))}\right)$.
To that end, we use an auxiliary problem
for which such a lower bound was presented in \cite{cyganParameterizedAlgorithms2015}.
Specifically, we define the $\kkindset$ problem where given a graph $G$
on a vertex set $[k] \times [k]$, we ask whether there exists an independent set
$X$ in $G$ that contains exactly one vertex from each row and each column.

\defproblem{$\kkindset$}
{A graph $G$ on the vertex set $[k] \times [k]$}
{Is there an independent set $X$ of $G$ such that $X$ induces a permutation on $[k] \times [k]$?}

The following hardness result was presented for the analogous $\kkclique$ problem
in \cite{cyganParameterizedAlgorithms2015} which translates easily to $\kkindset$.
Below we state it for $\kkindset$.

\begin{theorem}[Theorem 14.14 in \cite{cyganParameterizedAlgorithms2015}]\label{theorem:kkindset_lb_result}
	Unless \ethh\ fails, $\kkindset$ cannot be solved in time $2^{o( k \cdot \log k)}$.
\end{theorem}

To prove \cref{theorem:graph_packing_arbitrary_lower_bound}, we reduce the
$\kkindset$ problem to $\graphpartprob{H}$. The construction relies on
structural properties of $H$. Let $B$ be a block in $H$ with a minimum
separator $S$ such that $B \setminus S$ has at least two connected components.
Since $H$ contains a non-clique block, such $B$ and $S$ exist.
We partition $S$ into two subsets, $U$ and $D$, and create $k$ copies of each.
The construction in the proof of \cref{theorem:graph_packing_arbitrary_lower_bound}
ensures that any $\packing{\abs{H}}$ includes $k$ copies of
$H$, each covering exactly one copy of $U$ and one of $D$. This yields $k! = 2^{\mathcal{O}\left(k \cdot \log k\right)}$
configurations, corresponding to all permutations of a set of size $k$.
Additionally, for each edge $e$ in the $\kkindset$ instance, we introduce a
gadget to prevent $e$ from being included in the set of vertices
induced by this permutation. Moreover, these gadgets are designed to preserve
pathwidth by interacting only locally.

\subsection{Algorithmic Results}
Recall that a tree decomposition of a graph $G = (V, E)$ is a pair
$(T,\{X_t\}_{t \in T})$, where $T$ is a tree whose every node $t$ is assigned a
bag $X_t \subseteq V(G)$, satisfying
the following properties:
\begin{enumerate}
	\item For every vertex $v \in V$, there exists at least one bag $X_t$
		such that $v \in X_t$.
	\item For every edge $(u, v) \in E$, there exists a bag $X_t$ such that
		$\{u, v\} \subseteq X_t$.
	\item For every vertex $v \in V$, the set of nodes $\{t \in V(T) \mid v
		\in X_t\}$ induces a connected subtree of $T$.
\end{enumerate}
The \emph{width} of a tree decomposition is the size of its largest bag minus
one, and the \emph{treewidth} of $G$ is the minimum width over all possible
tree decompositions of $G$.
A nice tree decomposition is a rooted tree decomposition 
in which each node is one of the four types: leaf, introduce, forget or join.
We refer the reader to \cite{cyganParameterizedAlgorithms2015} for more details.

The algorithmic results in
\cref{theorem:graph_packing_arbitrary_algo,theorem:single_clique_packing_algo}
follow a fairly standard dynamic programming framework over tree
decompositions. We define a suitable set of states for each bag in the
decomposition, capturing the essential information needed to extend partial
solutions. The main challenge of the approach lies in carefully designing these states
and proving the correctness of the update rules that determine how states
transition from one bag to the next. Usually what determines the running time is
the state representation and
transitions to the specific structural constraints of our problem.

In particular, the single-exponential algorithm for the
$\singlecliquepackprob{c,d}$ problem defines the state of a bag based on how
many times each vertex in the bag is covered by a partial solution. This yields
$(c+1)^{\ell}$ states for a bag of size $\ell$, naturally leading to a running
time of $(c+1)^{tw} \cdot n^{\mathcal{O}\left(1\right)}$. 

In contrast, the slightly superexponential algorithm for the
$\graphpackprob{H}$ problem requires keeping track of a partition of the bag, in which
each part corresponds to a partial copy of $H$. This results in
$\mathcal{O}(\ell^{\ell})$ possible states for a bag of size $\ell$, leading to an overall
running time of $2^{\mathcal{O}(tw \cdot \log tw)} \cdot n^{\mathcal{O}(1)}$.

%% file: preliminaries.tex
A graph $H$ is $k$-connected if for each $A \subseteq V(H)$ such that $\abs{A} = k-1$
it holds that $H \setminus A$ is connected. A graph is also called biconnected if it is
$2$-connected, and a block is a maximal $2$-connected component of $H$.
Similarly, a graph $H$ is called a block graph if every block of $H$ is
a clique.

A vertex $v$ of a connected graph $H$ is a cutvertex if $G \setminus v$
is disconnected. In this paper, we refer to both the blocks of a graph and the nodes of the
corresponding block tree interchangeably, with a slight abuse of notation.
While formally distinct, we find it convenient to treat them as equivalent
entities for the sake of clarity and brevity.
For a vertex $h \in V(H)$ and a block $B \in \blocks$ , we write $h \in B$ if $h \in V(B)$.

\begin{definition}
	For a function $f$ and a set $X$, we let $\funcrest{f}{X}$ denote the restriction of $f$ to $X \cap \dom{f}$. Similarly, we let $\funcrem{f}{X}$ denote the restriction of $f$ to $\dom{f} \setminus X$. Finally, for $v \not\in \dom{f}$ and a value $y$, we let $\funcadd{f}{v}{y}$ denote a function $g$ which is defined as
	\begin{align*}
		g(x) = \begin{cases}
			f(x) &\text{if }  x \in \dom{f}\\
			y &\text{if } x = v 
		\end{cases}.
	\end{align*}
\end{definition}

We use the Iverson bracket $\iverson{P}$, which is defined
to be $1$ if the proposition $P$ is true and $0$ otherwise.

%% file: gadgets.tex
In this paper, we consider specific combinations of graphs $H$
and integer values for $c$.
Specifically, define the sets
\begin{align*}
	\mathcal{C}_1 &\coloneqq \Big\{ (s,K_d) \,\Big|\, s \geq 2, d \geq 3 \Big\},\\
	\mathcal{C}_2 &\coloneqq \Bigl\{(1,H) \,\Big|\, H \text{ is an arbitrary graph on at least $3$ vertices} \Bigr\}
\end{align*}
where $K_d$ is the clique on $d$ vertices.
Next, we describe a general approach for proving that a gadget strictly
realizes a given relation.

\begin{remark}\label{remark:strict_real_proof_guide}
	Let $c \geq 1$ be an integer, $H$ be a graph. Let $G$ be a gadget with portal vertices
	$p_1, \ldots, p_{\ell}$ and $R$ be a relation of arity $\ell$.
	Note that since any $\singlepacking{c}{H}$ of $G$ is, by definition, also a
	$\multipacking{c}{H}$ of $G$, to show that the gadget $G$ \strict{c,H}-realizes $R$,
	it is enough to prove the following two statements:
	\begin{enumerate}
		\item Given a $\multipacking{c}{H}$ $\mathcal{Z}$ of $G$ that covers each
	internal vertex exactly $c$ times and each portal vertex $p_i$ exactly $r_i$ times,
	it holds that $r \coloneqq (r_1, \ldots, r_{\ell}) \in R$.
		\item For any $r \in R$, there exists a $\singlepacking{c}{H}$
			$\mathcal{Z}$ of $G$ such that all internal vertices
			of $G$ are covered exactly $c$ times, and each portal vertex
			$p_i$ is covered $r_i$ times by $\mathcal{Z}$.
	\end{enumerate}
	These two statements together imply that $G$
	both \dist{c,H}- and \arb{c,H}-realizes $R$, i.e., $G$
	\strict{c,H}-realizes $R$.
\end{remark}

Given a relation $R$, we call a gadget $G$ that realizes $R$
an $R$-gadget, if $c$ and $H$ are clear from the context.
Let $G$ have $\ell$ portal vertices $p_1, \ldots, p_\ell$, and let $E$ be a graph.
Attaching $G$ to vertices $v_1, \ldots., v_\ell$ in $E$ means choosing an arbitrary
bijection $\alpha$ and identifying each $p_i$ with $v_{\alpha(i)}$ for $i \in [\ell]$.
If we want to specify the order, then we say we attach
$(p_1, \ldots, p_\ell)$ to $(v_1, \ldots, v_\ell)$ in $E$.
Now let $\ell \geq 1$ and define the relations
\begin{align*}
	\coverrel{\ell}{y} &\coloneqq \Bigl\{ \bigl(y,\ldots,y \bigr) \Bigr\} \subseteq \mathbb{Z}_{\geq 0}^{\ell} \text{ where } y \in \mathbb{Z}_{\geq 0},\\
	\eqrel{\,\ell}{X} &\coloneqq \bigcup_{x \in X} \coverrel{\ell}{x}, \text{ where } X \subseteq \mathbb{Z}_{\geq 0},\\
	\sumrel{\,\ell}{z} &\coloneqq \biggl\{ \bigl(x_1, \ldots, x_\ell \bigr) \biggm| 0 \leq x_i \leq z \text{ for } 1 \leq i \leq \ell, \,\sum_{i = 1}^{\ell} x_i = z\biggr\}, \text{ where } z \geq 1\\
\end{align*}
and
\begin{equation*}
	\cneqrel{z} \coloneqq \bigl\{(x,z-x) \mid x \in \{0, \ldots, z\} \bigr\} = \sumrel{\,2}{z} \text{ for } z \geq 1. 		
\end{equation*}
In \cite{kirkpatrickComplexityGeneralGraph2006}, a $\cneqrel{1}$-gadget
and an $\eqrel{\abs{H}}{\{0,1\}}$-gadget are presented for an arbitrary graph $H$.
Moreover, it is shown that
both are $\coherent{(1,H)}$.

Recall that we say a relation $R \subseteq \{0,\ldots,c\} ^{\ell}$
is $\regular{d}$ if for each $r \in R$, the weight of $r$ is divisible by $d$, i.e. 
$\w(r) \coloneqq \left( \sum_{i \in [\ell]} r_i \right)  =  0 \mod d$.
Moreover, given an integer $c \geq 1$ and a vector $s \in \{0, \ldots, c\}^{\ell}$,
we define
\begin{equation*}
	\Compl{c}{s} \coloneqq \bigl( c - s_1, \ldots, c-s_{\ell} \bigr).
\end{equation*}
Similarly, for a relation $R \subseteq \{0, \ldots, c\}^{\ell}$, we let
\begin{equation*}
	\Compl{c}{R} \coloneqq \Bigl\{ \Compl{c}{r} \Bigm| r \in R\Bigr\}.
\end{equation*}
For two tuples $s = \left( s_1, \ldots, s_{\ell_1} \right) $ and
$r = \left( r_1, \ldots, r_{\ell_2} \right) $ we define a new tuple that
concatenates them, i.e.
\begin{equation*}
	s \odot r = \left(  s_1, \ldots, s_{\ell_1},  r_1, \ldots, r_{\ell_2} \right).
\end{equation*}

Finally, given a graph $G$ and a vertex $v \in V(G)$, we define the operation of
replacing $v$ with $t$ copies. We create a graph $G'$ such that $V(G') = \Bigl(V(G) \setminus \{v\} \Bigr) \cup \{v_1, \ldots, v_t\}$
where
\begin{equation*}
	N_{G'}(v_1) = N_{G'}(v_2) = \ldots = N_{G'}(v_t) = N_{G}(v).
\end{equation*}

\subsection{Gadgets constructions}
In this section, we present gadgets that
realize relations with distinct copies of $H$, which will be used in
\cref{section:lb_clique}.
Next, we present a $\cneqrel{c}$-gadget.
\begin{lemma}\label{lemma:single_partition_neq}
	Let $(c,H) \in \left( \mathcal{C}_1 \cup \mathcal{C}_2 \right) $,
	then there exists a $\coherent{(c,H)}$ gadget $G$ 
	with $2$ portal vertices that \strict{c,H}-realizes the relation
	$\cneqrel{c}$. Furthermore,
	the size of $G$ is bounded by a constant.
\end{lemma}

\begin{figure}[htpb]
	\centering
	\includegraphics[scale=1]{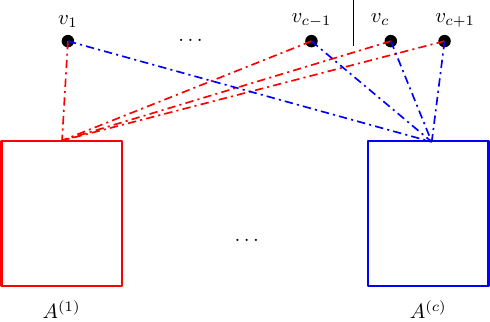}
	\caption{The gadget described in \cref{lemma:single_partition_neq}. Each $A^{(i)}$
	is a copy of $H$ with one vertex removed. The dash-dotted lines between $A^{(i)}$ and $v_{j}$ for
	$i \in [c]$ and $j \in [c+1]$ imply that $v_j$ is connected to each vertex
	of $A^{(i)}$. Red and blue colors are used to distinguish connections corresponding to different $A^{(i)}$'s.}
	\label{fig:single_neqrel_two}
\end{figure}

\begin{proof}
	As mentioned before, for $(c,H) \in \mathcal{C}_2$ there exists
	a $\coherent{(1,c)}$ gadget in \cite{kirkpatrickComplexityGeneralGraph2006}
	that $\abs{H}$-realizes the relation $\cneqrel{1}$.
	Therefore, we focus on the $(c,H) \in \mathcal{C}_1$ case and
	continue with the construction of the gadget.

	\subparagraph{Construction of the gadget.}		
	We introduce $c$ copies of $H$, called $A_1, \ldots,
	A_c$. Then, for each $i \in [c]$, we replace a vertex $v \in A_i$ with
	$c+1$ copies of $v$, which are denoted by $v^{i}_j$ for $j \in [c+1]$.
	Then, for each $j \in [c+1]$, we identify the vertices $\{v^{i}_j\}_{i
	\in [c]}$ and call the new vertex $v_j$. Finally, we designate $v_c$
	and $v_{c+1}$ as the portal vertices of the gadget we call $G$. This is
	the whole construction.

	\subparagraph{The $c$-coherence of $G$.}
	Let $E$ be an extension of $G$ and $\mathcal{Z}$ be a
	$\singlepartitioning{c}{H}$ / $\multipartitioning{c}{H}$
	of $E$.	Suppose for a contradiction that there exists $Z \in \mathcal{Z}$
	that contains an internal vertex $u$ of $G$, together with a vertex
	$x \in E \setminus G$. In $Z$, there should be two vertex disjoint paths
	from $u$ to $x$, because otherwise $Z$ would have a cutvertex, which is a
	contradiction since $Z$ is isomorphic to a clique. 
	Therefore, $Z$ should contain both $v_c$ and $v_{c+1}$, the two
	portal vertices of $G$.
	However, this leads to a contradiction,
	since there are no edges between $v_c$ and $v_{c+1}$ by construction.

	\subparagraph{The relation realized by $G$.}	
	Following \cref{remark:strict_real_proof_guide}, let
	$\mathcal{Z}$ be a $\multipacking{c}{H}$ of $G$ that covers each
	internal vertex exactly $c$ times and the portal vertices $a$ and $b$ times.
	We can partition $\mathcal{Z}$ into two: copies of $H$ that contain
	a portal vertex and those that do not. Let $\left( \mathcal{Z}_1,\mathcal{Z}_2 \right) $
	denote this partition of $\mathcal{Z}$. By the arguments above in the $c$-coherence section,
	each copy
	of $H$ in $\mathcal{Z}_1$ should contain exactly one portal vertex and all vertices of $A_i$
	for some $i \in [c]$. Similarly, each copy of $\mathcal{Z}_2$ should contain
	$v_i$ for some $i \in [c-1]$ and vertices of $A_j$ for some $j \in [c]$.
	Observe that the size of $\mathcal{Z}$ should be $c^{2}$ so that each
	internal vertex of $G$ is covered exactly $c$ times. On the other hand,
	the size of $\mathcal{Z}_2$ is $c \cdot (c-1)$ which implies that the size
	of $\mathcal{Z}_1$
	is equal to $c$. Since each copy of $H$ in $\mathcal{Z}_1$ covers exactly
	one portal vertex, $\mathcal{Z}$ covers $v_c$ and $v_{c+1}$ exactly $c$
	times in total, and we have $a + b = c$.

	On the other hand, for each $(x, c-x) \in \cneqrel{c}$, it is
	straightforward to construct a $\singlepacking{c}{H}$ of $H$ in $E$ that
	covers each internal vertex of $G$ exactly $c$ times, and the portal
	vertices $v_c$ and $v_{c+1}$ a total of $c = a + b$ times.
	Specifically, initialize $\mathcal{K} = \emptyset$. For each $i \in
	[c]$ and $j \in [c-1]$, add to $\mathcal{K}$ a copy of $H$ covering the
	vertices of $A_i$ and $v_j$. Then, 
	for $i \in [x]$, add to $\mathcal{K}$ a copy of $H$ covering $A_i$ and
	$v_c$. For the remaining $c-x$ sets (i.e., for $i \in \{x+1, \ldots,
	c\}$), add a copy covering $A_i$ and $v_{c+1}$.
	This construction yields a $\singlepacking{c}{H}$ $\mathcal{K}$ that
	covers each internal vertex of $G$ exactly $c$ times, and the portal
	vertices $x$ and $c-x$ times, respectively.
	
	\subparagraph{The size of $G$.}
	The size bound follows from the fact that we introduce $c$ copies of
	$H$ together with $c+1$ many portal vertices.
\end{proof}

Next, we make the following observation: by introducing $\coherent{(c,H)}$ gadgets that \strict{c,H}-realize the relation
$\cneqrel{c}$, we can assume that
without loss of generality all the gadgets introduced in this section are
$\coherent{(c,H)}$.
 
\begin{observation}\label{obs:multi_coherence}
	Let $G$ be a gadget with portal vertices $p_1, \ldots, p_{\ell}$
	that $(c,H)$-realizes some relation $R$.
	Then, for each $i \in [\ell]$, we introduce a
	$\coherent{(c,H)}$ gadget $N_i$ with portal vertices $a_i$ and $b_i$
	that \strict{c,H}-realizes the relation $\cneqrel{c}$.
	Then, we identify
	$a_i$ and $p_i$ for each $i \in [\ell]$. Let $G'$ denote the
	gadget constructed in this manner with portal vertices $b_1, \ldots, b_{\ell}$.
	Then, $G'$ $(c,H)$-realizes the relation $R$,
	and since each $N_i$ is $\coherent{(c,H)}$, $G'$ is also $\coherent{(c,H)}$.
\end{observation}

\begin{figure}[htpb]
	\centering
	\includegraphics{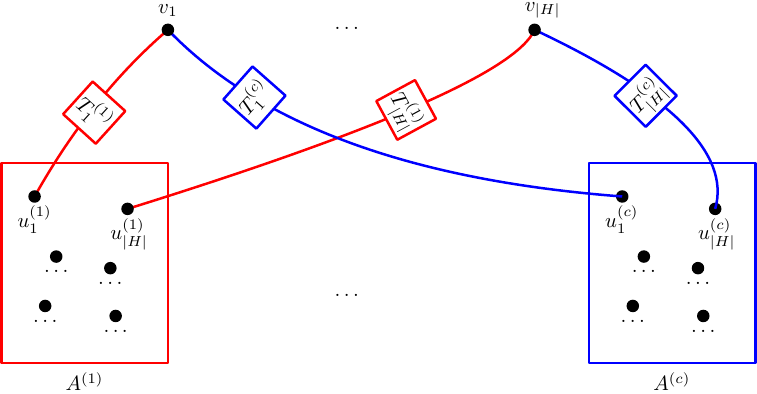}
	\caption{The gadget described in \cref{lemma:single_partition_eqrel_abs_H}. Each $A^{(i)}$
	is a copy of $H$ and each $T^{(i)}_j$ is a $\cneqrel{c}$ gadget from \cref{lemma:single_partition_neq}.
	The red and blue colors are used to distinguish between gadgets connected to each $A^{(i)}$.} 
	\label{fig:single_eqrel_H}
\end{figure}
\begin{lemma}\label{lemma:single_partition_eqrel_abs_H}
	Let $(c,H) \in \left( \mathcal{C}_1 \cup \mathcal{C}_2 \right) $,
	then there exists a $\coherent{(c,H)}$ gadget $G$ 
	with $\abs{H}$ portal vertices that \strict{c,H}-realizes the relation
	$\eqrel{\abs{H}}{[0,c]}$.
	Furthermore, 
	the size of $G$ is bounded by a constant.
\end{lemma}

\begin{proof}
	We start with the construction of the gadget.
	\subparagraph{Construction of the gadget.}
	For each $i \in [c]$, let 
	let $A^{(i)}$ be a copy of $H$. Then, for each vertex $u^{(i)}_j$ of $A^{(i)}$ for $j \in [\abs{H}]$,
	we create a $\cneqrel{c}$ gadget $T^{(i)}_j$ from \cref{lemma:single_partition_neq},
	and identify one of its portal vertices with $u^{(i)}_j$.
	Let $v^{(i)}_j$ denote the other portal vertex of $T^{(i)}_j$.
	Finally, for each $j \in [\abs{H}]$, we identify the vertices $\{v^{(1)}_j, \ldots, v^{(c)}_j\}$
	and call the resulting vertex $v_j$.
	We call the resulting gadget $G$ and let $\{v_1, \ldots, v_{\abs{H}}\} $ be the portal vertices of $G$.
	This is the whole construction.

	\subparagraph{The relation realized by $G$.}
	To show that $G$ \strict{c,H}-realizes the relation $\eqrel{\abs{H}}{[0,c]}$, we
	will follow \cref{remark:strict_real_proof_guide}.
	Suppose that $G$ has a $\multipacking{c}{H}$ $\mathcal{K}$ that covers
	all internal vertices of $G$ exactly $c$ times and each portal vertex
	$v_j$ exactly $r_j$ times for $j \in [\abs{H}]$ and for some $r_j \in \{0, \ldots, c\}$.
	Let $i \in [c]$ and $x_i$ be the number of copies of $A^{(i)}$ in $\mathcal{K}$.
	For each $j \in [\abs{H}]$,
	the gadget $T^{(i)}_j$ is $\coherent{(c,H)}$,
	therefore $\mathcal{K}$ contains a $\multipacking{c}{H}$ $\mathcal{T}^{(i)}_j
	\subseteq \mathcal{K}$ of $T^{(i)}_j$ that covers $u^{(i)}_j$ exactly
	$c - x_i$ times, so that $u^{(i)}_j$ is covered $c$ times in total. Then,
	$\mathcal{T}^{(i)}_j$ covers $v_j$ exactly $x_i$ times for each $j \in [\abs{H}]$.
	Since this holds for each $i \in [c]$, 
	it holds that each $v_j$ is covered $\alpha \coloneqq \sum_{i \in [c]} x_i$ times in
	total. Hence we get $0 \leq \alpha \leq c$ and moreover $r_j = \alpha$
	for each $j \in [\abs{{H}}]$. Therefore, we have $r \coloneqq (r_1,
	\ldots, r_{\abs{H}}) \in \eqrel{\abs{H}}{[0,c]}$.

	Now let $\bm{v} = (x, \ldots, x) \in \eqrel{\abs{H}}{[0,c]}$
	for $x \in \{0, \ldots, c\}$.
	We will now construct a $\singlepacking{c}{H}$ $\mathcal{K}$ of $G$,
	that covers each internal vertex of $G$ exactly $c$ times, and each
	portal vertex $x$ times.
	For $1 \leq i \leq x$, we add to $\mathcal{K}$ a copy of
	$A^{(i)}$. Note that since $\mathcal{K}$ is a $\singlepacking{c}{H}$,
	$\mathcal{K}$ cannot contain another copy of $A^{(i)}$ for $i \in [x]$.
	Then, for each $i \in [x]$ and $j \in [\abs{H}]$, we pick a $\singlepacking{c}{H}$ of $T^{(i)}_j$
	that covers $u^{(i)}_j$ and $v_{j}$ exactly $c-1$ and $1$ time, respectively.
	Finally, for each $x+1 \leq i \leq c$ and $j \in [\abs{H}]$, we add to $\mathcal{K}$
	a $\singlepacking{c}{H}$ of $T^{(i)}_j$ that covers $u^{(i)}_j$ exactly $c$ times.
	All in all, $\mathcal{K}$ satisfies the required properties.

	\subparagraph{The size of $G$.}
	Observe that $G$ is constructed from $c$ copies of $H$ and
	$c \cdot \abs{H}$ many copies of $\cneqrel{c}$-gadget,
	whose size depends on $c$ and $\abs{H}$.
\end{proof}

We continue with an
$\eqrel{k \cdot \abs{H}}{[0,c]}$-gadget.
In this case, using the appropriate gadgets introduced previously,
the same construction results in a gadget that realizes the relation with
arbitrary and distict copies of $H$.

\begin{figure}[htpb]
	\centering
	\includegraphics[scale=0.7]{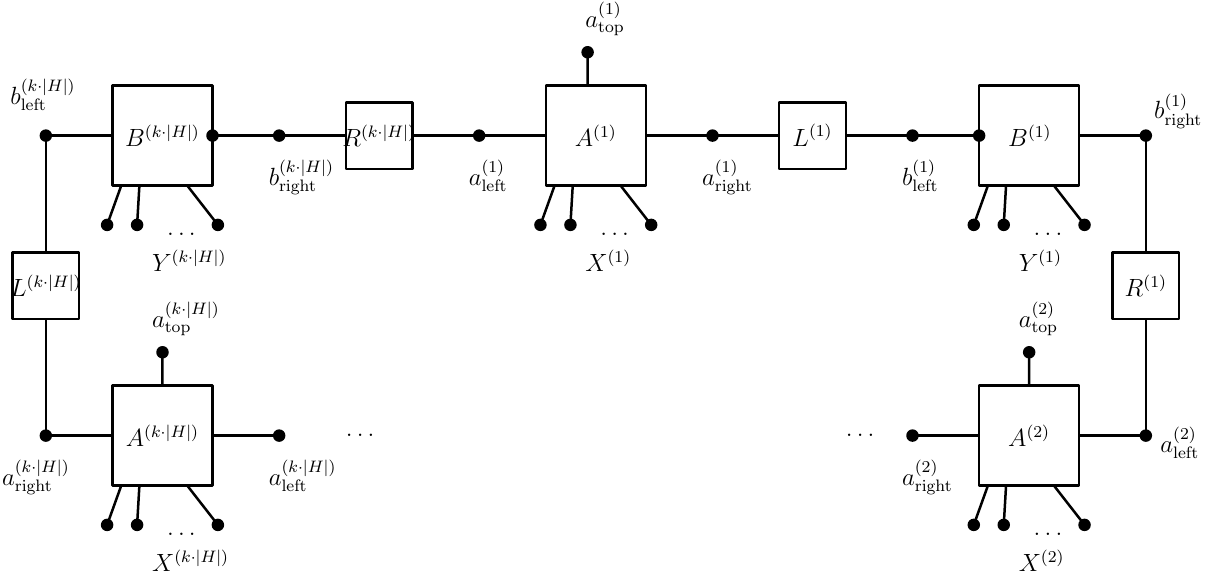}
	\caption{The gadget described in \cref{lemma:multi_partition_eqrel}. The equality gadgets
	that are attached to $X^{(i)}$ and $Y^{(i)}$ are omitted in the figure for the sake of presentation.}
	\label{fig:multi_eqrel_k_cdot_H}
\end{figure}

\begin{lemma}\label{lemma:multi_partition_eqrel}
	Let $(c,H) \in \bigl( \mathcal{C}_1 \cup \mathcal{C}_2 \bigr)$
	and $k \geq 1$ be a constant. Moreover, let $R$ denote the relation $\eqrel{k \cdot \abs{H}}{[0,c]}$.
	Then, there exists a $\coherent{(c,H)}$ gadget $G$ 
	with $k \cdot \abs{H}$ portal vertices that \strict{c,H}-realizes $R$.
	Moreover, the size of $G$ is bounded by some function of $k$.
\end{lemma}

\begin{proof}
	We start with the construction of the gadget.
	
	\subparagraph{Construction of the gadget.}		
	For each $1 \leq i \leq k \cdot \abs{H}$, we create two gadgets, $A^{(i)}$ and $B^{(i)}$,
	which \strict{c,H}-realize the relation $\eqrel{\abs{H}}{[0,c]}$.
	Note that such gadgets exist by \cref{lemma:single_partition_eqrel_abs_H}.
	Moreover, let $P^{(i)}_A$ and
	$P^{(i)}_B$ denote the portal vertices of $A^{(i)}$ and $B^{(i)}$ such that
	\begin{align*}
		P^{(i)}_A &= \{\aleft{i}, \aright{i}, \aup{i}\} \cup X^{(i)} \\
		P^{(i)}_B &= \{\bleft{i}, \bright{i}\} \cup Y^{(i)}.
	\end{align*}
	where $Y^{(i)} \coloneqq \Bigl\{ y^{(i)}_{1}, \ldots, y^{(i)}_{\abs{H} -
	2}\Bigr\}$ and $X^{(i)} \coloneqq \Bigl\{ x^{(i)}_{1}, \ldots, x^{(i)}_{\abs{H} - 3}\Bigr\}$.
	Note that $X^{(i)} = \emptyset$ if $\abs{H} = 3$.
        
	Then, for each $1 \leq i \leq k \cdot \abs{H}$, we create two copies of
	the $\cneqrel{c}$ gadget from \cref{lemma:single_partition_neq}, which
	we call $L^{(i)}$ and $R^{(i)}$. We identify the portal vertices of $L^{(i)}$
	with $\aright{i}$ and $\bleft{i}$, and similarly, identify the portal vertices
	of $R^{(i)}$ with $\bright{i}$ and $\aleft{i + 1 \mod k \cdot \abs{H}}$.

	Then, for $1 \leq j \leq \abs{H} -
	3$, we partition the set of vertices $\Bigl\{ x^{(i)}_j \mid 1 \leq i \leq k \cdot
	\abs{H}\Bigr\}$ into consecutive groups of size $\abs{H}$. For each set $S$ in the
	partition, we introduce an $\eqrel{\abs{H}}{[0,c]}$-gadget and identify its
	portal vertices with $S$.
	Similarly, for $1 \leq j \leq \abs{H} - 2$, we partition the set of vertices
	$\Bigl\{ y^{(i)}_j \mid 1 \leq i \leq k \cdot \abs{H}\Bigr\}$ into consecutive groups
	of size $\abs{H}$, introduce an $\eqrel{\abs{H}}{[0,c]}$-gadget for each set
	$S$ in the partition and identify its portal vertices with $S$.
	We call the gadgets introduces in this step filler gadgets.

	The portal vertices of the whole gadget is the set $\Bigl\{ \aup{i}
	\mid 1 \leq i \leq k \cdot \abs{H}\Bigr\}$. This concludes the
	construction of the gadget $G$.

	\subparagraph{The relation realized by $G$.}
	To show that $G$ \strict{c,H}-realizes $\eqrel{k \cdot
	\abs{H}}{[0,c]}$, we will follow \cref{remark:strict_real_proof_guide}.
	Suppose there exists a $\multipacking{c}{H}$ $\mathcal{K}$ such that
	each internal vertex of $G$ is covered $c$ times. Observe that the
	gadget $A^{(1)}$ is $\coherent{(c,H)}$ and realizes the relation
	$\eqrel{\abs{H}}{[0,c]}$, hence there exists a $\multipacking{c}{H}$
	$\mathcal{A}^{(1)} \subseteq \mathcal{K}$ of $A^{(1)}$ such that covers
	each of its portal vertices $x$ times for some $x \in [0,c]$.
	Similarly, since the gadget $L^{(1)}$ is $\coherent{(c,H)}$ and
	realizes the relation $\cneqrel{c}$, there exists a
	$\multipacking{c}{H}$ $\mathcal{L}^{(1)} \subseteq \mathcal{K}$ that
	covers $\aright{1}$ and $\bleft{1}$ exactly $c$ times in total. Now
	observe that $\aright{1}$ can only be covered by $\mathcal{A}^{1}$ and
	$\mathcal{L}^{(1)}$, which implies that $\mathcal{L}^{(1)}$ cover
	$\aright{1}$ exactly $c-x$ times. This argument can be iteratively
	applied to demonstrate that each $\aup{i}$ is covered exactly $x$ times
	for some $x \in [0,c]$.
	
	Conversely, let $\bm{v} = (x, \ldots, x) \in \eqrel{k \cdot \abs{H}}{[0,c]}$
	for $x \in \{0, \ldots, c\}$.
	We will now construct a $\singlepacking{c}{H}$ $\mathcal{K}$ of $G$,
	that covers each internal vertex of $G$ exactly $c$ times, and each
	portal vertex $x$ times.	
	For each $1 \leq i \leq k \cdot \abs{H}$, we add to $\mathcal{K}$ a
	$\singlepacking{c}{H}$ of $A^{(i)}$ that covers each of its portal
	vertices $x$ times.
	Similarly, we also add a $\singlepacking{c}{H}$ of $B^{(i)}$ that covers its
	portal vertices exactly $c - x$ times.
	Then, we add to $\mathcal{K}$ a $\singlepacking{c}{H}$ of $L^{(i)}$
	that covers $\aright{i}$ and $\bleft{i}$ exactly $c-x$ and $x$ times,
	respectively. Note that such a $\singlepacking{c}{H}$ exists because $L^{(i)}$
	\strict{c,H}-realizes the relation $\cneqrel{c}$.
	Next, we do the same with $R^{(i)}$.
	Finally, for each filler gadget, we add to $\mathcal{K}$ a $\singlepacking{c}{H}$
	of it that covers its portal vertices either $c-x$ and $x$ times,
	depending on whether they are attached to $A^{(i)}$ or $B^{(i)}$ respectively,
	for some $1 \leq i \leq k \cdot \abs{H}$.
	Observe that each portal vertex $\aup{i}$ is covered $x$ times, whereas
	all the internal vertices are covered exactly $c$ times.

	\subparagraph{The size of $G$.}
	Finally, to bound the size of $G$, observe that
	$G$ consists of $\mathcal{O}\left(k \cdot \abs{H}\right)$ many copies of
	$\eqrel{\abs{H}}{[0,c]}$ and $\cneqrel{c}$ gadgets, whose size
	depend on $\abs{H}$. Since $\abs{H}$ is constant, the size claim follows.
\end{proof}

Next, we provide a $\coverrel{k \cdot \abs{H}}{c-1}$-gadget.
\begin{figure}[htpb]
	\centering
	\includegraphics[scale=0.8]{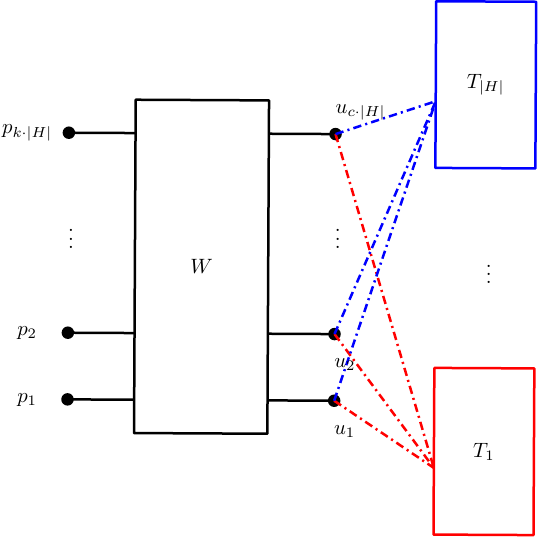}
	\caption{The gadget described in \cref{lemma:multi_partition_coverrel_c_H}.
	The dash-dotted lines between $T_i$ and $u_j$ indicate that $u_j$ is connected to the neighbourhoof of $v$ in $T_i$.}
	\label{fig:multi_single_cover_c-1}
\end{figure}

\begin{lemma}\label{lemma:multi_partition_coverrel_c_H}
	Let $(c,H) \in \mathcal{C}_{1}$ and $k \geq 1$ be an integer.
	Then, there exists a gadget $G$ 
	with $k \cdot \abs{H}$ portal vertices that \strict{c,H}-realizes the relation
	$\coverrel{k \cdot \abs{H}}{c-1}$.
	Moreover, the size of $G$ is bounded by some function of $k$.	
\end{lemma}

\begin{proof}
	We start with the construction of $G$.
	
	\subparagraph{Construction of the gadget.}		
	Let $v$ be an arbitrary vertex of $H$ and
	$T$ be a copy of $H$ where we replace $v$ with $c \cdot \abs{H}$
	copies.
	Next, we create $\abs{H}$ many copies
	of $T$, which are denoted by $T_1, \ldots, T_{\abs{H}}$.
	Let $v^{i}_j$ be the $j$'th copy of $v$ in $T_i$.
	Then for each $1 \leq j \leq c \cdot \abs{H}$, we identify the
	vertices $v^{i}_j$ for $1 \leq i \leq \abs{H}$, and call the new
	vertex $u_j$.
	Finally, we introduce $k \cdot \abs{H}$ many vertices
	$p_1, \ldots, p_{k \cdot \abs{H}}$ and introduce an gadget
	$W$ that \strict{c,H}-realizes the relation $\eqrel{(k + c) \cdot \abs{H}}{[0,c] }$,
	which exists by \cref{lemma:multi_partition_eqrel}.
	Then, we identify its portal vertices with $\{p_1, \ldots, p_{k \cdot \abs{H}}\} \cup \{u_1, \ldots, u_{c \cdot \abs{H}}\} $.
	We let $G$ denote the gadget constructed in this manner,
	and let $p_1, \ldots, p_{k \cdot \abs{H}}$ be the portal vertices of $G$.

	\subparagraph{The relation realized by $G$.}	
	To show that $G$ \arb{c,H}-realizes $\coverrel{k \cdot \abs{H}}{c-1}$,
	we will follow \cref{remark:strict_real_proof_guide}.
	Suppose that there exists a $\multipacking{c}{H}$ $\mathcal{K}$ of $G$
	such that each internal
	vertex of $G$ is covered $c$ times. 
	Since $W$ is $\coherent{(c,H)}$ and $(c,H)$-realizes the relation
	$\eqrel{(k + c) \cdot \abs{H}}{[0,c] }$, there exists a $\multipacking{c}{H}$
	$\mathcal{W} \subseteq \mathcal{K}$ of $W$
	that covers each of its portal vertices $x$ times for some $x \in [0,c]$.
	Hence, each $u_j$ needs to be covered $c-x$ times in total by the copies
	of $H$ arising from $\multipacking{c}{H}$ of each $T_i$.
	By a double counting argument, one can prove that $x$ is equal to $c-1$,
	which shows that $\mathcal{K}$ covers each portal vertex exactly $c-1$ times.

	Now let $\bm{v} = (c-1, \ldots, c-1) \in \coverrel{k \cdot \abs{H}}{c-1}$.
	We will construct a $\singlepacking{c}{H}$ $\mathcal{K}$ of $G$,
	that covers its portal vertices according to $\bm{v}$.
	Since $W$ \strict{c,H}-realizes $\eqrel{k \cdot \abs{H}}{[0,c]}$,
	there exists a $\singlepacking{c}{H}$ $\mathcal{W}$
	of $W$ such that $\mathcal{W}$ covers each internal vertex of $W$ exactly $c$
	times, and covers each vertex in $\{p_1, \ldots, p_{k \cdot \abs{H}}\} \cup \{u_1, \ldots, u_{c \cdot \abs{H}}\} $
	exactly $c-1$ times. Then, for each $i \in \abs{H}$ and $j \in c \cdot \abs{H}$,
	we add to $\mathcal{K}$ a $\singlepacking{c}{H}$ of $T_i$ such that
	in total each $u_j$ is covered once for $j \in [c \cdot \abs{H}]$.
	Hence, all internal vertices are covered exactly $c$ times,
	where the portal vertices are covered $c-1$ times.
	
	\subparagraph{The size of $G$.}	
	Observe that $G$ is constructed from $k \cdot \abs{H}$ many portal vertices,
	$\abs{H}$ many
	copies of $H$, together with $c \cdot \abs{H}$ many copies of
	$\neqgadget$-gadget. Since $c$ and $\abs{H}$ are constants,
	the size claim follows.
\end{proof}

Recall that we say a relation $R \subseteq \{0,\ldots,c\} ^{\ell}$
is $\regular{d}$ if the weight of each tuple in $R$ is equal to
$0 \mod d$.
Let $(c,H) \in \left( \mathcal{C}_1 \cup \mathcal{C}_2 \right)$ and
$R \subseteq \{0, \ldots, c\}^{\ell}$ be an $\regular{\abs{H}}$ relation.
We will now introduce gadgets that \strict{c,H}-realizes $R$.
Next, we introduce a gadget that covers either the first or the second half
of its portal vertices, which we call a toggle gadget.

\begin{figure}[htpb]
	\centering
	\includegraphics[scale=0.8]{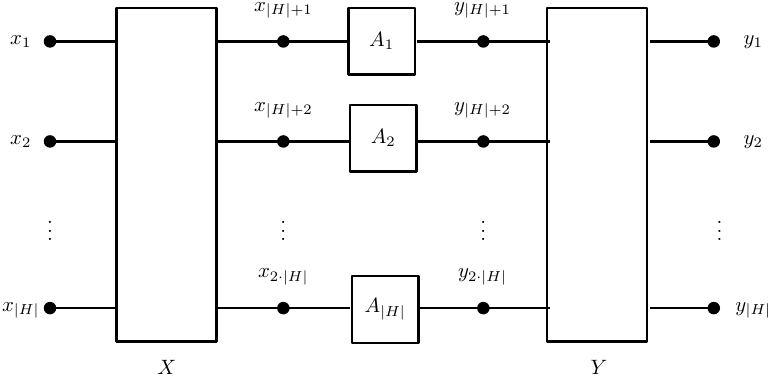}
	\caption{The gadget described in \cref{lemma:gadget_half_zero}.}
	\label{fig:selector_gadget}
\end{figure}

\begin{lemma}[Toggle gadget]\label{lemma:gadget_half_zero}
	Let $H$ be an arbitrary graph.
	Define the relation
	\begin{equation*}
		R = \Big\{\underbrace{\left( 0,\ldots,0,1,\ldots,1 \right) }_{\text{$\abs{H}$ zeros followed by $\abs{H}$ ones}} ,\quad \underbrace{\left( 1,\ldots,1,0,\ldots,0 \right) }_{\text{$\abs{H}$ ones followed by $\abs{H}$ zeros}}\Big\} \subseteq \{0,1\}^{2 \cdot \abs{H}}.
	\end{equation*}
        Then there exists a gadget $G$ that $H$-realizes $R$.
	Moreover the size of $G$ is constant.
\end{lemma}

\begin{proof}
	We start with the construction of the gadget.
	\subparagraph{Construction of the gadget.}
	Let $X,Y$ be
	gadgets $H$-realizing the relation $\eqrel{2 \cdot \abs{H}}{[0,1]
	}$ which exist by \cref{lemma:multi_partition_eqrel}. 
	Let the sets $\{x_1, \ldots, x_{2 \cdot \abs{H}}\}$ and $\{y_1, \ldots, y_{2 \cdot \abs{H}}\}$  be the
	portal vertices of $X$ and $Y$ respectively. 
	Observe that $(1,H) \in \mathcal{C}_2$, therefore by \cref{lemma:single_partition_neq} there exist a gadget
	that $H$-realizes the relation $\cneqrel{1}$.
	For each $1 \leq i \leq \abs{H}$, we introduce a
	$\cneqrel{1}$-gadget $A_i$ and
	identify its portal vertices with $x_{i + \abs{H}}$ and $y_{i + \abs{H}}$.
	The set of portal vertices of the gadget is
	\begin{equation*}
		\{x_1, \ldots, x_{\abs{H}}, y_{1}, \ldots, y_{\abs{H}}\}.
	\end{equation*}
	We call the the whole gadget $G$. This is the whole construction.

	\subparagraph{The relation realized by $G$.}
	To show that $G$ $(1,H)$-realizes the relation $R$,
	we will follow \cref{definition:gadget_realizing}.
	Let $\bm{v} = (0, \ldots, 0, 1, \ldots, 1) \in R$.
	We will now show that there exists a $\singlepacking{1}{H}$ $\mathcal{K}$ of
	$G$ that covers the portal vertices of $G$ according to $\bm{v}$.
	For $\bm{v} = (1, \ldots, 1, 0, \ldots, 0) \in R$, the same arguments
	work by the symmetry of the gadget.

	We start by adding a $\singlepacking{1}{H}$ of $X$ to $\mathcal{K}$, where
	each $x_i$ is covered once for $1 \leq i \leq 2 \cdot \abs{H}$.
	Then, for each $1 \leq i \leq \abs{H}$, we add to $\mathcal{K}$ a $\singlepacking{1}{H}$ of $A_i$ which covers $y_i$ for $1 \leq i \leq \abs{H}$.
	Finally, we add to $\mathcal{K}$ a $\singlepacking{1}{H}$ of $Y$ that does
	not cover any of its portal vertices. All in all, $\mathcal{K}$ is a
	$\singlepacking{1}{H}$ of $G$ that covers each internal vertex and $\{x_i\}_{1 \leq i \leq \abs{H}}$ once,
	whereas $\{y_i\}_{1 \leq i \leq \abs{H}}$ are not covered.

	Now suppose that there exists a $\multipacking{c}{H}$ $\mathcal{K}$ of $G$
	such that each internal
	vertex of $G$ is covered $c$ times.
	Since $X$ and $Y$ are both $\eqrel{2 \cdot \abs{H}}{[0,1]
	}$-gadgets, this implies that there exists $\mathcal{K}_1, \mathcal{K}_2 \subseteq \mathcal{K}$ such that they cover the portal vertices of $X$ and $Y$ according to $\eqrel{2 \cdot \abs{H}}{[0,1]}$. However, exactly one set of the vertices from $\{x_i\}_{1 \leq i \leq \abs{H}}$ and $\{y_i\}_{1 \leq i \leq \abs{H}}$ should be covered because of the $\cneqrel{1}$ gadgets $A_1, \ldots, A_{\abs{H}}$.
	Hence, the portal vertices are covered according to $R$.

	\subparagraph{The size of $G$.}
	Finally, observe that the size of $G$ depends only on $\abs{H}$ which is a constant.
\end{proof}

\begin{figure}[htpb]
	\centering
	\includegraphics[scale=0.8]{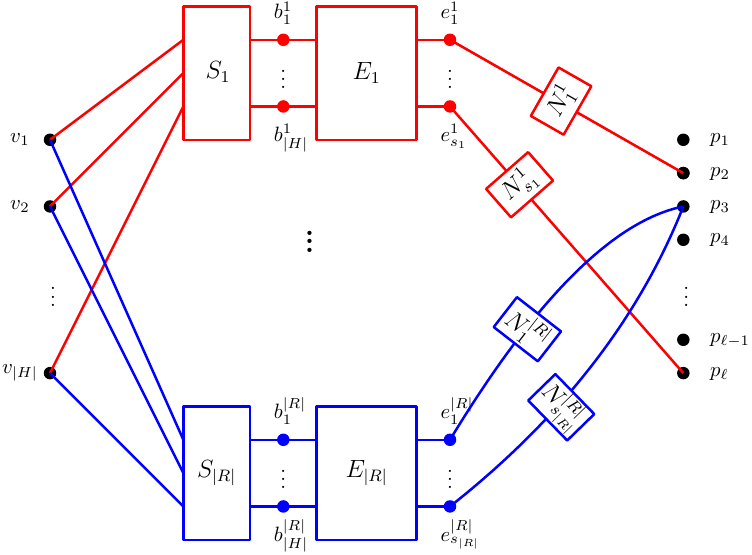}
	\caption{The gadget described in \cref{lemma:gadget_arb_any_relation} for $\psi = 0$.}
	\label{fig:arb_relation}
\end{figure}

Using the previously introduced gadgets, we now prove \cref{lemma:gadget_arb_any_relation}
and construct gadgets that realize
arbitrary relations. We begin with the case $(c, H) \in \mathcal{C}_2$, i.e.,
$c = 1$ and $H$ is an arbitrary graph.

\gadgetarbanyrelation*

\begin{proof}
	Let $0 \leq \psi \leq \abs{H} - 1$ such that
	$(x + \psi) = 0 \mod \abs{H}$.
	We start with the construction of the gadget.
	\subparagraph{Construction of the gadget.}		
	
	Introduce $\abs{H}$ many cental vertices
	$v_1, \ldots, v_{\abs{H}}$. We then introduce the portal vertices $\{p_1, \ldots, p_{\ell}\}$
	and the slack vertices $z_1, \ldots, z_{\psi}$.
	Then, for each $r^{i} \in R$, we do the following:
	\begin{enumerate}
		\item Introduce a toggle gadget $S_i$ from \cref{lemma:gadget_half_zero}
			with portal vertices
			$a^{i}_1, \ldots, a^{i}_{2 \cdot \abs{H}}$,
			and identify the vertices $v_1, \ldots, v_{\abs{H}}$ with $a^{i}_1, \ldots, a^{i}_{\abs{H}}$.

		\item Let $s_i \coloneqq w(r^{i})$, i.e. the weight of $r^{i}$
			which satisfies $s_i = x \mod \abs{H}$.
			Using \cref{lemma:multi_partition_eqrel}, create a gadget $E_i$ that $1$-realizes the relation $\eqrel{s_i + \psi + \abs{H}}{\{0,1\} }$
			with the portal vertices $\{b^{i}_1, \ldots, b^{i}_{\abs{H}}, e^{i}_1, \ldots, e^{i}_{s_i + \psi}\}$.
			Note that the gadget is well-defined since $(s_i + \psi) = 0 \mod \abs{H}$ by the definition of $R$.
			Identify the portal vertices $b^{i}_1, \ldots, b^{i}_{\abs{H}}$ with
			$a^{i}_{\abs{H} + 1}, \ldots, a^{i}_{2 \cdot \abs{H}}$.

		\item  Then, for each $r^{i} \in R$, fix a function $\alpha_i$ from $[s_i]$ to $\ell$
			such that for each $j \in \ell$ we have $\abs{\alpha_i^{-1}(j)} = r^{i}_j$.
			Then, for each $t \in [s_i]$, we introduce a $\cneqrel{1}$ gadget $N^{i}_t$,
			and identify its portal vertices with $e^{i}_t$ and $p_{\alpha_i(t)}$.
			Then, for each $u \in [\psi]$, we introduce a $\cneqrel{1}$ gadget $\tilde{N}^{i}_u$
			and identify its portal vertices with $e^{i}_{s_i + u}$ and $z_{u}$.
	\end{enumerate}
	This is the whole construction of the gadget $G$.

	\subparagraph{The relation realized by $G$.}
	To show that $G$ $H$-realizes the relation $R$,
	we will follow \cref{definition:gadget_realizing}.
	Let $r^i \in R$.
	We will now show that there exists a $\packing{H}$ $\mathcal{K}$ of
	$G$ that covers the portal vertices of $G$ according to $r^i$.

	We consider the gadget $S_i$ corresponding to $r^i$, and add to $\mathcal{K}$
	a $\packing{H}$ of $S_i$ that covers $v_1, \ldots, v_{\abs{H}}$ once.
	Then, we add to $\mathcal{K}$ a $\packing{H}$ of $E_i$ that covers
	all its portal vertices once.
	Then, for each $t \in [s_i]$, we add to $\mathcal{K}$ a $\packing{H}$
	that covers $p_{\alpha_i(t)}$.
	Note that each portal vertex $p_j$ is covered $\abs{\alpha^{-1}_{i}(j)} = r^{i}_j$ times.
	For all $j \in [\abs{R}] \setminus \{i\}$, we add to $\mathcal{K}$ a $\packing{H}$ of
	$S_j$ that covers $b^{j}_1, \ldots, b^{j}_{\abs{H}}$ once. Then, we add
	to $\mathcal{K}$ a $\packing{H}$ of $E_j$ that doesn't cover
	any of its portal vertices. Finally,
	for each $t \in [s_j]$, we add to $\mathcal{K}$ a $\packing{\abs{H}}$
	of $N^{j}_t$ that covers $e^{j}_{t}$.
	All in all, each internal vertex is covered $c$ times, and each portal vertex
	is covered corresponding to $r^{i}$.

	Now suppose that there exists a $\packing{H}$ $\mathcal{K}$ of $G$
	such that each internal
	vertex of $G$ is covered $c$ times.
	Recall that if $c \geq 2$, then in the construction of $G$ we have
	added a gadget $F$ that \arb{c,H}-realizes the relation $\coverrel{\abs{Z}}{c-1}$.
	In that case, since $F$ is $\coherent{(c,H)}$, there exists $\mathcal{F} \subseteq \mathcal{K}$
	that covers each vertex in $Z$ exactly $c-1$ times.
	Observe that $v_1, \ldots, v_{\abs{H}}$ need to be covered exactly once
	by $\mathcal{K} \setminus \mathcal{F}$.
	Therefore, there exists $i \in [\abs{R}]$ such that $\mathcal{K}$ contains
	a $\packing{\abs{H}}$ of $\mathcal{S}_i$ of $S_i$ such that
	$\mathcal{S}_i$ covers $v_1, \ldots, v_{\abs{H}}$.
	This holds because $S_i$ is $\coherent{(1,H)}$.
	By a similar logic, $\mathcal{K}$ also contains a $\packing{\abs{H}}$ of $E_i$
	that covers all its portal vertices, because otherwise $b^{i}_1, \ldots, b^{i}_{\abs{H}}$
	wouldn't be covered $c$ times.
	Finally, for each $t \in [s_j]$, since $N^{i}_t$ is $\coherent{(1,H)}$,
	$\mathcal{K}$ contains a $\packing{\abs{H}}$ of $N^{i}_t$ that
	covers the portal vertex $p_{\alpha_i(t)}$ once.
	It is easy to see that this is the only copies of $H$ that can cover portal vertices,
	hence the portal vertices are covered according to $r^{i}$.
	
	\subparagraph{The size and pathwidth of $G$.}	
	Finally, to prove the size claim, first observe that $R$ can
	have at most $(c+1)^{\ell}$ many tuples,
	and for each tuple in $R$ we add gadgets of constant size.
	Since the number of central vertices and portal vertices are equal to $\ell$ and $\abs{H}$,
	respectively,
	it holds that all the size of $G$ is a function of $\ell$.
	Moreover, observe that if we remove the portal vertices $p_1, \ldots, p_{\ell}$
	and the central vertices $v_1, \ldots, v_{\abs{H}}$,
	the pathwidth can decrease by at most $\ell + \abs{H}$.
	However, the remaining graph has constant pathwidth since the weight of the
	relation is constant. Therefore, the pathwidth of $G$ is at most $\mathcal{O}\left(\ell\right)$.
\end{proof}

Next, we prove \cref{lemma:gadget_clique_reg_relation} by constructing a gadget
that can \strict{c,H}-realize any $\regular{(0,\abs{H})}$ relation $R$ for $(c,H) \in \mathcal{C}_1$.

\gadgetcliqueregrelation*
\begin{proof}
	The proof—and the corresponding gadget construction—closely follows
	that of \cref{lemma:gadget_arb_any_relation}, but we include all
	details here for completeness. In particular, for $c = 1$,
	the result follows from \cref{lemma:gadget_arb_any_relation},
	therefore we assume $c \geq 2$ in the following.

	Define $0 \leq \gamma < \abs{H}$ such that
	$(\abs{R} + \gamma) \bmod \abs{H} = 0$. To simplify the construction,
	we define a new relation $R'$ by adding $\gamma$ duplicate copies of an
	arbitrary tuple from $R$, so that $\abs{R'} = \abs{R} + \gamma$ becomes
	divisible by $\abs{H}$.
	We emphasize that this modification is merely a notational
	convenience—it abstracts the need to repeat certain steps (specifically
	Steps 1-3 below) $\gamma$ additional times,
	without modifying the relation.
	Therefore, in the following, without loss of generality we may assume that
	$\abs{R} = 0 \mod \abs{H}$.
	
	\subparagraph{Construction of the gadget.}		
	
	Introduce $\abs{H}$ many cental vertices
	$v_1, \ldots, v_{\abs{H}}$. We then introduce the portal vertices $\{p_1, \ldots, p_{\ell}\}$. 
	Then, for each $r^{i} \in R$, we do the following:
	\begin{enumerate}
		\item Introduce a toggle gadget $S_i$ from \cref{lemma:gadget_half_zero}
			with portal vertices
			$a^{i}_1, \ldots, a^{i}_{2 \cdot \abs{H}}$,
			and identify the vertices $v_1, \ldots, v_{\abs{H}}$ with $a^{i}_1, \ldots, a^{i}_{\abs{H}}$.

		\item Let $s_i \coloneqq w(r^{i})$, i.e. the weight of $r^{i}$.
			Using \cref{lemma:multi_partition_eqrel}, create a gadget $E_i$ that $1$-realizes the relation $\eqrel{s_i + \abs{H}}{\{0,1\} }$
			with the portal vertices $\{b^{i}_1, \ldots, b^{i}_{\abs{H}}, e^{i}_1, \ldots, e^{i}_{s_i}\} $.
			Note that the gadget is well-defined since $s_i = 0 \mod \abs{H}$ by the definition of $R$.
			Identify the portal vertices $b^{i}_1, \ldots, b^{i}_{\abs{H}}$ with
			$a^{i}_{\abs{H} + 1}, \ldots, a^{i}_{2 \cdot \abs{H}}$.

		\item  Then, for each $r^{i} \in R$, fix a function $\alpha_i$ from $[s_i]$ to $\ell$
			such that for each $j \in \ell$ we have $\abs{\alpha_i^{-1}(j)} = r^{i}_j$.
			Then, for each $t \in [s_i]$, we introduce a $\cneqrel{1}$ gadget $N^{i}_t$,
			and identify its portal vertices with $e^{i}_t$ and $p_{\alpha_i(t)}$.
	\end{enumerate}
	Finally, let $G$ denote gadget constructed so far and $Z$ denote its internal vertices.
			Observe that $\abs{Z}$ is a multiple of $\abs{R}$,
			which itself is a multiple of $\abs{H}$.
			Therefore, we can assume that $\abs{Z} = 0 \mod H$ as well.
			Then, we introduce a gadget $F$ that \strict{c,H}-realizes the relation
			$\coverrel{\abs{Z}}{c-1}$, using \cref{lemma:multi_partition_coverrel_c_H}.
			Finally, we identify the portal vertices of $F$ with the vertices of $Z$.
	This is the whole construction of the gadget $G$.

	\subparagraph{The relation realized by $G$.}

	To show that $G$ \strict{c,H}-realizes the relation $R$,
	we will follow \cref{remark:strict_real_proof_guide}.
	Let $r^i \in R$.
	We will now show that there exists a $\singlepacking{c}{H}$ $\mathcal{K}$ of
	$G$ that covers the portal vertices of $G$ according to $r^i$.

	We start by adding a $\singlepacking{c}{H}$ of $F$ to $\mathcal{K}$ that covers
	each vertex in $Z$ exactly $c-1$ times.
	Then, we consider the gadget $S_i$ corresponding to $r^i$, and add to $\mathcal{K}$
	a $\singlepacking{1}{H}$ of $S_i$ that covers $v_1, \ldots, v_{\abs{H}}$ once.
	Then, we add to $\mathcal{K}$ a $\singlepacking{1}{H}$ of $E_i$ that covers
	all its portal vertices once.
	Then, for each $t \in [s_i]$, we add to $\mathcal{K}$ a $\singlepacking{1}{H}$
	that covers $p_{\alpha_i(t)}$.
	Note that each portal vertex $p_j$ is covered $\abs{\alpha^{-1}_{i}(j)} = r^{i}_j$ times.
	For all $j \in [\abs{R}] \setminus \{i\}$, we add to $\mathcal{K}$ a $\singlepacking{1}{H}$ of
	$S_j$ that covers $b^{j}_1, \ldots, b^{j}_{\abs{H}}$ once. Then, we add
	to $\mathcal{K}$ a $\singlepacking{1}{H}$ of $E_j$ that doesn't cover
	any of its portal vertices. Finally,
	for each $t \in [s_j]$, we add to $\mathcal{K}$ a $\singlepacking{1}{\abs{H}}$
	of $N^{j}_t$ that covers $e^{j}_{t}$.
	All in all, each internal vertex is covered $c$ times, and each portal vertex
	is covered corresponding to $r^{i}$.

	Now suppose that there exists a $\multipacking{c}{H}$ $\mathcal{K}$ of $G$
	such that each internal
	vertex of $G$ is covered $c$ times.
	Recall that if $c \geq 2$, then in the construction of $G$ we have
	added a gadget $F$ that \arb{c,H}-realizes the relation $\coverrel{\abs{Z}}{c-1}$.
	In that case, since $F$ is $\coherent{(c,H)}$, there exists $\mathcal{F} \subseteq \mathcal{K}$
	that covers each vertex in $Z$ exactly $c-1$ times.
	Observe that $v_1, \ldots, v_{\abs{H}}$ need to be covered exactly once
	by $\mathcal{K} \setminus \mathcal{F}$.
	Therefore, there exists $i \in [\abs{R}]$ such that $\mathcal{K}$ contains
	a $\multipacking{1}{\abs{H}}$ of $\mathcal{S}_i$ of $S_i$ such that
	$\mathcal{S}_i$ covers $v_1, \ldots, v_{\abs{H}}$.
	This holds because $S_i$ is $\coherent{(1,H)}$.
	By a similar logic, $\mathcal{K}$ also contains a $\multipacking{c}{\abs{H}}$ of $E_i$
	that covers all its portal vertices, because otherwise $b^{i}_1, \ldots, b^{i}_{\abs{H}}$
	wouldn't be covered $c$ times.
	Finally, for each $t \in [s_j]$, since $N^{i}_t$ is $\coherent{(1,H)}$,
	$\mathcal{K}$ contains a $\multipacking{1}{\abs{H}}$ of $N^{i}_t$ that
	covers the portal vertex $p_{\alpha_i(t)}$ once.
	It is easy to see that this is the only copies of $H$ that can cover portal vertices,
	hence the portal vertices are covered according to $r^{i}$.
	
	\subparagraph{The size of $G$.}	
	Finally, to prove the size claim, first observe that $R$ can
	have at most $\mathcal{O}\left((c+1)^{\ell}\right)$ many tuples,
	and for each tuple in $R$ we add gadgets of constant size.
	Since the number of central vertices and portal vertices are equal to $\ell$ and $\abs{H}$,
	respectively,
	it holds that all the size of $G$ is a function of $\ell$.
\end{proof}

%% file: lower_bounds_for_clique_packing_problems.tex
In this section we prove \cref{theorem:multi_clique_packing_lower_bound,theorem:single_clique_packing_lower_bound}.
\cref{theorem:pwseth_equiv} says that in order to prove a conditional lower bound
based on \ppseth, one can start the reduction from the \kcsp{2} problem.
In the following, we will present a reduction from the \kcsp{2} problem
to the $\multicliquepartprob{c,d}$. Subsequently, we will describe another reduction from
$\multicliquepartprob{c,d}$ to $\singlecliquepartprob{c,d}$ problem, and prove
\cref{theorem:multi_clique_packing_lower_bound,theorem:single_clique_packing_lower_bound}.

\subsection{Lower Bounds}
We first prove \cref{theorem:multi_clique_packing_lower_bound}.
To that end, intuitively, we show that a fast algorithm for the $\multicliquepartprob{c,d}$
problem implies a fast algorithm for the \kcsp{2} problem.
\begin{lemma}\label{lemma:multi_clique_hardness}
	Let $c \geq 1$ and $d \geq 3$ be integers. Suppose there exists an
	$\varepsilon > 0$ such that $\multicliquepartprob{c,d}$ can be solved
	in time $(c+1 - \varepsilon)^{\pw\left( G \right) } \cdot
	n^{\mathcal{O}\left(1\right)}$ for all $n$-vertex graphs $G$ given together
	with a path decomposition of width at most $\pw(G)$. Then,
	there exist $\eps^{\prime}, b^{\prime}>0$, an integer $B \geq 1$ and an
	algorithm that takes as input a \kcsp{2} instance $\psi$ on alphabet
	$[B]$, together with a path decomposition of $\psi$, and decides if $\psi$
	is satisfiable in time $(B-\eps^{\prime})^{\pw}
	\cdot \abs{\psi}^{b^{\prime}}$.
\end{lemma}

\begin{proof}
	Let $\varepsilon$, $b$ be as in the lemma statement and let $\mathcal{A}$ be the
	hypothetical algorithm for $\multicliquepartprob{c,d}$.
	Moreover, we let $H$ denote the clique $K_d$ and let $\ell$ be the
	smallest integer that is a multiple of $\abs{H}$ such that
	\begin{equation}\label{eq:ell_choice}
		\left( 1 - \varepsilon \right) \cdot \abs{H} < \frac{\varepsilon}{2} \cdot \left( l - \abs{H} \right).
	\end{equation}
	Observe that $\ell$ is a constant that only depends on $\varepsilon$ and $H$.
	Let $B = \left( c+1 \right)^{\ell - \abs{H}}$ and $\varepsilon' = \frac{\varepsilon}{2}$. We
	will now present a reduction from \textsc{2-CSP} with
	alphabet size $B$ to $\multicliquepartprob{c,d}$. The idea is as follows: in order to make use of
	regular relations, we will represent each variable of the
	\textsc{2-CSP} instance by $\ell$ vertices. Note that each of the
	$\ell$ vertices can be covered between $0$ and $c$ times, which can
	be thought of as the state of a vertex. In total, $\ell$ vertices
	combined give rise to $(c+1)^{\ell}$ states. These states can also be
	visualized as vectors $r \in \{0, \ldots, c\}^{\ell}$. Next, we
	consider the following subset of vectors
	\begin{equation*}
		Z \coloneqq \{x \in (c+1)^{\ell} \mid w(x) = 0 \mod \abs{H}\}. 
	\end{equation*}
	Observe that we have $\abs{Z} \geq (c+1)^{\ell - \abs{H}} = B$, because
	we can append to each vector in $r' \in \{0, \ldots, c\}^{\ell - \abs{H}}$
	at most $\abs{H}$ many $1$'s so that the new vector $r \in \{0, \ldots, c\}^{\ell}$
	constructed this way satisfies $w(r) = 0 \mod \abs{H}$.
	Hence there exists an injective map $\Phi \from [B]
	\to Z$ where we define $W \coloneqq \im\left( \Phi \right) $. We use the set $W$ to simulate $B$
	many assignments to a variable of the \textsc{2-CSP} instance.
	Observe that $W$, as a relation, is $\regular{(0,\abs{H})}$.

	\subparagraph{Construction of the $\multicliquepartprob{c,d}$ instance.}
	Let $\psi$ be a \textsc{2-CSP} instance with variables $x_1, \ldots,
	x_n$, contraints $C_1, \ldots, C_m$ and alphabet $[D]$. Let
	$\mathcal{P} = (B_1, \ldots, B_t)$ be a nice path decomposition of
	width $p$. Finally, let $b \from [m]
	\to [t]$ be a function that maps each constraint $C_i$ of $\psi$ to a
	bag such that $B_{b(i)}$ contains the variables occurring in $C_i$. We will
	construct an instance of the $\multicliquepartprob{c,d}$ problem as
	follows:

	\begin{enumerate}
		\item For each $1 \leq i \leq n$, define $l(i) \in [t]$ to be the
			smallest integer such that $x_i \in B_{l(i)}$.
			Similarly, let $r(i)$ be the largest integer such that
			$x_i \in B_{r(i)}$.
			For each $i \in [n]$ and $l(i) \leq j \leq r(i) + 1$,
			introduce $\ell$ vertices $\{a^{i,j}_{1}, \ldots,
			a^{i,j}_{\ell}\}.$
		\item Define the relation
			\begin{equation*}
				W^{C} \coloneqq \Compl{c}{W}.
			\end{equation*}
			Observe that $W$ is $\regular{(0,\abs{H})}$ by construction.
			The same also holds for $W^{C}$, because
			for each $x \in W^{C}$ such that $x = \Compl{c}{s}$ for $s \in W$,
			we have
			\begin{equation*}
				w(x) = \Bigl( \ell \cdot c - w(s) \Bigr)  = 0 \mod \abs{H},
			\end{equation*}
			where the last equivalence holds because
			$\ell$ and $w(s)$ are both equivalent to $0 \mod \abs{H}$.
			Introduce two gadgets $L^{i}$ and $R^{i}$ that \arb{c,H}-realize
			the relation $W^{C}$ and $W$, respectively, which exist by \cref{lemma:gadget_clique_reg_relation}.
			Then, identify the portal
			vertices of $L^{i}$ with $(a^{i,l(i)}_1, \ldots,
			a^{i,l(i)}_\ell)$, and similarly, identify the portal
			vertices of $R^{i}$ with the vertices
			$\{a^{i,r(i) + 1}_{1}, \ldots, a^{i, r(i)+1}_{\ell}\}$.
		\item   Let $j \in [t]$. We say that $j$ represents $s$ for $s \in [m]$ if $b(s) = j$.
			In that case, we
			define $S_j$ to be the relation, and $i_1$ and $i_2$ to
			be the indices of the variables associated with
			the constraint $C_s$.
			We define the new relation 
			\begin{equation*}
				R_j \coloneqq \bigg\{ \Phi(u_1) \odot \Phi(u_2) \odot \Compl{c}{\Phi(u_1)} \odot \Compl{c}{\Phi(u_2)}  \biggm| (u_1, u_2) \in S_j\bigg\}.
			\end{equation*}
			Observe that $R_j$ is $\regular{(0,\abs{H})}$ because for each
			$x \in R_j$ we have
			\begin{equation*}
				w(x) = \left( 2 \cdot \ell \cdot c \right) = 0 \mod \abs{H}
			\end{equation*}
			since $\ell$ is a multiple of $\abs{H}$.
			Hence, we can introduce a gadget $N^{j}_{i_1,i_2}$ by \cref{lemma:gadget_clique_reg_relation} that \arb{c,H}-realizes $R_j$.
			Then, we identify the portal vertices of $N^{j}_{i_1,i_2}$ with
			the vertices
			\begin{equation}\label{eq:N_j_portal_vertices}
				\left( a^{i_1,j}_{1}, \ldots, a^{i_1,j}_{\ell} \right) \odot \left( a^{i_2,j}_{1}, \ldots, a^{i_2,j}_{\ell} \right) \odot \left(  a^{i_1,j+1}_{1}, \ldots, a^{i_1,j+1}_{\ell} \right) \odot \left( a^{i_2,j+1}_{1}, \ldots, a^{i_2,j+1}_{\ell} \right).
			\end{equation}
			Next, we define the relation $\copyrel \subseteq \{0, \ldots, c\}^{2 \cdot \ell}$ where
			\begin{equation}\label{eq:copyrel_def}
				\copyrel \coloneqq \{\bm{u} \odot \Compl{c}{\bm{u}} \mid \bm{u} \in W \} 
			\end{equation}
			and let
			\begin{equation*}
				\Gamma_j \coloneqq \begin{cases}
					\{i_1, i_2\} &\text{if $j$ represents $s$,}\\
					\emptyset &\text{otherwise.} 
				\end{cases}
			\end{equation*}
			Observe that $\copyrel$ is $\regular{(0,\abs{H})}$ because for
			each $x \in \copyrel$ we have
			\begin{equation*}
				w(x) = \ell \cdot c = 0 \mod \abs{H}.
			\end{equation*}
			By \cref{lemma:gadget_clique_reg_relation}, there exists a gadget
			that \arb{c,H}-realizes $\copyrel$.
			Next, for each $i \in \bigl( [n] \setminus \Gamma_j \bigr)$ such that $l(i) \leq j \leq r(i)$, we introduce a gadget $F^{j}_i$ that \arb{c,H}-realizes the
			relation $\copyrel$ and identifies the portal vertices of $F^{j}_i$
			with
			\begin{equation}\label{eq:F_ij_portal_vertices}
				\left( a^{i,j}_1, \ldots, a^{i,j}_\ell \right) \odot \left( a^{i,j+1}_1, \ldots, a^{i,j+1}_\ell  \right).
			\end{equation}
			Finally, for each $i \in [n]$ and $j \in [t]$ such that $l(i) \leq j \leq r(i)$, we define the gadget that covers $i$ at step $j$ as

			\begin{equation*}
				K^{j}_i \coloneqq \begin{cases}
					N^{j}_{i_1,i_2} &\text{if $j$ represents $s \in [m]$ and $i \in \{i_1,i_2\} $} \\
					F^{j}_i &\text{otherwise.} 
				\end{cases}
			\end{equation*}
	\end{enumerate}
	This is the whole construction of the $\multicliquepartprob{c,d}$ instance which we call $G$.

	\subparagraph{Equivalence of the instances.}
	We now prove that $\psi$ is satisfiable if and only if $G$ admits a
	$\multipartitioning{c}{K_d}$, by establishing each direction
	separately.

	Suppose that there exists an assignment $(\alpha_1, \ldots, \alpha_n)$
	to $(x_1, \ldots, x_n)$ such that $\psi$ is satisfied.
	In the following, we will describe a $\multipartitioning{c}{K_d}$ $\mathcal{K}$ of $G$.
	For each $i \in
	[n]$, define $\bm{a}_i \coloneqq \Phi(\alpha_i) \in W$. Now let $j \in
	[t]$.
	Observe that since $\bm{a}_i \in W$ for each $i \in [n]$, it holds that
	there exists a $\multicliquepacking{c}{d}$ of
	$L^{i}$ that covers $(a^{i,l(i)}_1, \ldots,
	a^{i,l(i)}_\ell)$ according to $\Compl{c}{\bm{a}_i}$. Moreover, there exists a $\multicliquepacking{c}{d}$
	of $R^{i}$ that covers $(a^{i,r(i)+1}_1, \ldots,
	a^{i,l(i)+1}_\ell)$ according to $\bm{a}_i$. In both cases, the internal
	vertices of the gadgets are covered exactly $c$ times.

	Now let $j \in [t]$. For all $i \in \bigl( [n] \setminus \Gamma_j \bigr)$,
	there exists a $\multicliquepacking{c}{d}$ of $F^{j}_i$ that covers
	\begin{equation}\label{eq:F_j_portal_vertices}
		\left( a^{i,j}_1, \ldots, a^{i,j}_\ell \right) \odot \left( a^{i,j+1}_1, \ldots, a^{i,j+1}_\ell  \right)
	\end{equation}
	according to
	$\bm{a}_i \odot \Compl{c}{\bm{a}_i}$, because $\bm{a}_i \in W$.
	Moreover, if $j$ represents $s$ for some $s \in [m]$, and
	$x_{i_1}, x_{i_2}$ are the variables corresponding to $C_s$,
	then there exists a $\multicliquepacking{c}{d}$ that covers
	\begin{equation*}
		\left( a^{i_1,j}_{1}, \ldots, a^{i_1,j}_{\ell} \right)  \odot \left( a^{i_2,j}_{1}, \ldots, a^{i_2,j}_{\ell} \right) \odot \left(  a^{i_1,j+1}_{1}, \ldots, a^{i_1,j+1}_{\ell} \right) \odot \left( a^{i_2,j+1}_{1}, \ldots, a^{i_2,j+1}_{\ell} \right)
	\end{equation*}
	according to
	$\bm{a}_{i_1} \odot \bm{a}_{i_2} \odot \Compl{c}{\bm{a}_{i_1}} \odot \Compl{c}{\bm{a}_{i_2}}$.
	Again, in both cases, the internal vertices of the gadgets are covered
	exactly $c$ times.
	Next, we prove that
	the remaining vertices are covered $c$ times as well.
	
	\begin{claim}\label{claim:multi_clique_proof_forward}
		It holds that for each $i \in [n]$, $j \in [t]$ such that $l(i) \leq j \leq r(i) + 1$ and $x \in [\ell]$,
		$a^{i,j}_x$ is covered exactly $c$ times by $\mathcal{K}$.
	\end{claim}

	\begin{claimproof}
		We prove the claim by induction on $j$.
		Let $j = 1$, $i \in [n]$ and suppose that $l(i) \leq j \leq
		r(i) + 1$. Since $1 \leq l(i) \leq j = 1$, we have that $l(i) =
		1 = j$. Consider the vertices $\left( a^{i,1}_{1}, \ldots,
		a^{i,1}_{\ell} \right)$, which are covered according to
		$\Compl{c}{\bm{a}_i}$ by $L^{i}$. Moreover, $\left( a^{i,1}_{1}, \ldots, a^{i,1}_{\ell} \right)$
		is covered according to $\bm{a}_i$ by $K^{1}_i$ which follows from
		\cref{eq:F_j_portal_vertices} or \cref{eq:N_j_portal_vertices},
		depending on whether $i \in \Gamma_1$ or not, respectively.
		All in all, $a^{i,1}_x$
		is covered exactly $c$ times for $x \in [\ell]$.
	
		Now suppose that the claim holds for $1 < j \leq t$, and we will prove the claim
		for $j + 1$. Let $i \in [n]$
		such that $l(i) \leq j + 1 \leq r(i) + 1$.
		Consider
		the vertices $\left( a^{i,j+1}_{1}, \ldots, a^{i,j+1}_{\ell} \right)$,
		which are covered according to $\Compl{c}{\bm{a}_i}$
		by $K^{j}_i$.
		This follows from \cref{eq:F_j_portal_vertices,eq:N_j_portal_vertices}.
		Now, observe that we have either $j < r(i)$ or $j = r(i)$.
		In both cases, by using the arguments in the $j = 1$ case,
		one can conclude that $\left( a^{i,j+1}_{1}, \ldots, a^{i,j+1}_{\ell} \right)$
		is covered according to $\bm{a}_i$ by $K^{j+1}_i$.
		Therefore, all in all, it holds that each vertex $a^{i,j}_x$ is covered
		exactly $c$ times for $x \in [\ell]$.
	\end{claimproof}
	Now for the reverse implication, suppose that $G$ has a $\multipartitioning{c}{K_d}$.
	Since each gadget used in the construction of $G$ is $\coherent{(c,K_d)}$,
	this implies that for each gadget there is
	a $\multicliquepacking{c}{d}$ 
	that covers its interval vertices exactly $c$ times, and its portal vertices
	according to the relation associated with it.
	In particular, for each $i \in [n]$,
	let $\bm{b}_i \in W^{C}$ denote the vector such that $L^{i}$ covers the vertices
	$\{a^{i, l(i)}_{1},\ldots, a^{i, l(i)}_{\ell}\}$ according to $\bm{b}_i$.
	
	By induction, one can show that for each $j \in [t]$ and $i \in [n]$ such that $l(i) \leq j \leq r(i)$,
	the tuple $(a^{i, j}_{1},\ldots, a^{i, j}_{\ell})$ is covered by $K^{j}_{i}$
	according to $\bm{z}_i$ where $\bm{z}_i = \Compl{c}{\bm{b}_i}$.
	Hence, we let $\alpha_i \coloneqq \Phi^{-1}\left(\bm{z}_i\right) \in B$.
	Next, we will prove that $A \coloneqq (\alpha_1, \ldots, \alpha_n)$ is a satisfying assignment
	for $\psi$.
	To that end, let $s \in [m]$ and $j = b(s)$. To show that $C_s$ is satisfied by $A$,
	let $x_{i_1}$ and $x_{i_2}$ be the variables associated with $C_s$.
	By the above discussion, we know that $\left( a^{i_1,j}_{1}, \ldots, a^{i_1,l}_{\ell} \right)$ and
	$\left( a^{i_2,j}_{1}, \ldots, a^{i_2,j}_{\ell} \right)$ is covered
	according to $\bm{z}_{i_1}$ and $\bm{z}_{i_2}$, respectively.
	By the definition of $R_j$, $a_{i_1} = \Phi^{-1}( \bm{z}_{i_1})$
	and $a_{i_2} = \Phi^{-1}( \bm{z}_{i_2})$ satisfy $C_s$.
	Therefore, the assignment $A$ satisfies all constraints of $\psi$,
	and $\psi$ is satisfied if and only if $G$ has a
	$\multipartitioning{c}{K_d}$.

	\subparagraph{Pathwidth and size of the constructed instance.}
	To bound the pathwidth of $G$, we will create a path decomposition by
	following $\{B_j\}_{j \in [t]}$. Specifically, for each $j \in [t]$,
	we first create a new path decomposition $\mathcal{X} = \left( X_1, \ldots, X_t \right)$
	where each $X_j$ is a copy of $B_j$ and we replace each $x_i \in B$ with
	the vertices $a^{i,j}_1, \ldots, a^{i,j}_{\ell}$. Note that the size of each
	$X_j$ is at most $p \cdot \ell$ for $j \in [t]$.

	In the following, we will add the remaining vertices of $G$ to the bags in $\mathcal{X}$
	such that $\mathcal{X}$ is a valid path decomposition of $G$.
	Let $j \in [t]$. For each $i \in [n] \setminus \Gamma_j$ such that $l(i) \leq j \leq r(i)$,
	we replace the vertices $\{a^{i,j}_x\}_{x \in [\ell]}$ with $\{a^{i,j+1}_x\}_{x \in [\ell]}$ as follows.	
	After the bag $X_j$ ,we first insert the bag $X'_{j,i} \coloneqq \left( X_j \cup V\left( F^{j}_i \right) \right)$,
	and then add another bag $X^{''}_{j,i}$ where we replace $\{a^{i,j}_x\}_{x \in [\ell]}$ in $X'_{j,i}$ with
	$\{a^{i,j+1}_x\}_{x \in [\ell]}$, and finally, another bag $X^{'''}_{j,i}$ where we remove the vertices
	$V\left( F^{j}_i \right)$ from $X^{''}_{j,i}$. For a fixed $j \in [t]$,
	we keep adding the bags iteratively for each $i \in [n] \setminus \Gamma_j$ such that $l(i) \leq j \leq r(i)$, until all vertices in all the gadgets
	are contained in the bag decomposition. Note that, by our construction,
	edges that are adjacent to a vertex in $F^{j}_i$ are also covered by either $X'_{j,i}$ or $X^{''}_{j,i}$.
	
	Finally, for each $s \in [m]$ and $j = b(s)$, we add the gadgets
	$N^{j}_{i_1,i_2}$ to the tree decomposition in a similar way.
	Since the size of each $N^{j}_{i_1,i_2}$ and $F^{j}_i$ is a function of $\ell$,
	it is bounded by a constant.
	All in all, the pathwidth of $G$ increases at most by a constant.
	We have
	\begin{equation*}
		\pw(G) = l \cdot p + \mathcal{O}\left(1\right).
	\end{equation*}
	Finally, since $t$ and $m$ are bounded by a polynomial of $n$,
	it follows from the construction that
	\begin{equation*}
		V(G) = n^{\mathcal{O}\left(1\right)}.
	\end{equation*}

	\subparagraph{Running Time.}
	Constructing the graph $G$ from $\psi$ takes time polynomial in $n$.
	By our assumption on $\mathcal{A}$, the whole reduction takes time
	\begin{align*}
		(c+1)^{(1-\varepsilon) \cdot \pw(G)} \cdot V(G)^{b} &= (c+1)^{(1 - \varepsilon) \cdot l \cdot p} \cdot n^{\mathcal{O}\left(1\right)}\\
								    &= (c+1)^{(1 - \varepsilon) \cdot (\ell - \abs{H}) \cdot p} \cdot (c+1)^{(1 - \varepsilon) \cdot \abs{H} \cdot p} \cdot n^{\mathcal{O}\left(1\right)}\\
								    &= B^{(1 - \varepsilon) \cdot p} \cdot (c+1)^{(1 - \varepsilon) \cdot \abs{H} \cdot p} \cdot n^{\mathcal{O}\left(1\right)}\\
								    &< B^{(1 - \varepsilon) \cdot p}  \cdot (c+1)^{\frac{\varepsilon}{2} \cdot \left( \ell - \abs{H} \right)  \cdot p}\cdot n^{\mathcal{O}\left(1\right)}\\
								    &= B^{( 1 - \varepsilon)\cdot p} \cdot B^{\frac{\varepsilon}{2}\cdot p} \cdot n^{\mathcal{O}\left(1\right)}\\								    
								    &= B^{( 1 - \varepsilon')\cdot p} \cdot \abs{\psi}^{b'}
	\end{align*}
	where the inequality follows from \cref{eq:ell_choice} and $b'$ is a constant.
\end{proof}

The proof of \cref{theorem:multi_clique_packing_lower_bound} follows from
\cref{theorem:pwseth_equiv,lemma:multi_clique_hardness}. We note here that the same
construction can be used to prove \cref{theorem:single_clique_packing_lower_bound},
because the gadgets used in \cref{lemma:multi_clique_hardness} \strict{c,K_d}-realize
their relations (see \cref{definition:gadget_realizing,lemma:gadget_clique_reg_relation}). However, by presenting a simple reduction,
we demonstrate
that a fast algorithm for $\singlecliquepartprob{c,d}$ implies a fast algorithm
for $\multicliquepartprob{c,d}$, which is used to prove \cref{theorem:single_clique_packing_lower_bound}
in a more formal way.

\begin{lemma}\label{lemma:single_clique_hardness}
	Let $c \geq 1$ and $d \geq 3$ be integers. Suppose there exists an
	$\varepsilon > 0$ such that $\singlecliquepartprob{c,d}$ can be solved
	in time $(c+1 - \varepsilon)^{\pw\left( G \right) } \cdot
	n^{\mathcal{O}\left(1\right)}$ for all $n$-vertex graphs $G$ given together
	with a path decomposition of width at most $\pw(G)$. Then,
	there exist $\eps^{\prime}>0$ and $b^{\prime}>0$ 
	such that $\multicliquepartprob{c,d}$ can be solved
	in time $(c+1 - \varepsilon^{\prime})^{\pw\left( G \right) } \cdot
	n^{b^{\prime}}$ for all $n$-vertex graphs $G$ given together
	with a path decomposition of width at most $\pw(G)$.
\end{lemma}

\begin{proof}
	Let $\varepsilon$ and $b$ be as in the statement of
	the lemma. Moreover, let $\mathcal{A}$ be the hypothetical algorithm
	for $\singlecliquepartprob{c,d}$.

	\subparagraph{Construction of the $\singlecliquepartprob{c,d}$ instance.} Let $G$
	be an instance of $\multicliquepartprob{c,d}$.
	Moreover, we let $H$ denote the clique $K_d$ and
	construct a $\singlecliquepartprob{c,d}$ instance $G'$ as follows. Let
	$G'$ have the same vertex set as $G$. We call these vertices the
	original vertices of $G'$. Moreover, for each clique $X$ in $G$ of size
	$d$, we add a  gadget $E_X$ that \dist{c,H}-realizes the relation
	$\eqrel{d}{[0,c]}$,whose existence was given in
	\cref{lemma:single_partition_eqrel_abs_H}. We identify the portal
	vertices of $E_X$ with $V(X)$. This is the whole construction.

	\subparagraph{Equivalence of Instances.} We will now prove that $G$
	admits a $\multipartitioning{c}{K_d}$
	if and only if
	$G'$ admits a $\singlepartitioning{c}{K_d}$.	
	
	Suppose that $G$ has a $\multipartitioning{c}{K_d}$ which is denoted by $\mathcal{Z}$.
	For each clique $X$ in $G$, let $a_x$ denote the number of occurences of $X$
	in $\mathcal{Z}$. We construct a $\singlepartitioning{c}{K_d}$ $\mathcal{K}$ for $G'$
	as follows. For each clique $X$ in $G$, consider a $\singlecliquepacking{c}{d}$
	of $E_X$ such that the portal vertices of $E_X$ are covered $a_X$ times. Add this
	$\singlecliquepacking{c}{d}$ to $\mathcal{K}$.  Since the gadgets are disjoint,
	the copies of $K_d$ are also disjoint. The internal vertices of the gadgets are covered
	exactly $c$ times. Moreover, the original vertices of $G'$ are also covered $c$ times, since
	each equality gadget simulates a clique. Hence $\mathcal{K}$ is a valid $\singlepartitioning{c}{K_d}$
	and therefore $G'$ is a yes instance.

	Now suppose that $G'$ has a $\singlepartitioning{c}{K_d}$ $\mathcal{K}$. Then, for
	each clique $X$ in $G$, let $a_X$ denote the number of times $E_X$ covers
	its portal vertices. We construct a $\multicliquepacking{c}{d}$ $\mathcal{Z}$
	by including each $X$ exactly $a_X$ times in $\mathcal{Z}$. Moreover, $\mathcal{Z}$
	covers each vertex exactly $c$ times, hence it is a valid $\multipartitioning{c}{K_d}$.
	Therefore, $G$ is a yes-instance.

	\subparagraph{Running Time.} First, we show that the pathwidth of $G'$
	is $p + \mathcal{O}(1)$. Take the path decomposition of $G$
	and for each clique $X$, let $B_X$ denote the bag that contains the vertices
	of $X$. Note that since $X$ is a clique, there exists such a bag. Moreover, we may assume each bag $B_X$
	is unique for each clique $X$, duplicating bags if necessary.
	Then, we simply add the vertices of $E_X$ to the bag $B_X$, increasing its
	size at most by a constant. Therefore, we get a new path decomposition
	for $G'$ where the size of a bag is at most $p + \mathcal{O}\left(1\right)$.
	Hence, it holds that $\pw(G') = p + \mathcal{O}\left(1\right)$.

	Constructing the graph takes time polynomial in $n = V(G)$, and we also have $V(G') = n^{\mathcal{O}\left(1\right)}$. Therefore, the whole reduction takes time
	\begin{align*}
		(c+1 - \varepsilon)^{\pw(G')} \cdot V(G')^{b} &= (c+1 - \varepsilon)^{p + \mathcal{O}\left(1\right)} \cdot n^{\mathcal{O}\left(1\right)}\\
		&= (c+1 - \varepsilon)^{p} \cdot n^{\mathcal{O}\left(1\right)}\\
		&= (c+1 - \varepsilon)^{\pw(G)} \cdot n^{b'}
	\end{align*}
	where $b'$ is a constant.
\end{proof}

The proof of \cref{theorem:single_clique_packing_lower_bound} follows from
\cref{theorem:multi_clique_packing_lower_bound} and \cref{lemma:single_clique_hardness}.

%% file: lower_bounds_arb_graph.tex
Let $D$ be a set and $R \subseteq D^{n}$ be a relation of arity $n = k \cdot m$, for some integers $k, m \geq 1$. We represent each tuple $r \in R$ using double-index notation $r_{(i,j)}$, where $1 \leq i \leq k$ and $1 \leq j \leq m$, corresponding to a partition of $[n]$ into $k$ consecutive blocks of size $m$. The correspondence between the single-index and double-index notation is given by
\begin{equation}
r_{(i,j)} = r_{(i-1)\cdot m + j}.
\end{equation}
We always specify the decomposition $n = k \cdot m$ when introducing the relation
using double-index notation.

We also define the \emph{stacking operation} to construct a tuple in $D^{k \cdot m}$ from $k$ vectors in $D^m$. Given $\mathbf{v}^{(1)}, \ldots, \mathbf{v}^{(k)} \in D^m$, we define
\begin{equation}
	\stack{\mathbf{v}^{(1)}, \ldots, \mathbf{v}^{(k)}} \in D^{k \cdot m}
\end{equation}
as the tuple $\mathbf{v}$ satisfying
\begin{equation}
	\mathbf{v}_{(i,j)} \coloneqq \mathbf{v}^{(i)}_j.
\end{equation}
This operation arranges the $k$ vectors consecutively, aligning with the double-index partitioning.
\begin{definition}\label{definition:relations}
Let $x, y \geq 1$ and let $\mathcal{X} = (\mathcal{X}_1, \ldots, \mathcal{X}_u)$ be a partition of $[x]$.  
For each $w = (w_1, \ldots, w_u) \in \mathcal{X}_1 \times \cdots \times \mathcal{X}_u$, define $r^w, s^w \in \{0,1\}^{x \cdot y}$ by
\begin{align}
	\tau^w &\coloneqq \stack{\mathbf{u}^{(1)}, \ldots, \mathbf{u}^{(x)}},
\end{align}
where for each $1 \leq i \leq x$, the vectors $\mathbf{u}^{(i)} \in \{0,1\}^y$ are defined as
\begin{equation}
\mathbf{u}^{(i)} \coloneqq 
\begin{cases}
\overrightarrow{\mathbf{0}} & \text{if } i \in \{w_1, \ldots, w_u\}, \\
\overrightarrow{\mathbf{1}} & \text{otherwise},
\end{cases}
\end{equation}
and $\overrightarrow{\mathbf{0}}, \overrightarrow{\mathbf{1}} \in \{0,1\}^y$ denote the all-zero and all-one row vectors of length $y$, respectively.
We also define the relations $\relheavy{\mathcal{X}}{x,y} \subseteq \{0,1\}^{x \cdot y}$ as
\begin{align}
\relheavy{\mathcal{X}}{x,y} &\coloneqq \bigl\{\, \tau^w \mid w \in \left( \mathcal{X}_1 \times \cdots \times \mathcal{X}_u \right)  \,\bigr\}.
\end{align}
We write $\srelheavy{x,y}$ as shorthand for $\relheavy{[x]}{x,y}$.
\end{definition}

\begin{theorem}\label{theorem:reduction_graphpart_kk_indset}
	Let $H$ be a graph with at least $3$ vertices that is not a block graph.
	Suppose there exists an algorithm that solves every $n$-vertex $\graphpartprob{H}$ instance $G$ in time $2^{o\bigl( \pw(G) \cdot \log  \pw(G) \bigr) } \cdot n^{\mathcal{O}\left(1\right)}$. Then there exists an algorithm that solves every instance of $\kkindset$ in time $2^{o(k \cdot \log(k))}$.
\end{theorem}

\begin{proof}
	Recall that for a graph $G$, the set of vertices $S \subseteq V(G)$ is called
	a separator if $G \setminus S$ consists of at least two connected components.
	Observe that by definition $H$ has at least one block that is not a clique, which in
	particular implies that $B$ has at least $4$ vertices.
	Moreover, each such block has a separator of size at least $2$.
	In the following, let $B$ be a block of $H$ with a separator $S$
	of minimum cardinality.
	Partition the set $S$ into two sets $U$ and $D$ arbitrarily.
	Let $C_1, \ldots,C_t$ be the connected components of $H \setminus S$ for $t \geq 2$.
	
	\subparagraph{Construction of the $\graphpartprob{H}$ instance.}
	Recall that for any relation $R$ which is $\regular{(x, \abs{H})}$ for some $0 \leq x \leq \abs{H} -1$,
	there exists a $\coherent{(1,H)}$ gadget by \cref{lemma:gadget_arb_any_relation}
	that $H$-realizes $R$. Therefore, in the following,
	we simply assume the existence of a gadget for any relation considered,
	assuming is $\regular{(x, \abs{H})}$ for some $0 \leq x \leq \abs{H} -1$
	Now let $G$ be an instance of $\kkindset$ such that
	\begin{equation*}
		E(G) = \{e_1, \ldots, e_{\abs{E(G)}}\} .
	\end{equation*}
	\begin{figure}[htpb]
		\centering
		\includegraphics[scale=0.7]{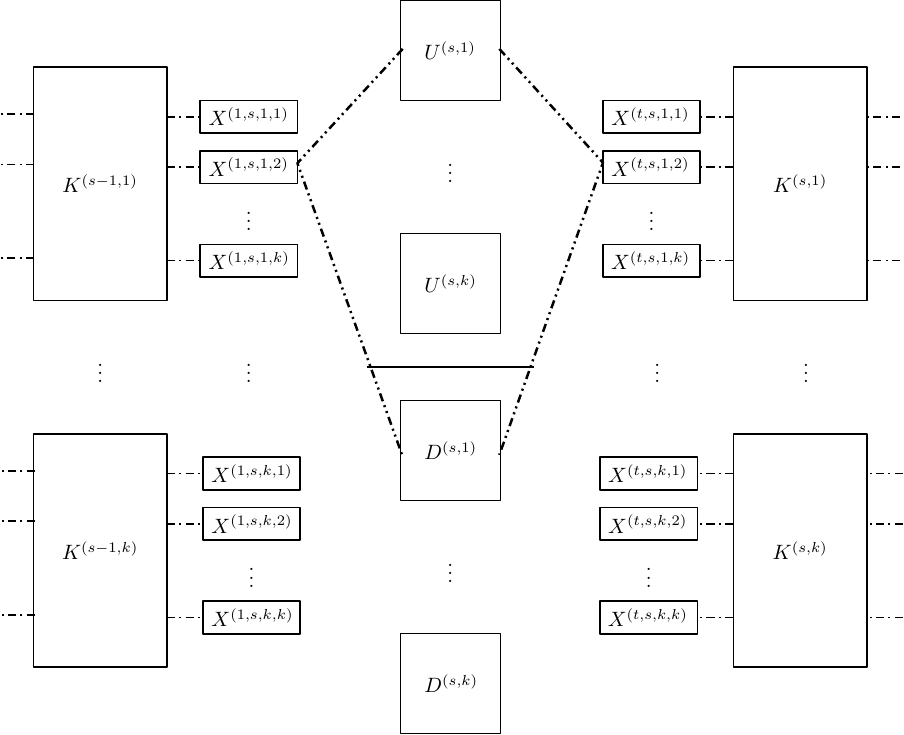}
		\caption{Part of the construction used to establish the lower bound in \cref{theorem:reduction_graphpart_kk_indset}. For clarity of presentation, some of the connections between $U^{(s,i)}, D^{(s,j)}$ and $X^{(\ell, s, i,j)}$ are omitted.} 
		\label{fig:X_U_D}
	\end{figure}	
	For each $1 \leq s \leq \abs{E(G)}$, we introduce
	$k$ copies of $U$ denoted by $U^{(s,1)}, \ldots, U^{(s,k)}$, and $k$ copies of
	$D$ denoted by $D^{(s,1)}, \ldots, D^{(s,k)}$. 
	We define
	\begin{align*}
		\mathcal{U}^{s} &\coloneqq \bigcup_{i = 1}^{k} U^{(s,i)}, \\
		\mathcal{D}^{s} &\coloneqq \bigcup_{i = 1}^{k} D^{(s,i)}
	\end{align*}
	and
	\begin{equation*}
		\mathcal{V}^{s} \coloneqq \mathcal{U}^{s} \cup \mathcal{D}^{s}.
	\end{equation*}
	
	Then, for each $1 \leq \ell \leq
	t$ and $1 \leq s \leq
	\abs{E(G)}$, we introduce $k^{2}$ copies of $C_\ell$ denoted by
	$X^{(\ell, s,i,j)}$ for $1 \leq i,j \leq k$,
	where we have $ V\left( X^{(\ell, s,i,j)} \right)   = \Big\{ x^{(\ell, s,i,j)}_1, \ldots, x^{(\ell, s,i,j)}_{\abs{C_{\ell}}} \Big\}$ and
	\begin{equation*}
		 \mathcal{X}^{(\ell, s,i,j)} \coloneqq \left( x^{(\ell, s,i,j)}_1, \ldots, x^{(\ell, s,i,j)}_{\abs{C_{\ell}}} \right).
	\end{equation*}
	We also let
	\begin{equation*}
		\mathcal{X}^{(\ell, s,i)} \coloneqq \stack{\mathcal{X}^{(\ell, s,i,1)} , \mathcal{X}^{(\ell, s,i,2)} ,  \,\ldots\, \mathcal{X}^{(\ell, s,i,k)}}.
	\end{equation*}
	\begin{figure}[htpb]
	  \centering
	  \begin{subfigure}[t]{0.48\textwidth}
	    \centering
	    \includegraphics[page=1, width=0.7\textwidth]{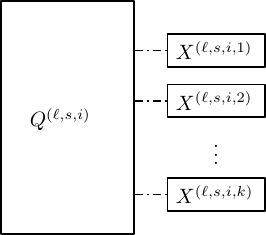}
	  \end{subfigure}
	  \hfill
	  \begin{subfigure}[t]{0.48\textwidth}
	    \centering
	   \includegraphics[page=2, width=\textwidth]{figures/arb_lb_gadgets_parts.pdf}
	  \end{subfigure}
	  \caption{Part of the construction used to establish the lower bound in \cref{theorem:reduction_graphpart_kk_indset}.}
	  \label{fig:arb_lb_gadgets_1}
	\end{figure}
	Then, for each $1 \leq \ell \leq t$, $1 \leq s \leq \abs{E(G)}$ and $1 \leq i,j \leq k$,
	we add the edges between $U^{(s,i)}, D^{(s,j)}$ and $X^{(\ell,s,i,j)}$
	according to their copies in $H$.
	For each $s \in \bigl[\abs{E(G)}\bigr]$, $\ell \in \{2,\ldots,t-1\} $
	and $i \in [k]$, we introduce a gadget $Q^{(\ell,s,i)}$ that $H$-realizes the relation
	$\srelheavy{k, \abs{C_\ell}}$ (which is $\regular{(\abs{C_{\ell}}, \abs{H})}$) and identify the portal vertices of
	$Q^{(\ell,s,i)}$ with $\mathcal{X}^{(\ell, s,i)}$. We also define $\mathcal{Q}^{(\ell,s)} = \bigcup_{i = 1}^{k} V(Q^{(\ell,s,i)} ) $.

	Now let $1 \leq s \leq \abs{E(G)} - 1$ and $e_s = \big\{  (i_s, j_s), (i'_s, j'_s) \big\}$
	be the $s$th edge of $G$. Without loss of generality, assume that $i_s \leq i'_s$.
	For each $i \in \{1, \ldots, k\} \setminus \{i_s, i'_s\}$,
	introduce a gadget $K^{^{(s,i)}}$ that realizes the relation $\srelheavy{k,
	\abs{C_t} + \abs{C_1}}$ and identify the portal vertices with
	\begin{equation*}
		\stack{\mathcal{X}^{(t, s,i,1)} \odot \mathcal{X}^{(1, s+1, i,1)},\, \mathcal{X}^{(t, s,i,2)} \odot \mathcal{X}^{(1, s+1, i,2)},\, \ldots,\, \mathcal{X}^{(t, s,i,k)} \odot \mathcal{X}^{(1, s+1, i,k)}}.
	\end{equation*}
	
	Then, define $\mathcal{Y} = \Bigl( \{1, \ldots, k\},  \{k+1, \ldots, 2\cdot k\}  \Bigr)$, $w = \bigl( j_s, j'_s \bigr)$ and introduce a gadget $Z^{s}$ that realizes the relation 
	\begin{equation*}
		\relheavy{\mathcal{Y}}{2k, \abs{C_t} + \abs{C_1}}  \setminus \{\tau^w\} 
	\end{equation*}
	where $\tau^{w}$ is defined in \cref{definition:relations}. 
	Then we let
	\begin{align*}
		\alpha_1 &\coloneqq \stack{\mathcal{X}^{(t, s,i_s,1)} \odot \mathcal{X}^{(1, s+1, i_s,1)},\, \mathcal{X}^{(t, s,i_s,2)} \odot \mathcal{X}^{(1, s+1, i_s,2)},\, \ldots,\, \mathcal{X}^{(t, s,i_s,k)} \odot \mathcal{X}^{(1, s+1, i_s,k)}},\\
		\alpha_2 &\coloneqq \stack{\mathcal{X}^{(t, s,i'_s,1)} \odot \mathcal{X}^{(1, s+1, i'_s,1)},\, \mathcal{X}^{(t, s,i'_s,2)} \odot \mathcal{X}^{(1, s+1, i'_s,2)},\, \ldots,\, \mathcal{X}^{(t, s,i'_s,k)} \odot \mathcal{X}^{(1, s+1, i'_s,k)}}
	\end{align*}
	and identify the portal vertices of $Z^{s}$ with $\stack{\alpha_1, \alpha_2}$.
	We also define
	\begin{equation*}
		\mathcal{C}^{s} \coloneqq \left( \bigcup_{i \in \{1, \ldots, k\} \setminus \{i_s, i'_s\}} V(K^{(s,i)}) \right) \cup V(Z^{s}).
	\end{equation*}
	For $s = \abs{E(G)}$ we introduce a gadget $F$ for the edge $\Big\{ 
	(i_s, j_s), (i'_s, j'_s) \Big\} $ that $H$-realizes the relation
	$\relheavy{\mathcal{Y}}{2k, \abs{C_t}} \setminus \{\tau^{w}\}$
	and identify its portal vertices with
	\begin{equation*}
		\stack{\mathcal{X}^{(t, s,i_s)} \odot \mathcal{X}^{(t, s, i'_s)} }.
	\end{equation*}

	Finally, for $s = 1$ and $i \in [k]$ we introduce a gadget $A_i$
	that $H$-realizes the relation $\srelheavy{k, \abs{C_1}}$ and identify
	its portal vertices with $\mathcal{X}^{(1,1,i)}$.
	Similarly, for
	$s = \abs{E(G)}$ and $i \in [k] \setminus \{i_s, i'_s\}$ we introduce
	a gadget $E_i$ that realizes the relation $\srelheavy{k, \abs{C_t}}$ and
	we identify
	its portal vertices with $\mathcal{X}^{(t,\abs{E(G)},i)}$.
	We also define $\mathcal{E} \coloneqq \left( \bigcup_{i \in \{1, \ldots, k\} \setminus \{i_s, i'_s\}} V(E_i) \right) \cup V(F)$.
	This concludes the construction of $G'$.

	\subparagraph{Equivalence of the instances.} Next, we will prove that $G$ has a permutation $k \times k$
	independent set if and only if $G'$ can be partitioned by vertex disjoint
	copies of $H$.
	
	First, assume that $G$ has a permutation $k \times k$ independent set,
	$\mathcal{T} = \{(i_1, j_1), \ldots, (i_k, j_k) \}$.
	Let $\mathcal{H} \coloneqq \emptyset$.

	Let $1 \leq s \leq \abs{E(G)}$. Observe that for each $1 \leq q \leq k$,
	$G'$ induced on the vertices of $\{X_{1}^{(s, i_q, j_q)},\ldots, X_{t}^{(s, i_q, j_q)}, U^{(s, i_q)}, D^{(s, j_q)}\} $ is a copy of $H$.
	Since they are vertex disjoint, we can add these $\abs{E(G)} \cdot k$ many
	copies of $H$ to $\mathcal{H}$.
	Observe that for $s = 1$, for each $i \in [k]$ there exists one, and only one $1 \leq j_i \leq k$
	such that $X^{(1,1,i,j_i)}$ is covered by a copy of $H$ in $\mathcal{H}$.
	Since $A_i$ realizes the relation $\srelheavy{k, \abs{C_1}}$, there exists
	an $\partitioning{H}$ of $A_i$ that covers $X^{(1,1,i,j)}$ for $j \in \left( [k] \setminus \{j_i\}  \right) $.
	We add this $\partitioning{H}$ to $\mathcal{H}$.
	Similarly, for each gadget used in the construction, one can verify that
	there exist an $\partitioning{H}$ of the gadget such that the remaining vertices are covered.
	Hence, $G'$ admits an $\partitioning{H}$.
	
	Now suppose that $G'$ is a yes instance, i.e. there exists an $H$-partition
	$\mathcal{H}$ of $G'$ that covers each vertex of $G'$ exactly once.
	Observe that for each $1 \leq \ell \leq t$ and $1 \leq i \leq k$, there exists a $\kappa(\ell,i) \in [k]$ such that
	$X^{\bigl(\ell, 1, i, \kappa(\ell,i)\bigr)}$ is left uncovered by the gadgets attached
	to $\mathcal{X}^{(\ell, 1, i)}$.
	This is because of the relations realized by those gadgets.
	Since each gadget used in the construction is $\coherent{(1,c)}$, this implies that
	there exists a collection $\mathcal{M}^{1} = \{M^{(1,1)}, M^{(1,2)}, \ldots, M^{(1,k)} \}  \subseteq \mathcal{H}$ of $k$ copies of $H$ covering
	the vertices
	\begin{equation*}
		\Omega^{1} \coloneqq \bigcup_{\ell \in [t], i \in  [k]} X^{\bigl(\ell, 1, i, \kappa(\ell,i)\bigr)}  \cup \{U^{(1,i)}\}_{i \in [k]} \cup \{D^{(1,i)}\}_{i \in [k]}.   
	\end{equation*}

	\begin{claim}\label{claim:arb_lb_M}
		For each $g \in [k]$, there exists $a \in [k]$, $b \in [k]$
		such that $M^{(1,g)}$ covers $U^{(1,a)}, D^{(1,b)}$, and no vertex
		from
		\begin{equation*}
			\mathcal{V}^{1} \setminus \left( U^{(1,a)} \cup D^{(1,b)} \right).
		\end{equation*}
	\end{claim}

	\begin{claimproof}
		Let $\Phi \geq 1$ denote the number of blocks in $H$ isomorphic to $B$.
		Observe that $C_1, \ldots, C_t$ contains in total $\Phi - 1$ blocks
		isomorphic to $B$. Moreover, observe that each block in $X^{\bigl(\ell, 1, i, \kappa(\ell,i)\bigr)}$
		is covered by a copy of $H$ in $\mathcal{M}^{1}$, because $\mathcal{M}^{1}$
		is an $\partitioning{H}$ and two vertex-disjoint copies of $H$ cannot
		both cover a cutvertex.
		Therefore, $k$ copies of $H$ give rise to $k \cdot \Phi$ blocks isomorphic to $B$, in total.

		Now, suppose that there exists $g \in [k]$ such that $M^{(1,g)}$ covers
		strictly less than $\abs{S}$ vertices from $V^{1}$.
		Since the size of the minimum separator that is contained strictly inside $B$ is equal
		to $\abs{S}$, it turns out that $\mathcal{M}^{1}$ contains at most $k \cdot \Phi - 1$
		block isomorphic to $B$, which is a contradiction.
		Therefore, each copy of $H$ in $\mathcal{M}^{1}$ contains at least $\abs{S}$
		vertices from $\mathcal{V}^{1}$. Since there are $k$ copies of $H$ in
		$\mathcal{M}^{1}$ and $\mathcal{V}^{1} = k \cdot \abs{S}$,
		it holds that for each $g \in [k]$ there exists $a,b \in [k]$ such
		that $M^{(1,g)}$ covers $U^{(1,a)}$ and $D^{(1,b)}$.
	\end{claimproof}

	\cref{claim:arb_lb_M} implies that for each $i \in [k]$, there exists $\kappa(i)$
	such that $\kappa(\ell,i) = \kappa(i)$ for $\ell \in [t]$,
	and there exists a copy of $H$ in $M^{1}$ that consitsts of $U^{(1,i)}$, 
	$D^{(1, \kappa(i))}$ and $X^{(\ell, 1, i, \kappa(i))}$ for $\ell \in [t]$.
	Since $M^{1}$ is an $\partitioning{H}$, it holds that $\{(i, \kappa(i)\}_{i \in [k]}$ is a permutation
	on $[k]$. It is easy to define $\mathcal{M}^{2}, \ldots, \mathcal{M}^{\abs{E(G)}}$
	and verify that the choice of $\kappa(i)$ is propagated through the construction.
	Moreover, for each edge $e = \{(a,b), (a', b')\}$ of $G$,
	there exists a gadget that makes sure that we don't have $e \subseteq \{(i, \kappa(i) )\}_{i \in [k]}$.
	Hence, $\{(i, \kappa(i) )\}_{i \in [k]}$ is a permutation $k \times k$ independent set.

	\subparagraph{Pathwidth of $G'$.} 
	In the following, we describe a path decomposition of $G'$ with width equal to $\mathcal{O}\left(k\right)$.
	First, we remark that all gadgets used in the construction have constant weight.
	Therefore, by \cref{lemma:gadget_arb_any_relation}, each gadget with $\ell$
	portal vertices admits a path decomposition with width at most $\mathcal{O}\left(\ell\right)$.

	Next, we define the bags $B_1, \ldots, B_{\abs{E(G)}}$ such that $B_s = \mathcal{V}^{s}$
	for $1 \leq s \leq \abs{E(G)}$.
	Then, before $B_1$, we introduce the following bags:
	\begin{enumerate}
		\item $\Bigl(B_1 \cup \mathcal{X}^{(1,1,i)} \cup V\left( A_i \right)\Bigr)$ for $i \in [k]$, and
		\item $\Bigl(B_1 \cup \mathcal{X}^{(\ell,1,i)} \cup V\left(Q^{(\ell,1,i)}\right)\Bigr)$ for $2 \leq \ell \leq t-1$ and $i \in [k]$.
	\end{enumerate}
	Observe that since the number of portals of $A_i$ and
	$V\left(Q^{(\ell,1,i)}\right)$ are bounded by $\mathcal{O}\left(\kappa\right)$,
	both gadgets have path decompositions of
	width at most $\mathcal{O}\left(\kappa\right)$ by
	\cref{lemma:gadget_arb_any_relation}.
	For $2 \leq s \leq \abs{E(G)}$ and $i \in [k]$, define the sets
	\begin{align*}
		\alpha^{(s,i)} &\coloneqq \mathcal{X}^{(\ell, s-1, i)} \cup V\left( K^{s-1, i} \right) \cup \mathcal{X}^{(1, s, i)},\\
		\beta^{s} &\coloneqq \mathcal{X}^{(\ell, s-1, i_s)} \cup \mathcal{X}^{(\ell, s-1, i'_s)} \cup V(Z^{s}) \cup \mathcal{X}^{(1, s, i_s)} \cup \mathcal{X}^{(1, s, i'_s)}.
	\end{align*}
	For $2 \leq s \leq \abs{E(G)}$, we insert between $B_{s-1}$ and $B_s$ a sequence of bags which starts with $B_s$,
	and at each step replaces
	\begin{enumerate}
		\item $V\left( U^{(s-1, i)} \right)$ with $\left( U^{(s, i)} \right)$ by adding/removing the set $\alpha^{(s,i)}$, for each $i \in \left( [k] \setminus \{i_s, i'_s\}  \right) $,
		\item $V\left( U^{(s-1, i_s)} \right) \cup V\left( U^{(s-1, i'_s)} \right)$ with $V\left( U^{(s, i_s)} \right) \cup V\left( U^{(s, i'_s)} \right)$, by adding/removing $\beta^{s}$, and
		\item $\bigcup_{i = 1}^{k} D^{(s-1, i)}$ with $\bigcup_{i = 1}^{k} D^{(s, i)}$.
	\end{enumerate}
	We note that $\abs{\alpha^{(s,i)}}$ and $\abs{\beta^{s}}$ are also bounded by $\mathcal{O}\left(k\right)$.
	Then, before $B_s$, we add the bags
	\begin{equation*}
		\left( B_s \cup \mathcal{X}^{(\ell, s, i)} \cup V\left( Q^{(\ell, s ,i)} \right)  \right) 
	\end{equation*}
	for $2 \leq \ell \leq t-1$ and $i \in [k]$.
	
	Finally, after the bag $B_s$, we insert the bags
	\begin{enumerate}
		\item $\Bigl(B_s \cup \mathcal{X}^{(\ell,\abs{E(G)},i)} \cup V\left( E_i \right)\Bigr)$ for $i \in [k] \setminus \{i_s, i'_s\}$, and
		\item $\Bigl(B_s \cup \mathcal{X}^{(\ell,\abs{E(G)},i_s)}, \mathcal{X}^{(\ell,\abs{E(G)},i'_s)}, V\left( F \right)  \Bigr)$.
	\end{enumerate}
	It can be verified that all edges are covered by the path decomposition and each bag has size
	at most $\mathcal{O}\left(k\right)$.
	
	\subparagraph{Running Time.}
	Constructing the graph $G'$ takes time polynomial in $k$, and we also have $n = V(G') = k^{\mathcal{O}\left(1\right)}$ and
	$\pw(G') = \mathcal{O}\left(k\right)$. Therefore, the whole reduction takes time
	\begin{align*}
		2^{o\bigl( \pw(G') \cdot \log  \pw(G') \bigr) } \cdot n^{\mathcal{O}\left(1\right)} &= 2^{o(k \cdot \log(k))} \cdot k^{\mathcal{O}\left(1\right)}\\
		&= 2^{o(k \cdot \log(k))}.
	\end{align*}
\end{proof}
The proof of \cref{theorem:graph_packing_arbitrary_lower_bound} follows from \cref{theorem:kkindset_lb_result,theorem:reduction_graphpart_kk_indset}.

%% file: algorithm_clique.tex
In this section, we present the proof of \cref{theorem:single_clique_packing_algo},
which we state here for completeness.
\singlecliquepackingalgo*

To prove the theorem, we construct a dynamic programming algorithm solving
the $\singlecliquepackprob{c,d}$ problem for $c \geq 1$ and $d
\geq 3$. Note that the case
$d = 1$ is trivial, while the case $d = 2$ corresponds to the maximization
version of the General Factor problem, which is known to be solvable in
polynomial time \cite{marxDegreesGapsTight2021b}. Therefore, we restrict our
attention to the case $d \geq 3$, as stated in
\cref{theorem:single_clique_packing_algo}.

Let $G$ be a graph, and let $\mathcal{T} = (T, {X_t}_{t \in V(T)})$ be a nice
tree decomposition of $G$. For each node $t$ in the decomposition, we define a
set of states encoding the necessary information to extend partial solutions
for the subtree rooted at $t$. Using dynamic programming, we compute the
optimal value for each state in a bottom-up manner, combining solutions from
child nodes according to the structure of the decomposition. The solution to
the original problem is then obtained from the computed values at the root.
In order to define the subproblems for a node $t$,
we first define the notion of types for the bag $X_t$.
\begin{definition}
	Given a graph $G$ and a node $t$ of its tree decomposition $\mathcal{T}$,
	a type for the bag $X_t$ is a function $f \from X_t \to [0, c]$. We let
	$\cliquetype{t}$ denote the set of all types of $X_t$.
\end{definition}

We now define a variant of $\singlecliquepacking{c}{d}$ for a node $t \in \mathcal{T}$.
\begin{definition}\label{definition:clique_packing}
	Let $c \geq 1$ and $d \geq 3$ be integers.
	Given a node $t \in \mathcal{T}$, a clique packing in $G_t$ is a set
	$\mathcal{P} = \{ Y_1, Y_2, \ldots, Y_\ell \} $ such that
	\begin{enumerate}
		\item $Y_i \subseteq V_t$ and $Y_i \nsubseteq X_t$ for $1 \leq i \leq \ell$,
		\item $G_t[Y_i]$ is a clique of size $d$,
		\item For all $x \in V_t$ it holds that
			\begin{equation*}
				\covernum{x,\mathcal{P}} \coloneqq \abs{\{P \in \mathcal{P} \mid x \in P\} } \leq c,
			\end{equation*}
			i.e., there are at most $c$ cliques in $\mathcal{P}$ that contain $x$.
	\end{enumerate}
\end{definition}

We say that a clique packing
has type $f \in \cliquetype{t}$, if for all $x \in X_t$, the number of cliques in $\mathcal{P}$ that contain $x$
is equal to $f(x)$. More formally, for all $x \in X_t$
\begin{equation*}
	\covernum{x,\mathcal{P}} = f(x).
\end{equation*}
Now we can describe the dynamic programming on $\mathcal{T}$.
For each node $t \in \mathcal{T}$ and type $f \in \cliquetype{t}$, we define
\begin{equation*}
	\partialpackings{t}{f} \coloneqq \{\mathcal{P} \mid \mathcal{P} \text{ is a clique packing in } G_t \text{ of type } f\} 
\end{equation*}
and
\begin{equation}\label{eq:clique_forget_A_def}
	A[t,f] \coloneqq \max_{\mathcal{P} \in \partialpackings{t}{f}} \abs{\mathcal{P}}.
\end{equation}
We call a clique packing $\mathcal{P} \in \partialpackings{t}{f}$ a witness for $A[t,f]$.
The solution to the original problem is given by $A[t_{\text{root}}, \emptyset]$
which follows from the fact that $X_{t_{\text{root} }} = \emptyset$.

\subsection{Update Rules For The Dynamic Programming Algorithm}
\label{sec:update_dp_clique_packing}
Next, in \cref{sec:update_dp_clique_packing_leaf,sec:update_dp_clique_packing_introduce,sec:update_dp_clique_packing_forget,sec:update_dp_clique_packing_join} we state the update rules for the dynamic programming algorithm for the
introduce, forget, join and leaf nodes.

\subsubsection{Leaf Node}
\label{sec:update_dp_clique_packing_leaf}
Let $t$ be a leaf node. Since $\mathcal{T}$ is a nice tree decomposition, it holds that
$X_t = \emptyset$ and there is only one function from the empty set, which we denote
by $\emptyset$. All in all, there is no witness for $A[t,\emptyset]$, and we let $A[t,\emptyset] = 0$.

\subsubsection{Introduce Node}
\label{sec:update_dp_clique_packing_introduce}
Let $t \in \mathcal{T}$ be an introduce node with a
single child $t'$ such that $X_{t} = X_{t'} \cup \{v\}$ and $f \in
\cliquetype{t}$. Recall that since $t$ is an introduce node, all neighbours
of $v$ are included in the bag $X_t$. Therefore,
any clique that contains $v$ should be entirely contained in $X_t$, which can not
be a part of a clique packing. Observe tahat if $f(v) = 0$,
then we can simply remove $v$ from $G_t$.
All in all, it holds that
\begin{equation*}
	A[t,f] = \begin{cases}
		0 &\text{if } f(v) \neq 0\\
		A[t', \funcrem{f}{v}] &\text{if } f(v) = 0.
	\end{cases}
\end{equation*}

\subsubsection{Forget Node}
\label{sec:update_dp_clique_packing_forget}
Now we state the update rule for the forget node.
Let $t \in \mathcal{T}$ be a forget node with a
single child $t'$ such that $X_{t'} = X_t \cup \{v\}$
and let $f \in \cliquetype{t}$.
For each $\ell \geq 1$, we define
\begin{align*}
	\Omega_{\ell} \coloneqq \Bigl\{ &\{ C_1, \ldots, C_\ell \}  \Bigm|  \{v\}
		\subseteq  C_i \subseteq X_{t'} \text{ and } C_i  \text{ is a clique of size $d$ for } 1 \leq i \leq
		\ell,\\
					&\covernum{x, \{C_1, \ldots, C_{\ell}\} } \leq f(x) \text{ for } x \in
	X_{t'}  \Bigr\}.
\end{align*}
We also set $\Omega_{0} = \emptyset$.
Moreover, given $0 \leq \zeta \leq c$
and $\mathcal{C} \in \Omega_{\ell}$ for some $\ell \geq 0$,
we define the function $\cliqueupdforget_{ f, \mathcal{C},\zeta } \from X_{t'} \to
[0, c]$  such that for $x \in X_{t'}$
\begin{equation}\label{eq:updforget_defin}
	\cliqueupdforget_{ f, \mathcal{C}, \zeta }\left( x \right) \coloneqq \begin{cases}
		f(x) - \covernum{x, \mathcal{C}} &\text{if } x \neq v\\
		\zeta &\text{if } x = v. 
	\end{cases}
\end{equation}

\begin{lemma}
Let $t \in \mathcal{T}$ be a forget node with a
single child $t'$ such that $X_{t'} = X_t \cup \{v\}$
and let $f \in
\cliquetype{t}$.
Then, it holds that
\begin{equation*}
	A[t,f] = \max_{\substack{\ell, \zeta \geq 0 \\ \text{s.t. } \ell + \zeta \leq c}} \,\max_{\mathcal{C} \in \Omega_{\ell}}\, \biggl( A[t', \cliqueupdforget_{f, \mathcal{C}, \zeta}] + \ell \biggr).
\end{equation*}	
\end{lemma}

\begin{proof}
	To prove the equality, we verify both the $\geq$ and $\leq$ inequalities independently.
	
	\medskip
	\noindent
	\textbf{(The inequality $\leq$)} \quad Let $\mathcal{P}$ be a witness
	for $A[t,f]$ such that $A[t,f] = \abs{\mathcal{P}}$. Observe that the
	number of cliques in $\mathcal{P}$ that contain $v$ is at most $c$
	since $\mathcal{P}$ is a clique packing. Among those cliques, let
	$\mathcal{C}_v$ denote the ones in which  $v$ is the only vertex that
	is not contained in $X_t$. Formally,
	\begin{equation*}
		\mathcal{C}_v \coloneqq \{C \in \mathcal{P} \mid \{v\}  = C \setminus X_t\}.
	\end{equation*}
	We define $\ell_v \coloneqq \abs{\mathcal{C}_v}$ such that $\mathcal{C}_v \in \Omega_{\ell_v}$,
	and let $\zeta_v \coloneqq \covernum{v,\mathcal{P}} - \ell_v$. 
	Observe that $\ell_v + \zeta_v \leq c$.

	Next, consider the clique packing $\mathcal{P}' \coloneqq \mathcal{P}
	\setminus \mathcal{C}_v$, which is obtained by removing the cliques in
	$\mathcal{C}_v$ from $\mathcal{P}$. Intuitively, the reason we remove
	those cliques is that each such clique is entirely contained within
	$X_{t'}$, which by definition is not part of any clique packing for
	$t'$. Then, $\mathcal{P}'$ is a clique packing for $t'$ of size
	$\abs{\mathcal{P}'} = \abs{\mathcal{P}} - \ell_v$, and it has type $\cliqueupdforget_{f,
	\mathcal{C}_v, \zeta_v}$. In particular, $\mathcal{P}'$ is a witness
	for
	\begin{equation*}
	  A[t', \cliqueupdforget_{f, \mathcal{C}_v, \zeta_v}].		
	\end{equation*}
	Hence,
	\begin{equation*}
	  A[t,f] = \abs{\mathcal{P}}
	  = \bigl(\abs{\mathcal{P}'} + \ell_v\bigr)
	  \leq \Bigl( A[t', \cliqueupdforget_{f, \mathcal{C}_v, \zeta_v}] + \ell_v \Bigr) 
	  \leq \max_{\substack{\ell, \zeta \ge 0 \\ \ell + \zeta \le c}} 
	       \,\max_{\mathcal{C} \in \Omega_{\ell}}
	       \bigl( A[t', \cliqueupdforget_{f, \mathcal{C}, \zeta}] + \ell \bigr),
	\end{equation*}
	where the first inequality follows from the definition of $A[t',
	\cliqueupdforget_{f, \mathcal{C}_v, \zeta_v}]$.

	\medskip
	\noindent
	\textbf{(The inequality $\geq$)} \quad Let $\ell_v, \zeta_v \geq 0$ such that $\ell_v + \zeta_v \leq c$ and
	$\mathcal{C}_v \in \Omega_{\ell_v}$. We want to show
	\begin{equation}\label{equation:A_t_f_single_rev_ineq}
		A[t,f] \geq \left( A[t', \cliqueupdforget_{f, \mathcal{C}_v, \zeta_v}] + \ell_v \right).
	\end{equation}
	Since $\ell_v, \zeta_v$ and $\mathcal{C}_v$ are chosen arbitrarily, this will imply
	\begin{equation*}
		A[t,f] \geq \max_{\substack{\ell, \zeta \geq 0 \\ \text{s.t. } \ell + \zeta \leq c}} \,\max_{\mathcal{C} \in \Omega_{\ell}}\, \biggl( A[t', \cliqueupdforget_{f, \mathcal{C}, \zeta}] + \ell \biggr).		
	\end{equation*}
	
	Let $\mathcal{P}_v$ be a witness for $A[t', \cliqueupdforget_{f, \mathcal{C}_v, \zeta_v}]$
	such that $A[t', \cliqueupdforget_{f, \mathcal{C}_v, \zeta_v}] = \abs{\mathcal{P}_v}$.
	In the following we will demonstrate that $\mathcal{P}_v \cup \mathcal{C}_v$
	is a witness for $A[t,f]$,
	i.e. we will prove that $\left( \mathcal{P}_v \cup \mathcal{C}_v \right) \in \partialpackings{t}{f}$.
	Following \cref{definition:clique_packing}, observe that for $Y \in \mathcal{P}_v$ we have
	$Y \nsubseteq X_{t'}$, and since $X_{t} \subseteq X_{t'}$, it holds that $Y \nsubseteq X_{t}$.
	Similarly, for $Y \in \mathcal{C}_v$, we have $v \in Y$, which again implies that
	$Y \nsubseteq X_{t}$.

	Since every clique in $\mathcal{C}_v$ is entirely contained in
	$X_{t'}$, while, by \cref{definition:clique_packing}, no clique in
	$\mathcal{P}_v$ is, it follows that $\mathcal{C}_v \cap \mathcal{P}_v =
	\emptyset$.
	Hence, for $x \in \Bigl(V_t \setminus \bigl( X_t \cup \{v\}  \bigr) \Bigr) = V_t  \setminus X_{t'}$, we have
	\begin{align*}
		\covernum{x, \mathcal{P}_v \cup \mathcal{C}_v} &= \covernum{x, \mathcal{P}_v} + \covernum{x, \mathcal{C}_v}\\
							       &= \covernum{x, \mathcal{P}_v}\\
							       &\leq c,
	\end{align*}
	where the inequality holds because $\mathcal{P}_v$ is a witness for $A[t', \cliqueupdforget_{f, \mathcal{C}_v, \zeta_v}]$.
	Similarly, for $v$ it holds that
	\begin{align*}
		\covernum{v, \mathcal{P}_v \cup \mathcal{C}_v} &= \covernum{v, \mathcal{P}_v} + \covernum{v, \mathcal{C}_v}\\
							       &= \covernum{v, \mathcal{P}_v} + \ell_v\\
							       &= \cliqueupdforget_{f, \mathcal{C}_v, \zeta_v}(v) + \ell_v\\
							       &= \zeta_v + \ell_v\\
							       &\leq c,
	\end{align*}
	where the second equality holds because $\mathcal{C}_v \in \Omega_{\ell_{v}}$, the third equality holds because $\mathcal{P}_v$ is a witness for $A[t', \cliqueupdforget_{f, \mathcal{C}_v, \zeta_v}]$ and the last equality holds by \cref{eq:updforget_defin}. Finally, for $x \in X_t$, we similarly have
	\begin{align*}
		\covernum{x, \mathcal{P}_v \cup \mathcal{C}_v} &= \covernum{x, \mathcal{P}_v} + \covernum{x, \mathcal{C}_v}\\
							       &= \cliqueupdforget_{f, \mathcal{C}_v, \zeta_v}(x) + \covernum{x, \mathcal{C}_v}\\
							       &= f(x) - \covernum{x, \mathcal{C}_v} + \covernum{x, \mathcal{C}_v}\\
							       &=f(x),
	\end{align*}
	where the third equality follows from the definition of $\cliqueupdforget_{f, \mathcal{C}_v, \zeta_v}$.
	Therefore, $\left( \mathcal{P}_v \cup \mathcal{C}_v \right)$ has type $f$ and therefore is a witness for $A[t,f]$.
	We have
	\begin{align*}
		A[t,f] \geq \abs{\mathcal{P}_v \cup \mathcal{C}_v} = \abs{\mathcal{P}_v} + \abs{\mathcal{C}_v} = \Bigl( A[t', \cliqueupdforget_{f, \mathcal{C}_v, \zeta_v}] + \ell_v \Bigr).
	\end{align*}
	This establishes \cref{equation:A_t_f_single_rev_ineq} and completes the proof.
\end{proof}

\subsubsection{Join Node}
\label{sec:update_dp_clique_packing_join}
Now we state the update rule for the join node.
Let $t \in \mathcal{T}$ be a join node with two children
$t_1,t_2 \in \mathcal{T}$ such that $X_t = X_{t_1} = X_{t_2}$
and also let $f \in
\cliquetype{t}$.
For $f_1 \in \cliquetype{t_1}$
and $f_2 \in \cliquetype{t_2}$, let $f_1 \oplus f_2$ be a function from $X_t$ to $[0,c]$
such that for $x \in X_t$ we have
\begin{equation*}
	\left( f_1 \oplus f_2 \right)(x) = f_1(x) + f_2(x).
\end{equation*}
In the case of a join node, the update rule can be expressed compactly using
the $\oplus$ operation, as formalized in the following theorem.
\begin{theorem}\label{theorem:single_clique_join}
	Let $t \in \mathcal{T}$ be a join node with two children
$t_1,t_2 \in \mathcal{T}$ such that $X_t = X_{t_1} = X_{t_2}$.
Also let $f \in
\cliquetype{t}$.
Then it holds that
\begin{equation}\label{equation:single_clique_join}
	A[t, f] = \max_{\substack{f_1 \in \cliquetype{t_1}, f_2 \in \cliquetype{t_2}\\ \text{s.t. } f = f_1 \oplus f_2 }}A[t_1, f_1] + A[t_2, f_2].
\end{equation}
\end{theorem}

\begin{proof}
	To prove the equality, we verify both the $\geq$ and $\leq$ inequalities independently.
	
	\medskip
	\noindent
	\textbf{(The inequality $\leq$)} \quad Let $\mathcal{P}$ be a witness for $A[t,f]$ such that $A[t,f] = \abs{\mathcal{P}}$.
	Define the sets
	\begin{align*}
		\mathcal{P}_1 &\coloneqq \{C \in \mathcal{P} \mid C \subseteq V_{t_1}\} \\
		\mathcal{P}_2 &\coloneqq \{C \in \mathcal{P} \mid C \subseteq V_{t_2}\}.
	\end{align*}
	Observe that $\left( \mathcal{P}_1,\mathcal{P}_2 \right) $ is a
	partition of $\mathcal{P}$, since a clique $C$ cannot be simultaneously
	contained in both $V_{t_1}$ and $V_{t_2}$ without being entirely
	contained in $X_t$, contradicting the fact that $\mathcal{P}$ is a
	witness for $A[t,f]$.

	Next, for each $x \in X_t = X_{t_1} = X_{t_2}$, define
	\begin{align*}
		f_1(x) &\coloneqq \covernum{x, \mathcal{P}_1} \quad \text{and} \\
		f_2(x) &\coloneqq \covernum{x, \mathcal{P}_2}.
	\end{align*}

	Since $(\mathcal{P}_1,\mathcal{P}_2)$ is a partition, we have $f_1(x) +
	f_2(x) = f(x)$ for all $x \in X_t$, which implies $f = f_1 \oplus f_2$.
	Because $X_t = X_{t_1} = X_{t_2}$, it follows that $f_1 \in
	\cliquetype{t_1}$ and $f_2 \in \cliquetype{t_2}$.

	By the definitions of $f_1$ and $f_2$, the sets
	$\mathcal{P}_1$ and $\mathcal{P}_2$
	are witnesses for $A[t_1,f_1]$ and $A[t_2,f_2]$ respectively.
	Therefore,
	\begin{align*}
		A[t,f] = \abs{\mathcal{P}} = \abs{\mathcal{P}_1} + \abs{\mathcal{P}_2} &\leq \Bigl(A[t_1,f_1] + A[t_2, f_2]\Bigr)\\
										       &\leq \max_{\substack{f_1 \in \cliquetype{t_1}, f_2 \in \cliquetype{t_2}\\ \text{s.t. } f = f_1 \oplus f_2 }} \biggl( A[t_1, f_1] + A[t_2, f_2]\biggr).
	\end{align*}

	\medskip
	\noindent
	\textbf{(The inequality $\geq$)} \quad Now let $f_1 \in \cliquetype{t_1}$ and $f_2 \in \cliquetype{t_2}$ be arbitrary functions
	such that $f = f_1 \oplus f_2$. We show that
	\begin{equation*}
		A[t,f] \geq A[t_1,f_1] + A[t_2, f_2].
	\end{equation*}
	Since $f_1$ and $f_2$ are chosen arbitrarily, this will imply
	\begin{equation*}
		A[t,f] = \abs{\mathcal{P}} = \abs{\mathcal{P}_1} + \abs{\mathcal{P}_2} \leq \Bigl(A[t_1,f_1] + A[t_2, f_2]\Bigr) \leq \Bigl(\max_{\substack{f_1 \in \cliquetype{t_1}, f_2 \in \cliquetype{t_2}\\ \text{s.t. } f = f_1 \oplus f_2 }}A[t_1, f_1] + A[t_2, f_2]\Bigr).
	\end{equation*}

	Let $\mathcal{P}_1$ and $\mathcal{P}_2$ be witnesses for $A[t_1,f_1]$ and $A[t_2,f_2]$ respectively,
	such that $\abs{\mathcal{P}_1} = A[t_1,f_1]$ and $\abs{\mathcal{P}_2} = A[t_2,f_2]$.
	We claim that $\mathcal{P}_1 \cup \mathcal{P}_2$ is
	a witness for $A[t, f]$.
	Observe that no $Y \in \mathcal{P}_1$ (or $\mathcal{P}_2$) is contained
	in $X_{t_1} = X_{t_2} = X_{t}$.
	Moreover, for $i \in \{1,2\}$ and
	$x \in V_{t_i} \setminus X_t$ we have
	\begin{equation}\label{eq:clique_join}
		\covernum{x, \mathcal{P}_1 \cup \mathcal{P}_2} = \Bigl( \covernum{x, \mathcal{P}_1} + \covernum{x, \mathcal{P}_2} \Bigr) = \Bigl(\covernum{x, \mathcal{P}_i}\Bigr) \leq c.
	\end{equation}
	Note that \cref{eq:clique_join} holds because if we let $j \in \{1,2\}$
	such that $\{i,j\} = \{1,2\}$, then
	\begin{equation*}
		x \in V_{t_i} \implies x \not\in  V_{t_j} \implies \covernum{x, \mathcal{P}_j} = 0.
	\end{equation*}
	On the other hand, for $x \in X_t$,
	\begin{align*}
		\covernum{x, \mathcal{P}_1 \cup \mathcal{P}_2} = \Bigl( \covernum{x, \mathcal{P}_1} + \covernum{x, \mathcal{P}_2} \Bigr) = f_1(x) + f_2(x) = f(x).
	\end{align*}
	Hence, $\mathcal{P}_1 \cup \mathcal{P}_2$ satisfies the conditions of a witness for $A[t,f]$, and
	\begin{align*}
		A[t,f] \geq \abs{\mathcal{P}_1 \cup \mathcal{P}_2} = \abs{\mathcal{P}_1} + \abs{\mathcal{P}_2} = A[t_1,f_1] + A[t_2,f_2].
	\end{align*}
	This completes the proof.
\end{proof}

\subsection{Correctness and Running Time}
	Let $\mathcal{A}$ denote the dynamic programming algorithm that fills the table $A[t,f]$
	for each $t \in \mathcal{T}$ and $f \in \cliquetype{t}$,
	following the update rules in \cref{sec:update_dp_clique_packing}.
	Next, we give the proof of \cref{theorem:single_clique_packing_algo}.
	
	\begin{proof}[Proof of \cref{theorem:single_clique_packing_algo}]
		Let $G$ be a graph and $\mathcal{T}$ be a tree decomposition of $G$
		of width at most $w$.
		We can compute a nice tree decomposition of $G$
		of the same width as $\mathcal{T}$ with $\mathcal{O}\left(w \cdot n\right)$
		nodes, in time $\mathcal{O}\left(w^{2} \cdot n\right)$ \cite{cyganParameterizedAlgorithms2015}.
		Therefore without loss of generality we can assume that $\mathcal{T}$
		is a nice tree decomposition.

		The update rules for each node of $\mathcal{T}$, and their correctness
		were described in \cref{sec:update_dp_clique_packing}.
		Recall that for the root node $t_{\text{R}}$  of $\mathcal{T}$ we have
		that $X_{t_\text{R}} = \emptyset$ and $V_{t_{\text{R}}} = G$, therefore
		$A[t_{\text{R}}, \emptyset]$ is the solution for the instance $G$.

		Now observe that for each $t \in \mathcal{T}$, there are at most $\mathcal{O}\left((c+1)^{w}\right)$
		many functions from $X_t$ to $[0,c]$. Therefore,
		\begin{equation*}
			\abs{\cliquetype{t}} = \mathcal{O}\left((c+1)^{w}\right).
		\end{equation*}
		If we let $t$ be a leaf, introduce or forget node,
		then we can compute the value  $A[t,f]$ for a single $f \in
		\cliquetype{t}$ in constant time, assuming the table
		entries for the children of $t$ are already computed.
		Hence, we can compute all the
		entries $A[t,f]$ over $f \in \cliquetype{t}$ in time $(c+1)^{w}
		\cdot n^{\mathcal{O}\left(1\right)}$.

		If $t$ is a join node, then a naive computation of the $A[t,f]$
		for a single $f \in \cliquetype{t}$ does not take constant time anymore, since we
		have to traverse over all $f_1 \in \cliquetype{t_1}$ (which also determines $f_2 \in \cliquetype{t_2}$
		such that $f = f_1 \oplus f_2$). 
		This results in a running time of $(c+1)^{2w} \cdot n^{\mathcal{O}\left(1\right)}$.
		However,
		the operation in \cref{equation:single_clique_join} can be described
		as a convolution which results in a faster running-time.
		We define the functions
		$\alpha_1, \alpha_2 \from \cliquetype{t} \times [0,n] \to \{0,1\} $ where	
		\begin{equation*}
			\alpha_1(f, i) = \begin{cases}
				1 &\text{if } A[t_1, f] = i\\
				0 &\text{otherwise} 
			\end{cases}
		\end{equation*}
		for $f \in \cliquetype{t}$ and $i \in [0,n]$.
		Similarly,
		\begin{equation*}
			\alpha_2(f, i) = \begin{cases}
				1 &\text{if } A[t_2, f] = i \\
				0 &\text{otherwise.}
			\end{cases}
		\end{equation*}
		Then we compute the function $\beta \from \cliquetype{t} \times [0,2n] \to \{0,1\}$
		where
		\begin{equation*}
			\beta(f,i) = \sum_{f_1 \oplus f_2 = f} \sum_{i_1 + i_2 = i} \alpha_1(f_1, i_1) \cdot \alpha_2(f_2, i_2).
		\end{equation*}
		Observe that $\beta$ is the convolution of the functions $\alpha_1$ and $\alpha_2$.
		Note that $\beta(f,i) \geq 1$ if and only if there exists $f_1, f_2 \in \cliquetype{t}$
		and $i_1,i_2 \in [0,n]$ such that $f = f_1 \oplus f_2$, $i = i_1 + i_2$, $A[t_1,f_1] = i_1$
		and $A[t_2, f_2] = i_2$. Therefore, given the function $\beta$, in order to compute
		\begin{equation*}
			A[t, f] = \max_{\substack{f_1 \in \cliquetype{t_1}, f_2 \in \cliquetype{t_2}\\ \text{s.t. } f = f_1 \oplus f_2 }}A[t_1, f_1] + A[t_2, f_2],
		\end{equation*}
		we simply iterate over all $i \in [0,n]$ and find the maximum value of $i$
		such that $\beta(f,i)$  is strictly positive, i.e., more formally we have
		\begin{equation}\label{equation:single_clique_convolution}
			A[t,f] = \max \Bigl( \bigl\{ 0 \bigr\} \cup \bigl\{ i \in [1,n] \mid \beta(f,i) \geq 1 \bigr\} \Bigr).
		\end{equation}
		Note that it is enough to consider $i \in [0,n]$ since we already
		know that $A[t,f] \leq n$.
		Computing the convolution, i.e. the function $\beta$, takes time $(c+1)^{w} \cdot  w^{2} \cdot n^{\mathcal{O}\left(1\right)}$ by Theorem 4.30 in
		\cite{fockeTightComplexityBounds2023}. Finally, given the function $\beta$,
		we can compute $A[t,f]$ for a single $f \in \cliquetype{t}$ in polynomial time.
		Therefore, all in all, computing $A[t,f]$ for all $f \in \cliquetype{t}$
		takes time
		\begin{equation*}
			(c+1)^{w} \cdot n^{\mathcal{O}\left(1\right)}.
		\end{equation*}
		Since $\mathcal{T}$ has $\mathcal{O}\left(w \cdot n\right)$ nodes, this
		proves the theorem.
	\end{proof}

%% file: algorithm_general_graph.tex
In this section we prove \cref{theorem:graph_packing_arbitrary_algo} by describing
a dynamic programming algorithm for
the $\graphpackprob{H}$ problem
where $H$ is an arbitrary graph.
We state \cref{theorem:graph_packing_arbitrary_algo} here for completeness.
\graphpackingarbitraryalgo*
Let $G$ be a graph and $\mathcal{T} = (T, \{X_t\}_{t \in V(T)})$ be a nice tree
decomposition of $G$. For a node $t \in T$ of the tree decomposition,
we occasionally adopt a slight abuse of notation and write $t \in \mathcal{T}$.
\subsection{Partial Packings and Types}
In this section, we introduce the notation and formally define the constructs
necessary to describe the dynamic programming algorithm.

\subparagraph{Partial Packings.}
To define the subproblems associated with a node $t \in \mathcal{T}$ for the
dynamic programming algorithm, we first introduce the notion of a
$\partialpacking{H}$. A homomorphism from a graph $H$ to $G$ is a function
$\theta \colon V(H) \to V(G)$ that preserves edges, i.e., for each edge $(a,b) \in E(H)$ it holds that
$\left( \theta(a), \theta(b) \right) \in E(G)$.

\begin{definition}[Partial Packing]\label{definition:part_packing_gen_graph}
    Given a graph $G$, and a node $t$ of its tree decomposition, a
    $\partialpacking{H}$ in $G_t$ is a tuple
    \begin{equation*}
    	\mathcal{P} = \left( h_1, \ldots, h_p,
        S_1, \overline{h}_1, \ldots, S_q, \overline{h}_q
    \right) 
    \end{equation*}
    for $p,q \geq 0$, subject to the following:
    \begin{enumerate}
	\renewcommand{\theenumi}{A\arabic{enumi}}
    	\renewcommand{\labelenumi}{\theenumi.}
		\item \label{item:h_i_inj} The function $h_i$ is an injective
			homomorphism from $H$ to $G_t \setminus X_t$ for $1
			\leq i \leq p$. Each $h_i$ is called a \emph{full copy}
			of $H$ in $\mathcal{P}$. We let
			$\numfullcopies{\mathcal{P}}$ denote the the number of
			full copies in $\mathcal{P}$, i.e., $p$.

	 	\item \label{item:h_bar_j} For each $1 \leq j \leq q$, it holds that
		 \begin{enumerate}
			 \item $S_j$ is a non-empty subset of $V(H)$,
			\item $\overline{h}_j$ is an injective homomorphism
				from $H[S_j]$ to $G_t$ such that
				$(\overline{h}_j)^{-1}(X_t) \neq \emptyset$,
			\item Let
				\begin{equation*}
					\border{j}{\mathcal{P}} = \{u \in S_j \mid \exists y \not\in  S_j \text{ such that } y \in N_{H}(u)\}.
				\end{equation*}
				Then, it holds that $\{\overline{h}_j(u) \mid u
				\in \border{j}{\mathcal{P}}\} \subseteq X_t$.
				That is, any vertex of $S_j$ that
				has a neighbor in $H \setminus S_j$
				must be mapped into $X_t$.
		 \end{enumerate} 
	 \item \label{item:im_disj} The images of functions in $\{h_i\}_{1 \leq i \leq p} \cup \{\overline{h}_j\}_{1 \leq j \leq q}$ are pairwise vertex-disjoint.		 
    \end{enumerate}
\end{definition}

Intuitively, a partial packing consists of full copies of $H$,
together with \textit{partial} copies of $H$ that might continue outside $G_t$.

\subparagraph*{Types.}
We now encode how a partial packing $\mathcal{P}$ intersects the bag $X_t$ by defining a \emph{type}, which summarizes that information more compactly.
The number of
types will depend on the size of the bag, i.e. $\abs{X_t}$, as opposed to the partial
packings whose size can be much larger. The idea is to use
the types to describe solutions for partial packings.

\begin{definition}[Type]\label{definition:type_gen_graph}
	Let $G$ be a graph and $t$ be a node of its tree decomposition. A type for $t$ is a tuple
	\begin{equation*}
		K = \left( \partition, Z_1, \phi_1, \ldots, Z_q, \phi_q \right) 
	\end{equation*}
	where $0 \leq q \leq \abs{X_t}$ and the components satisfy:
	\begin{enumerate}
		\renewcommand{\theenumi}{B\arabic{enumi}}
    		\renewcommand{\labelenumi}{\theenumi.}
		\item \label{item:part} $\partition$ is a function $X_t \rightarrow [0, q]$. For each $i \in [q]$, define
			\begin{equation*}
				X_t(i) \coloneqq \partition^{-1}(i).
			\end{equation*}
		\item \label{item:Z_i_nonempty} For all $i \in [q]$, $Z_i$ is a non-empty subset of $H$ such that $\abs{Z_i} = \abs{X_t(i)}$.
		\item \label{item:phi_i} For all $i \in [q]$, $\phi_i \colon V(H) \to X_t(i) \cup \{\downvertex, \upvertex\}$ satisfies:
		\begin{enumerate}
			\item The function $\phi_i$ restricted to $Z_i$, i.e. ${\phi_i}|_{Z_i}$, is an injective homomorphism from $H[Z_i]$ into $G[X_t(i)]$.
			\item There is no edge in $H$ that connects a vertex $u_1$ with $\phi_i(u_1) = \upvertex$, to a vertex
				$u_2$ with $\phi_i(u_2) = \downvertex$. In other words, $Z_i$ separates ``up'' vertices
				from ``down'' vertices in $H$.
		\end{enumerate}
	\end{enumerate}
\end{definition}

Let \(\generaltype{t}\) be the set of all types for the node \(t\). 
We next show how to derive a type from any partial packing \(\mathcal{P}\).

\subparagraph*{Type of a partial packing.}
Let $\mathcal{P} = \left( h_1, \ldots, h_p, S_1, \overline{h}_1, \ldots, S_q, \overline{h}_q \right)$ be a $\partialpacking{H}$ in $G_t$. We define its associated type
\begin{equation*}
	K = \left( \partition, Z_1, \phi_1, \ldots, Z_q, \phi_q\right)
\end{equation*}
as follows:

\begin{enumerate}
	\renewcommand{\theenumi}{C\arabic{enumi}}
    	\renewcommand{\labelenumi}{\theenumi.}	
	\item \label{item:tp_part} For each $x \in X_t$, we let
		\begin{equation}\label{eq:part_packing_to_type_arb}
			\partition(x) = \begin{cases}
				j &\text{if } x \in \im(\overline{h}_j) \text{ for some } 1 \leq j \leq q\\
				0 &\text{otherwise.} 
			\end{cases}
		\end{equation}
		Thus, $\partition(x) = j$ precisely when $x$ is in the image of $\overline{h}_j$, and $\partition(x)=0$ otherwise.
		By \cref{item:im_disj}, $\partition$ is well-defined.
	\item \label{item:tp_phi_j} For $1 \leq j \leq q$ and $u \in V(H)$, $\phi_j$ satisfies
		\begin{equation}\label{eq:phi_j_packing_to_type_arb}
			\phi_j(u) \coloneqq \begin{cases}
				\downvertex &\text{if } \overline{h}_j(u)\in \left( V_t\setminus X_t \right)\\
				\overline{h}_j(u) &\text{if } \overline{h}_j(u)\in X_t\\
				\upvertex &\text{if } u \not\in \dom{\overline{h}_j},\text{ i.e. } u \not\in S_j.
			\end{cases}
		\end{equation}
		We call $\phi_j$ the \emph{imprint} of $\overline{h}_j$ on $X_t$.
		In words, $\phi_j(u) = \;\downvertex$ if $u$ is mapped to a vertex in $V_t \setminus X_t$, equals $\overline{h}_j(u)$ if $u$ is mapped into $X_t$, and equals $\upvertex$ when $u$ is not in the domain of $\overline{h}_j$, i.e. $u \not\in  S_j$.
	\item \label{item:tp_Z_j} Finally, we let
		\begin{equation}\label{eq:H_j_rest_arb}
			Z_j \coloneqq \phi_j^{-1}(X_t).\
		\end{equation}
		Note that by \cref{eq:phi_j_packing_to_type_arb}, \cref{eq:H_j_rest_arb} is equivalent to $Z_j = \left( \phi_j^{-1} \right) (X_t) =  \left( \overline{h}_j^{-1} \right) (X_t)$.		
\end{enumerate}

These definitions ensure that $K$ satisfies \cref{definition:type_gen_graph} exactly when 
\cref{item:tp_part,item:tp_phi_j,item:tp_Z_j} hold for $(\mathcal{P},K)$. 
Thus each partial packing in $G_t$ determines exactly one type $K\in \generaltype{t}$.

\subsection{Description of the dynamic programming algorithm}
\label{section:dp_arb_graph}
Having defined partial packings and types for a node, we can now describe
the dynamic programming algorithm on a given tree decomposition.
For each node $t \in \mathcal{T}$ and $K \in \generaltype{t}$, define
\begin{equation*}
	\partialpackings{t}{K} \coloneqq \{\mathcal{P} \mid \mathcal{P} \text{ is a } \partialpacking{H} \text{ in } G_t \text{ of type } K\}.
\end{equation*}
Any $\mathcal{P} \in \partialpackings{t}{K}$ is called a \textit{witness} for $A[t,K]$.
Then, for any node $t$ and $K \in \generaltype{t}$ we define the table entry $A[t,K]$ as
\begin{equation*}
	A[t,K] \coloneqq \max_{\mathcal{P} \in \partialpackings{t}{K}} \numfullcopies{\mathcal{P}}, 
      \end{equation*}
i.e. the maximum number of full copies of $H$ found in a partial packing of type $K$.
We will compute the entries of the table $A$ in a bottom-up manner.
Next, we state the update rules for leaf, introduce, forget and join nodes.

\subsubsection*{Leaf Node} Let $t \in \mathcal{T}$ be a leaf node. Since
$\mathcal{T}$ is a nice tree decomposition, it holds that $X_t = V_t = \emptyset$.
Therefore, there is only one possible $K \in \generaltype{t}$ which is $K = (\emptyset)$
and we set $A[t, (\emptyset)] = 0$.

\subsubsection*{Introduce Node}
Let $t \in \mathcal{T}$ be an introduce node with a single child $t'$ such that $X_{t} = X_{t'} \cup \{v\}$,
where $v$ is the vertex introduced at node $t$. For each $K = \left( \partition, Z_1, \phi_1, \ldots, Z_q, \phi_q \right)  \in \generaltype{t}$, define a new type $\genupdintr(K) \in \generaltype{t'}$ such that \begin{equation*}
	\genupdintr(K) \coloneqq \begin{cases}
		\left( \gamma,\, Z_1, \phi_1, \ldots, Z_q, \phi_q\right) 
		& \text{if } \partition(v) = 0, \\
		\left( \gamma,\, Z_1, \phi_1, \ldots, Z_j \setminus u, \funcadd{\phi_j}{u}{\upvertex}, \ldots, Z_q, \phi_q \right) 
		& \begin{aligned}[t]
			&\text{if } \partition(v) = j \text{ for some} \\
			&j \in [q], u = \phi_j^{-1}(v), \abs{Z_j} \geq 2
		\end{aligned}\\
		\left( \gamma,\, Z_1, \phi_1, \ldots, Z_{j-1}, \phi_{j-1}, Z_{j+1}, \phi_{j+1}, \ldots, Z_q, \phi_q \right) 		& \begin{aligned}[t]
			\text{if } & \partition(v) = j \text{ for some} \\
				   &j \in [q], \abs{Z_j} = 1.
		\end{aligned}
	\end{cases}
\end{equation*}
where $\gamma \coloneqq \funcrem{\partition}{v}$ is the function $\partition$ restricted on $\Bigl( \dom{\partition} \setminus \{v\} \Bigr)  = \Bigl( X_t \setminus \{v\}  \Bigr) = X_{t'}$.

\begin{lemma}\label{lemma:introduce_arb_graph}
	Let $t \in \mathcal{T}$ be an introduce node with a single child $t'$
	such that $X_t = X_{t'} \cup \{v\}$.
	Then, for $K \in \generaltype{t}$, it holds that
	\begin{equation*}
		A[t,K] = A[t', \genupdintr\left( K,v \right) ].
	\end{equation*}
\end{lemma}

Consider some $K \in \generaltype{t}$ where $K = \left( \partition, Z_1, \phi_1, \ldots, Z_q, \phi_q\right)$.
We prove the lemma by considering three cases, each handled in a separate lemma, based on the value of $v$:
\begin{enumerate}
	\item $\partition(v) = 0$
	\item $\partition(v) = j$ for some $j \in [q]$ and $\abs{Z_j} \geq 2$
	\item $\partition(v) = j$ for some $j \in [q]$ and $\abs{Z_j} = 1$.		
\end{enumerate}

We start with the first case where $\partition(v) = 0$, which implies that for
any partial packing $\mathcal{P}$ that has type $K$, $v$ is not covered by any
partial copy of $H$ in $\mathcal{P}$. Therefore, since $X_t = \left( X_{t'}
\setminus \{v\} \right)$, it holds that $\mathcal{P}$ is also a valid partial
packing for the node $t'$. Observe that the full copies in $\mathcal{P}$ and
$\mathcal{P}'$ are identical, therefore their number also stays the same.

\begin{lemma}\label{lemma:introduce_arb_lem_1}
	Let $t \in \mathcal{T}$ be an introduce node with a single child $t'$
	such that $X_t = X_{t'} \cup \{v\}$.
	Assume $K \in \generaltype{t}$ is given by
	\begin{equation*}
		K = \left( \partition, Z_1, \phi_1, \ldots, Z_q, \phi_q\right)
	\end{equation*}
	with $\partition(v) = 0$. Define
	\begin{equation*}
		K' \coloneqq \genupdintr(K) = \left( \funcrem{\partition}{v}, Z_1, \phi_1,
		\ldots, Z_q, \phi_q\right).
	\end{equation*}
	Then it holds that $A[t,K] = A[t', K']$.
\end{lemma}

\begin{proof}
		First, observe $K'$ is a valid type for the node $t'$ as
		$X_{t'} = X_t \setminus \{v\}$. In the following, we will show
		that $A[t,K] \leq A[t', K']$ and $A[t,K] \geq A[t', K']$,
		starting with the former.
		
		\medskip
		\noindent
		\textbf{(The inequality $A[t,K] \leq A[t', K']$)} \quad Let $\mathcal{P} = \left( h_1,
			\ldots, h_p,
			S_1, \overline{h}_1, \ldots, S_q, \overline{h}_q
		\right) $ be a witness for $A[t,K]$ such that $A[t,K] =
		\numfullcopies{\mathcal{P}}$.
		It holds that $\mathcal{P}$ does not contain $v$, i.e.
		\begin{enumerate}
			\item $v \not\in \im(h_i)$ for $1 \leq i \leq p$, and
			\item $v \not\in \im(\overline{h}_j)$ for $1 \leq j \leq q$.
		\end{enumerate}
		The first condition holds because $\im(h_i) \subseteq \left(G_t
		\setminus X_t \right)$ for $1 \leq i \leq p$ and $v \in X_t$.
		The second condition also holds, otherwise $\partition(v) \neq
		0$ by \eqref{eq:part_packing_to_type_arb}. Therefore, since
		$G_{t'} = G_t \setminus \{v\}$, it holds that $\mathcal{P}' =
		\mathcal{P}$ is also a partial packing in $G_{t'}$ with type
		$\left(\funcrem{\partition}{v}, Z_1, \phi_1, \ldots, Z_q,
		\phi_q\right)$, i.e. $\mathcal{P}'$ is a witness for $A[t',
		K']$. Therefore
		\begin{equation}\label{eq:part_v_0_ineq_1}
			A[t,K] = \numfullcopies{\mathcal{P}} = \numfullcopies{\mathcal{P}'} \leq A[t', K'].
		\end{equation}

		\medskip
		\noindent
		\textbf{(The inequality $A[t,K] \geq A[t', K']$)} \quad Now let
		$\mathcal{P}'$ be a witness for $A[t', K']$ such that $A[t',
		K'] = \numfullcopies{\mathcal{P}'}$ and let $\mathcal{P} =
		\mathcal{P}'$. Observe that \cref{item:tp_phi_j} holds for
		$\left( \mathcal{P}, K \right) $ if and only if it holds for
		$\left( \mathcal{P}', K' \right)$, because the functions
		$\left( \phi_1, \ldots, \phi_q \right) $ are identical in $K$
		and $K'$. The same argument also shows that \cref{item:tp_Z_j}
		holds for $(\mathcal{P}, K)$. Finally, \cref{item:tp_part}
		holds for $(\mathcal{P}, K)$ because for each $x \in \left( X_t
		\setminus \{v\} \right)$ we have
		\begin{equation*}
			\partition(x) = \funcrem{\partition}{v}(x) = j
		\end{equation*}
		where $x \in \im(\overline{h}_j)$. The last step holds because $\mathcal{P}'$
		is a witness for $A[t', K']$. Note that we also have
		\begin{equation*}
			\partition(v) = 0,
		\end{equation*}
		which is consistent with $\mathcal{P}$ as $v \not\in \im\left( \overline{h}_j \right)$ for any $1 \leq j \leq q$.
		Therefore,
		$\mathcal{P}$ is a witness for $A[t,K]$. Hence
		\begin{equation}\label{eq:part_v_0_ineq_2}
			A[t', K'] = \numfullcopies{ \mathcal{P}' } = \numfullcopies{\mathcal{P}} \leq A[t,K].
		\end{equation}
		All in all, the lemma holds by \cref{eq:part_v_0_ineq_1,eq:part_v_0_ineq_2}.
\end{proof}

Now we prove the special case of \cref{lemma:introduce_arb_graph} where
$\partition(v) = j$ for $j \in [q]$ and $\abs{Z_j} \geq 2$. In that case, for
any partial packing in $G_t$ that has type $K$, there exists a partial copy $Y$ of
$H$ in $\mathcal{P}$ that covers $v$, together with another vertex from $X_t$.
Intuitively, one can remove $v$ from $Y$, and obtain a new partial packing
$\mathcal{P}'$ for the node $t'$ that has type $K'$. Similar to the previous case,
the number of full copies in $\mathcal{P}$ and $\mathcal{P}'$ are equal.

\begin{lemma}\label{lemma:introduce_arb_lem_2}
	Let $t \in \mathcal{T}$ be an introduce node with a single child $t'$
	such that $X_t = X_{t'} \cup \{v\}$.
	Assume $K \in \generaltype{t}$ is given by
	\begin{equation*}
		K = \left( \partition, Z_1, \phi_1, \ldots, Z_q, \phi_q\right)
	\end{equation*}
	with $\partition(v) = j$ for $j \in [q]$, $\abs{Z_j} \geq 2$ and $u = \phi_j^{-1}(v)$.
	Define
	\begin{equation*}
		K'  \coloneqq \genupdintr(K) = \left(\funcrem{\partition}{v},\, Z_1, \phi_1, \ldots, Z_j \setminus u, \funcadd{\phi_j}{u}{\upvertex}, \ldots, Z_q, \phi_q \right).
	\end{equation*}
	Then it holds that $A[t,K] = A[t', K']$.
\end{lemma}

\begin{proof}
	Let us first argue that $K'$ is a valid type for the node $t'$.
	In particular, we must show that $K'$
	satisfies \cref{definition:type_gen_graph}.

	\begin{claim}\label{claim:introduce_arb_lem_2_claim_1}
		It holds that $K'$ is a valid type for the node $t'$.
	\end{claim}

	\begin{claimproof}
		Observe that $\funcrem{\partition}{v}$ is a function from
		$\left( X_t \setminus \{v\} \right) = X_{t'}$ to $[0,q]$.
		Moreover, for $i \in \left( [1,q] \setminus \{j\} \right)$,
		$Z_i$ is a non-empty subset of $V(H)$ such that
		$\abs{Z_i} = \abs{\partition^{-1}(i)} = \abs{\left( \funcrem{\partition}{v} \right)^{-1} (i)}$,
		which follows from the fact that $K$ is a valid type
		for the node $t$.
		Finally, observe that $u = \phi_j^{-1}(v)$ implies that
		$\phi_j(u) = v \in X_t$, hence $u \in Z_j$ by
		\cref{definition:type_gen_graph}. Then, $Z_j \setminus u$ is a non-empty
		subset of $V(H)$ with
		\begin{equation*}
			\abs{Z_j \setminus u} = \left( \abs{Z_j}  - 1\right) = \left( \abs{\partition^{-1}(j)} - 1 \right)
			= \abs{\left( \funcrem{\partition}{v} \right) ^{-1}(j)},
		\end{equation*}
		where the second equality follows from \cref{item:Z_i_nonempty} for $K$.

		Similarly, for $i \in \left( [1,q] \setminus \{j\} \right) $, $\phi_i$ satisfies
		\cref{item:phi_i} in \cref{definition:type_gen_graph}, because $\phi_i$ is
		part of $K \in \generaltype{t}$.
		Moreover, $\funcadd{\phi_j}{u}{\upvertex}$
		is again a function from $H$ to $\left( \funcrem{\partition}{v} \right) ^{-1}(j)$.
		Also, restricted on $Z_j \setminus \{u\} $, the function $\funcadd{\phi_j}{u}{\upvertex}$
		is equal to $\phi_j$ and is an injective homomorphism.
		Therefore, $K'$ satisfies \cref{definition:type_gen_graph} and
		is a valid type for the node $t'$.
	\end{claimproof}
	
	In the following, we will show that $A[t,K] \leq A[t', K']$ and $A[t,K] \geq A[t', K']$,
	starting with the former.

	\begin{claim}\label{claim:introduce_arb_lem_2_claim_2}
		It holds that $A[t,K] \leq A[t', K']$.
	\end{claim}

	\begin{claimproof}
		Let $\mathcal{P} = \left( h_1,
			\ldots, h_p,
			S_1, \overline{h}_1, \ldots, S_q, \overline{h}_q
		\right)$ be a witness for $A[t,K]$ such that $A[t,K] = \numfullcopies{ \mathcal{P} }$.
		Since $P$ has type $K$, by \cref{item:tp_phi_j}
		we have that $v = \phi_j(u) = \overline{h}_j(u)$,
		where the last step holds because $u \in \phi_j^{-1}(X_t) = \overline{h}_j^{-1}(X_t)$ by \cref{eq:H_j_rest_arb}.
		Consider the partial packing
		\begin{equation*}
			\mathcal{P}' = \left( h_1,\, \ldots,\, h_p,\,
				S_1,\, \overline{h}_1,\, \ldots,\, S_j \setminus \{u\},\, \funcrem{\overline{h}_j}{v}  ,\, \ldots,\, S_q,\, \overline{h}_q.
		\right)
		\end{equation*}
		Observe that since $\mathcal{P}'$ is a restriction of $\mathcal{P}$, it is a well-defined
		partial packing for the node $t'$.
	
		Note that $\funcrem{\partition}{v}$ is identical to
		$\partition$, except that its domain no longer includes $v$.
		Likewise, $\funcadd{\phi_j}{u}{\upvertex}$ differs from
		$\phi_j$ only at $u$, and for every $i \in
		[1,q]\setminus\{j\}$, the function $\phi_i$ remains unchanged
		in both $\mathcal{P}$ and $\mathcal{P}'$.
		Hence, \cref{item:tp_part,item:tp_phi_j,item:tp_Z_j} hold for $(\mathcal{P},K)$, if and
		only if
		they hold for $(\mathcal{P}',K')$ since $X_{t'} = \left( X_t \setminus \{v\}  \right) $.
		Therefore, $\mathcal{P}'$ is a witness for
		$A[t', K']$. We get
		\begin{equation}\label{eq:intr_arb_Zj_geq_2_ineq_1}
			A[t,K] = \numfullcopies{\mathcal{P}} = \numfullcopies{ \mathcal{P}'} \leq A[t', K'].
		\end{equation}
		This proves the claim.
	\end{claimproof}

	Now we prove the other direction of the inequality.

	\begin{claim}\label{claim:introduce_arb_lem_2_claim_3}
		It holds that $A[t,K] \geq A[t', K']$.
	\end{claim}

	\begin{claimproof}
		Let 
		\begin{equation*}
			\mathcal{P}' = \left(b_1, \ldots, b_p, \eta_1, \alpha_1, \ldots, \eta_q, \alpha_q\right) 
		\end{equation*}
		be a witness for $A[t',K']$ such that $A[t', K'] =
		\numfullcopies{ \mathcal{P}' } $. In particular, this
		implies that $\phi_i$ is the imprint of $\alpha_i$ on
		$X_{t'}$ for $i \in [1,q] \setminus \{j\}$, and
		$\funcadd{\phi_j}{u}{\upvertex}$ is the imprint of
		$\alpha_j$ on $X_{t'}$. Since
		$\funcadd{\phi_j}{u}{\upvertex}(u) = \upvertex$, by
		\cref{item:tp_phi_j} for $\left( \mathcal{P}', K'
		\right) $, we have that $u \not\in  \eta_j$. Moreover,
		by \cref{item:tp_phi_j}, we have that $Z_i =
		\phi_i^{-1}\left( X_{t'} \right) =
		\alpha_i^{-1}\left( X_{t'} \right) $ for $i \in [1,q]
		\setminus \{j\}$, and
		\begin{equation}\label{eq:Zj_minus_u}
			\left(Z_j \setminus \{u\} \right) = \left( \funcadd{\phi_j}{u}{\upvertex} \right)^{-1}\left( X_{t'} \right) = \alpha_j^{-1}\left( X_{t'} \right).
		\end{equation}

		Now consider the partial packing
		\begin{equation*}
			\mathcal{P} \coloneqq \left( b_1, \ldots, b_p, \eta_1, \alpha_1, \ldots, \eta_j \cup \{u\} , \funcadd{\alpha_j}{u}{v}, \ldots, \eta_q, \alpha_q \right).
		\end{equation*}
		The partial packing $\mathcal{P}$ is constructed from a
		partial packing, $\mathcal{P}'$, by modifying $\eta_j$
		and $\alpha_j$. Hence, in order to prove that $\mathcal{P}$
		is a valid partial packing and verify that
		$\mathcal{P}$ satisfies
		\cref{definition:part_packing_gen_graph}, it is enough
		to check whether the modified graph/function satisfies
		the properties.

		Following \cref{definition:part_packing_gen_graph}, we
		first show that $\zeta \coloneqq
		\funcadd{\alpha_j}{u}{v}$ is an injective homomorphism
		from $H[\eta_j \cup \{u\}]$ to $G_t$.
		We already know that $\zeta$ is an injective
		homomorphism when restricted to $\eta_j$, because
		$\zeta$ and $\alpha_j$ agree on $\eta_j$. Therefore,
		consider an edge $\{u, x\}$ in $H[\eta_j \cup \{u\}]$
		for some $x \in \eta_j \subseteq X_{t'}$. We have
		\begin{equation*}
			\zeta\left( u \right) = v = \phi_j(u)\\
		\end{equation*}
		by definition of $u$, and
		\begin{equation*}
			\zeta\left( x \right) = \alpha_j(x) = \funcadd{\phi_j}{u}{\upvertex}(x) = \phi_j(x)
		\end{equation*}
		where the second equality holds because
		$\funcadd{\phi_j}{u}{\upvertex}$ is the imprint of
		$\alpha_j$ on $X_{t'}$. Moreover, observe that $x \in
		\border{j}{\mathcal{P}'}$, which implies that
		$\alpha_j(x) \in X_{t'}$ and $x \in \left( Z_j
		\setminus \{u\}  \right) $ by \cref{eq:Zj_minus_u}.
		Now recall that $\phi_j$ is a homomorphism on $Z_j$.
		Since $u,x \in Z_j$, it holds that there is an edge
		between $\zeta\left( u \right)$ and $\zeta\left( x
		\right)$. Moreover, $\zeta$ is injective because
		$\alpha_j$ is injective and $x \not\in  X_{t'}$ implies
		$v \not\in \im\left( \alpha_j \right) $. Therefore,
		$\zeta$ is an injective homomorphism.

		Next, we need to prove that the images of the functions
		in $\mathcal{P}$ are disjoint. However this property
		already holds for $\mathcal{P}'$ and we furthermore
		have $v \not\in  V_{t'}$. Hence this property holds for
		$\mathcal{P}$ as well.
		Finally, observe that $\border{j}{\mathcal{P}}
		\subseteq \left( \border{j}{\mathcal{P}'} \cup \{u\}
		\right)$ and $\funcadd{\alpha_j}{u}{v}(u) = v \in X_t$.
		Hence all conditions in
		\cref{definition:part_packing_gen_graph} hold and
		therefore $\mathcal{P}$ is a valid partial packing for the node
		$t$.

		Next, we prove that $\mathcal{P}$ has type $K$, i.e. $\left(
		\mathcal{P}, K \right) $ satisfies
		\cref{item:tp_part,item:tp_phi_j,item:tp_Z_j}. Note that
		\cref{item:tp_part} holds for all $x \in X_{t'}$ since
		$\partition$ and $\funcrem{\partition}{v}$ agree on $X_{t'}$.
		Moreover, we also have $\partition(v) = j$ and $v \in \im\left(
		\funcadd{\alpha_j}{u}{v} \right)$. Hence \cref{item:tp_part}
		holds for $\left( \mathcal{P}, K \right) $.

		We already know that $\phi_i$ is the imprint of
		$\alpha_i$ on $X_{t'}$ for $i \in \left( [1,q]
		\setminus \{j\} \right) $. 
		Since $\im\left( \alpha_i \right) \cap V_t = \im\left( \alpha_i \right) \cap V_{t'}$
		for $i \in \left( [1,q] \setminus \{j\}  \right) $, it holds that $\phi_i$ is the
		imprint of $\alpha_i$ on $X_t$ as well.
		Now recall that $\funcadd{\phi_j}{u}{\upvertex}$ is the imprint of $\alpha_j$ on $X_{t'}$,
		which means for all $s \in \left( V(H) \setminus \{u\}  \right) $, we have
		\begin{equation}\label{eq:phi_j_eq_arb}
			\phi_j(s) = \funcadd{\phi_j}{u}{\upvertex}(s) =
			\begin{cases}
				\downvertex &\text{if } \alpha_j(s) \in \left( V_t \setminus X_t \right) \\
				\alpha_j(s) &\text{if } \alpha_j(s) \in X_t\\
				\upvertex &\text{if } s \not\in \dom{\alpha_j}.
			\end{cases}
		\end{equation}
		Since we can replace $\alpha_j$ with $
		\funcadd{\alpha_j}{u}{v} $ in \eqref{eq:phi_j_eq_arb},
		we get that \cref{item:tp_phi_j} holds for all $s \in
		\left( V(H) \setminus \{u\}  \right) $. We also
		have
		\begin{equation*}
			\phi_j(u) = v = \funcadd{\alpha_j}{u}{v}(u)
		\end{equation*}
		and therefore \cref{item:tp_phi_j} holds for $\left(
		\mathcal{P}, K \right)$.

		Finally, recall that $Z_i = \alpha_i^{-1}\left( X_{t'}
			\right) $ for $i \in \left( [1,q] \setminus
			\{j\} \right) $. Since $\im\left( \alpha_i
		\right)
		\subseteq X_{t'}$, we
		have that
		\begin{equation*}
			Z_i = \alpha_i^{-1}\left( X_{t'} \right)  = \alpha_i^{-1}\left( X_{t} \right).
		\end{equation*}
		We also have
		\begin{equation*}
			\left( \funcadd{\alpha_j}{u}{v} \right)^{-1}(X_t) = \biggl(\{u\} \cup \alpha_j^{-1}\left( X_{t'} \right)\biggr) = \biggl( \{u\} \cup \left( Z_j \setminus \{u\}  \right) \biggr) = Z_j
		\end{equation*}
		since $u \in Z_j$. Therefore,
		\cref{item:tp_Z_j} holds for $\left( \mathcal{P}, K \right)$,
		and $\mathcal{P}$ is a witness for $A[t,K]$, which implies that
		\begin{equation*}
			A[t', K'] = \numfullcopies{\mathcal{P}'} = \numfullcopies{\mathcal{P}} \leq A[t,K].
		\end{equation*}
		This proves the claim.
	\end{claimproof}
	The lemma follows from \cref{claim:introduce_arb_lem_2_claim_2,claim:introduce_arb_lem_2_claim_3}.
\end{proof}

Finally, we prove the final special case of \cref{lemma:introduce_arb_graph}
where $\partition(v) = j$ for $j \in [q]$ and $\abs{Z_j} = 1$. Intuitively,
given any partial packing $\mathcal{P}$ that has type $K$, one can construct a
new partial packing $\mathcal{P}'$ that has type $K'$ by removing the $j$'th
partial copy of $H$ from $\mathcal{P}$. Note that as in the previous cases, the
number of full copies in $\mathcal{P}'$ stays the same.

\begin{lemma}\label{lemma:introduce_arb_lem_3}
	Let $t \in \mathcal{T}$ be an introduce node with a single child $t'$
	such that $X_t = X_{t'} \cup \{v\}$.
	Assume $K \in \generaltype{t}$ is given by
	\begin{equation*}
		K = \left( \partition, Z_1, \phi_1, \ldots, Z_q, \phi_q\right)
	\end{equation*}
	with $\partition(v) = j$ for $j \in [q]$ and $\abs{Z_j} = 1$. Define
	\begin{equation*}
		K' \coloneqq \genupdintr(K) = \left(\funcrem{\partition}{v} ,\, Z_1, \phi_1, \ldots, Z_{j-1}, \phi_{j-1}, Z_{j+1}, \phi_{j+1}, \ldots, Z_q, \phi_q \right).
	\end{equation*}
	Then it holds that $A[t,K] = A[t', K']$.
\end{lemma}

\begin{proof}
	First , observe that since $K'$ is obtained from $K$ by removing $Z_j$ and $\phi_j$,
	$K'$ is a valid type for the node $t'$.
	In the following, we will show
		that $A[t,K] \leq A[t', K']$ and $A[t,K] \geq A[t', K']$,
		starting with the former.
	Let $\mathcal{P} = \left( h_1,
		\ldots, h_p,
		S_1, \overline{h}_1, \ldots, S_q, \overline{h}_q
	\right)$ be a witness for $A[t,K]$ such that $A[t,K] = \numfullcopies{ \mathcal{P} }$.
	We define
	\begin{equation*}
		\mathcal{P}' \coloneqq \left( h_1,
		\ldots, h_p,
		S_1, \overline{h}_1, \ldots, S_{j-1}, \overline{h}_{j-1}, S_{j+1}, \overline{h}_{j+1}, \ldots, S_q, \overline{h}_q
		\right)
	\end{equation*}
	Since $\mathcal{P}'$ is obtained from $\mathcal{P}$ by removing
	$S_j$ and $\overline{h}_j$, $\mathcal{P}'$ is a valid partial packing
	for the node $t'$. Therefore,
	\begin{equation}\label{eq:introduce_arb_lem_3_ineq_1}
		A[t,K] = \numfullcopies{ \mathcal{P} } = \numfullcopies{\mathcal{P}'} \leq A[t', K'].
	\end{equation}

	Now let $\mathcal{P}' = \left( h_1,
		\ldots, h_p,
		S_1, \overline{h}_1, \ldots, S_{j-1}, \overline{h}_{j-1}, S_{j+1}, \overline{h}_{j+1}, \ldots, S_q, \overline{h}_q
	\right)$ be a witness for $A[t', K']$ such that $A[t', K'] =
	\numfullcopies{\mathcal{P}'}$. Let us define $S_j = \{v\}$ and
	$\overline{h}_j$ be the function with the domain $\{\phi_j^{-1}(v)\}$
	that maps $\phi_j^{-1}(v)$ to $v$.
	Then
	\begin{equation}\label{eq:introduce_arb_lem_3_ineq_2}
		\mathcal{P} = \left( h_1,
		\ldots, h_p,
		S_1, \overline{h}_1, \ldots, S_{j-1}, \overline{h}_{j-1},  S_{j}, \overline{h}_{j}, S_{j+1}, \overline{h}_{j+1}, \ldots, S_q, \overline{h}_q
	\right)
	\end{equation}
	is a valid partial packing for the node $t$,
	since \cref{item:h_bar_j} is satisfied trivially for $j$ and
	$v$ is not
	contained in any of the functions $\{\overline{h}_i\}_{i \in \bigl([q] \setminus \{j\} \bigr)}$.
	Therefore we get
	\begin{equation*}
		A[t', K'] = \numfullcopies{\mathcal{P}'} = \numfullcopies{\mathcal{P}} \leq A[t,K].
	\end{equation*}
	The lemma holds by \cref{eq:introduce_arb_lem_3_ineq_1,eq:introduce_arb_lem_3_ineq_2}.
\end{proof}

The proof of \cref{lemma:introduce_arb_graph} follows from \cref{lemma:introduce_arb_lem_1}, \cref{lemma:introduce_arb_lem_2} and \cref{lemma:introduce_arb_lem_3}.

\subsubsection*{Forget Node}
Let $t \in \mathcal{T}$ be a forget node with a single child $t'$ such that
$X_{t'} = X_t \cup \{v\}$. For each $K = \left( \partition, Z_1, \phi_1,
\ldots, Z_q, \phi_q \right) \in \generaltype{t}$, we construct three disjoint sets of types for the node $t'$,
corresponding to different ways $v \in X_{t'}$ can appear in these types.

\begin{enumerate}
	\item  We first consider the types for the node $t'$ in which $v$
		appears in the image of a function that maps every other
		element to $\downvertex$.
		Define
		\begin{equation*}
			\genupdforgetone\left( K \right) \coloneqq \Bigl\{ \left( \funcadd{\partition}{v}{q+1} , Z_1, \phi_1, \ldots, Z_q, \phi_q, \{u\}, g_{u,v} \right) \Big|\, u \in V(H) \Bigr\} 
		\end{equation*}
		where $g_{u,v}(u) = v$, and $g_{u,v}(x) = \downvertex$ for $x \in V(H) \setminus \{u\}$.
		We also set
		\begin{equation*}
			\Gamma_1 \coloneqq \max_{K' \in \genupdforgetone\left( K \right)} A[t', K'] + 1.
		\end{equation*}
	\item Next, we consider the types $K'$ which is constructed from $K$ by modifying $Z_j$ and $\phi_j$
		for some $j \in [q]$. In particular, for some $j \in [q]$, we pick a vertex
		$u \in V(H)$ such that $\phi_j(u) =\, \downvertex$, which will be mapped to $v$.
		However, the neighbors of $u$ in $H[Z_j]$ should be mapped to neighbours of $v$, in order
		for the new function to be a homomorphism. In other words, we need
		\begin{equation*}
			N_H(u) \subseteq \phi_j^{-1}\Bigl( N_G(v) \Bigr).
		\end{equation*}
		Therefore we define
		\begin{align*}
			\genupdforgettwo(K) \coloneqq \Bigl\{\Big( \funcadd{\partition}{v}{j},\, &Z_1, \,\phi_1, \,\ldots,\, Z_j \cup \{u\} , \,\funcadd{\phi_j}{u}{v}, \,\ldots, Z_q, \,\phi_q \Big) \\
														     & \Bigm| 1 \leq j \leq q,\, u \in \phi_j^{-1}(\downvertex),\, \left( N_H(u) \cap Z_j \right)  \subseteq \phi_j^{-1}\Bigl( N_G(v) \Bigr)\Bigr\}
		\end{align*}
		and
		\begin{equation*}
			\Gamma_2 \coloneqq \begin{cases}
				\max_{K' \in \genupdforgettwo\left( K \right)} A[t', K'] &\text{if } \genupdforgettwo\left( K \right) \neq \emptyset\\
				0 &\text{otherwise}.
			\end{cases}
		\end{equation*}
	\item Finally, we consider the type derived from $K$ in which $v$ is not
		included in any function's image.
		Hence, we define
		\begin{equation*}
			\Gamma_3 \coloneqq A\Bigl[t', \left( \funcadd{\partition}{v}{0}, Z_1, \phi_1, \ldots, Z_q, \phi_q\right)\Bigr].
		\end{equation*}
\end{enumerate}

Next, we prove that one can compute $A[t,K]$ by taking the maximum of $\Gamma_1, \Gamma_2$ and $\Gamma_3$.
\begin{lemma}\label{lemma:forget_arb_graph}
	Let $t \in \mathcal{T}$ be a forget node with a single child $t'$ such that $X_{t'} = X_t \cup \{v\}$. Let $K = \left( \partition, Z_1, \phi_1, \ldots, Z_q, \phi_q \right)  \in \generaltype{t}$ and $\Gamma_1, \Gamma_2$ and $\Gamma_3$ be defined as above.Then, it holds that
\begin{equation*}
	A[t,K] = \max \Big\{ \Gamma_1, \Gamma_2, \Gamma_3 \Big\}.
\end{equation*}	
\end{lemma}

\begin{proof}
	First, we show that $\Gamma_1, \Gamma_2$ and $\Gamma_3$ are well-defined.
	\begin{claim}
		It holds that each
		\begin{equation*}
			K' \in \biggl( \genupdforgetone\left( K \right) \cup \genupdforgettwo\left( K \right) \cup \Bigl\{ \bigl( \funcadd{\partition}{v}{0}, Z_1, \phi_1, \ldots, Z_q, \phi_q\bigr)\Bigr\} \biggr)
		\end{equation*}
		is a valid type for the node $t'$.
	\end{claim}

	\begin{claimproof}
		Since $K$ is a valid type for the node $t$, every type
		$K' \in \genupdforgetone\left( K \right)$ is also a valid type
		for $t'$, because the function $g_{u,v}$ is trivially an
		injective homomorphism. Similarly, the type
		\begin{equation*}
			K' = \left(
			\funcadd{\partition}{v}{0}, Z_1, \phi_1, \ldots, Z_q,
		\phi_q\right)
		\end{equation*}
		is also a valid type for the node $t'$, as it differs from $K$
		only in the first function.

		Therefore, let us consider a type $K' \in \genupdforgettwo\left( K \right)$.
		It is easy to see that $K'$ satisfies \cref{item:part}
		because $K'$ differs from $K$ only on $Z_j$ and $\phi_j$.
		Moreover, observe that $u \not\in Z_j$ since $\phi_j(u) = \,\downvertex\,$, and
		\begin{equation*}
			\abs{Z_j \cup \{u\} } = \Bigl(\abs{Z_j} + 1\Bigr) = \Bigl(\partition^{-1}(j) + 1\Bigr) = \Bigl(\funcadd{\partition}{v}{j}\Bigr)^{-1}(j),
		\end{equation*}
		hence $K'$ satisfies \cref{item:Z_i_nonempty} as well.
		Next, recall that $\phi_j$ is an injective homomorphism and for any neighbour $u' \in Z_j$ of $u$ in $H$,
		there exists an edge between $\phi_j(u)$ and $\phi_j(u')$ because
		$\left( N_H(u) \cap Z_j \right)  \subseteq \phi_j^{-1}\Bigl( N_G(v) \Bigr)$.
		Therefore, $\funcadd{\phi_j}{u}{v}$ is an injective homomorphism.
		Moreover, since $\left( \funcadd{\phi_j}{u}{v} \right)$ differs from
		$\phi_j$ only on $u$, there is no edge between two vertices
		that are mapped to $\upvertex$ and $\downvertex$.
		Therefore, $K'$ satisfies \cref{item:phi_i} in \cref{definition:type_gen_graph}
		and hence is a valid type for the node $t'$.
	\end{claimproof}
	
	In the following, we show that $A[t,K] \leq \max \Big\{ \Gamma_1,
	\Gamma_2, \Gamma_3 \Big\}$ and $A[t,K] \geq \max \Big\{ \Gamma_1,
	\Gamma_2, \Gamma_3 \Big\}$, starting with the former.

	\begin{claim}\label{claim:forget_arb_graph_claim_1}
		It holds that $A[t,K] \leq \max \Big\{ \Gamma_1, \Gamma_2, \Gamma_3 \Big\}$.
	\end{claim}

	\begin{claimproof}
		Let $\mathcal{P} = \left( h_1,
			\ldots, h_p,
			S_1, \overline{h}_1, \ldots, S_q, \overline{h}_q
		\right)$ be a witness for $A[t,K]$ such that $A[t,K] =
		\numfullcopies{ \mathcal{P} } $. We have three possible
		cases:
		\begin{enumerate}
			\item $\exists\, 1 \leq j \leq p$ such that $v \in \im\left( h_j \right) $,
			\item $\exists\, 1 \leq j \leq q$ such that $v \in \im\left( \overline{h}_j \right) $,
			\item $v \not\in \biggl(\Bigl(\bigcup_{i = 1}^{p} \im(h_i) \Bigr) \cup \Bigl(\bigcup_{i = 1}^{q} \im(\overline{h}_i) \Bigr)\biggr)$.
		\end{enumerate}

		If the first case holds, then we define $u \coloneqq h_j^{-1}(v)$,
		\begin{equation*}
			\mathcal{P}' \coloneqq \left( h_1, \ldots, h_{j-1}, h_{j+1}, \ldots, h_p, S_1, \overline{h}_1, \ldots, S_q, \overline{h}_q, \{u\}, h_j \right)
		\end{equation*}
		and let
		\begin{equation*}
			K' \coloneqq \left( \funcadd{\partition}{v}{q+1} , Z_1, \phi_1, \ldots, Z_q, \phi_q, \{u\}, g_{u,v} \right) \in \genupdforgetone\left( K \right) .
		\end{equation*}
		Note that $v$ belongs to the image of the $(q+1)$-th imprint
		function $g_{u,v}$. Moreover, it is easy to check that
		$g_{u,v}$ is the imprint of $h_j$. Hence, $\mathcal{P}'$ has
		type $K'$ and
		\begin{equation*}
			A[t,K] = \numfullcopies{\mathcal{P}} = \Bigl( \numfullcopies{\mathcal{P}'} + 1 \Bigr)  \leq \Bigl(A[t', K']+1\Bigr) \leq \Bigl(\max_{K' \in \genupdforgetone\left( K \right)} A[t', K'] + 1\Bigr) = \Gamma_1.
		\end{equation*}

		If the second case holds, then we let $u \coloneqq
		\left( \overline{h}_j \right) ^{-1}(v)$. By \cref{eq:phi_j_packing_to_type_arb}, it holds that $\phi_j(u) =
		\,\downvertex$ since $v \not\in  X_t$. Therefore $u \not\in
		Z_j$.
		Note that for any neighbour $u' \in Z_j$ of $u$ in
		$H$,
		there exists an edge between $\phi_j(u') = \overline{h}_j(u')$
		and
		$\overline{h}_j(u) = v$ since $\overline{h}_j$ is an injective
		homomorphism from $H[S_j]$ to $G_t$.
		Therefore it holds that $N_H(u)
		\subseteq \phi_j^{-1}\Bigl( N_G(v) \Bigr)$.
		We let $\mathcal{P}' \coloneqq \mathcal{P}$ and
		\begin{equation*}
			K' \coloneqq \left( \funcadd{\partition}{v}{j}, Z_1, \phi_1, \ldots, Z_j \cup \{u\} , \funcadd{\phi_j}{u}{v}, \ldots, H_q, \phi_q \right) \in \genupdforgettwo\left( K \right).
		\end{equation*}

		Observe that $\mathcal{P}'$ is a valid partial packing for the
		node $t'$. Moreover, $\mathcal{P}'$ has type $K'$ since
		$\funcadd{\phi_j}{u}{v}$ is the imprint of $\overline{h}_j$ on
		$X_{t'}$, as $\overline{h}_j(u) = v \in X_t$. Hence, it holds
		that
		\begin{equation*}
			A[t,K] = \numfullcopies{\mathcal{P}} = \numfullcopies{\mathcal{P}'} \leq A[t', K'] \leq \Bigl(\max_{K' \in \genupdforgetone\left( K \right)} A[t', K']\Bigr) = \Gamma_2.
		\end{equation*}

		Finally, if the third condition holds, then $v$ doesn't belong
		to any partial or full copy of $H$ in $\mathcal{P}$. Hence we
		let
		$\mathcal{P}' \coloneqq \mathcal{P}$ and $K' \coloneqq \left(
			\funcadd{\partition}{v}{0}, Z_1, \phi_1, \ldots, Z_q,
		\phi_q\right)$. It is easy to verify that $\mathcal{P}'$ has
		type $K'$, and therefore
		\begin{equation*}
			A[t,K] = \numfullcopies{\mathcal{P}} = \numfullcopies{\mathcal{P}'} \leq A[t', K'] = \Gamma_3.
		\end{equation*}

		All in all, in any of the three cases, $A[t,K]$ is smaller than either
		$\Gamma_1, \Gamma_2$ or $\Gamma_3$, which implies that
		\begin{equation}\label{eq:forget_ineq_one_arb}
			A[t,K] \leq \max \{\Gamma_1, \Gamma_2, \Gamma_3\}.
		\end{equation}
		This proves the claim.
	\end{claimproof}
	
	Next, we prove the remaining direction of the inequality.

	\begin{claim}\label{claim:forget_arb_graph_claim_2}
		It holds that $A[t,K] \geq \max \Big\{ \Gamma_1, \Gamma_2, \Gamma_3 \Big\}$.
	\end{claim}

	\begin{claimproof}
		The claim is equivalent to proving that $A[t,K] \geq \Gamma_i$ for $i \in
		\{1,2,3\}$. To show that $A[t,K] \geq \Gamma_1$, let $u \in
		V(H)$ such that $\Gamma_1 = A[t',K^{1}] + 1$ where 
		\begin{equation*}
			K^{1} \coloneqq \left( \funcadd{\partition}{v}{q+1} , Z_1, \phi_1, \ldots, Z_q, \phi_q, \{u\}, g_{u,v} \right) \in \genupdforgetone(K).
		\end{equation*}
		Now let
		\begin{equation*}
			Q^{1} \coloneqq \left( b^{1}_1, \ldots, b^{1}_p, \eta^{1}_1, \alpha^{1}_1,\ldots, \eta^{1}_{q + 1}, \alpha^{1}_{q + 1} \right)
		\end{equation*}
		be a witness for $A[t', K^1]$ such that $A[t', K^{1}] =
		\numfullcopies{Q^{1}}$. Since $\im(g_{u,v}) \subseteq
		\{\downvertex, v\} $
		and $g_{u,v}$ is the imprint of $\alpha^{1}_{q + 1}$, it holds
		that $\eta^{1}_{q + 1} = V(H)$, because otherwise there would
		be $x \in V(H)$ such that $g_{u,v}(x) = \upvertex$. Hence
		$\alpha^{1}_{q + 1}$ is a injective homomorphism from $V(H)$ to
		$G_{t'}$ such that $\im(\alpha^{1}_{q + 1}) \cap X_{t'} =
		\{v\}$. We define the partial packing
		\begin{equation*}
			\mathcal{P}^{1} \coloneqq \left( b^{1}_1, \ldots, b^{1}_p, \alpha^{1}_{q+1}, \eta^{1}_1, \alpha^{1}_1,\ldots, \eta^{1}_{q}, \alpha^{1}_{q} \right).
		\end{equation*}
		Note that $\mathcal{P}^{1}$ has an extra full copy of $H$ compared to $\mathcal{Q}^{1}$
		which is $\alpha^{1}_{q+1}$.
		It is easy to verify that $\mathcal{P}^{1}$ is a valid partial packing for the node $t$.

		Moreover, since $\mathcal{P}^{1}$ and $K$ are restrictions of $\mathcal{Q}^{1}$ and $K^{1}$, respectively, \cref{item:tp_part,item:tp_phi_j,item:tp_Z_j}
		hold for $\left( \mathcal{P}^{1}, K \right)$ which are inherited from $\left( Q^{1}, K^{1} \right) $.
		Hence,
		\begin{equation}\label{eq:forget_arb_graph_claim_2_ineq_1}
			A[t,K] \geq \numfullcopies{ \mathcal{P}^{1} } = \left( \numfullcopies{ Q^{1} } + 1 \right) = \left( A[t', K^{1}] + 1  \right) = \Gamma_1.
		\end{equation}

		To show that $A[t,K] \geq \Gamma_2$, let $1 \leq j \leq q$ and
		$u \in \phi_j^{-1}(\downvertex)$ such that $N_H(u) \subseteq
		\phi_j^{-1}\Bigl( N_G(v) \Bigr)$ and  $\Gamma_2 = A[t', K^{2}]$
		where
		\begin{equation*}
			K^{2} \coloneqq \left( \funcadd{\partition}{v}{j}, Z_1, \phi_1, \ldots, Z_j \cup \{u\} , \funcadd{\phi_j}{u}{v}, \ldots, Z_q, \phi_q \right) \in \genupdforgettwo(K).
		\end{equation*}
		Let $Q^{2} = \left( b^{2}_1, \ldots, b^{2}_p, \eta^{2}_1,
		\alpha^{2}_1,\ldots, \eta^{2}_{q }, \alpha^{2}_{q } \right)$ be
		a witness for $A[t', K^{2}]$ such that $A[t', K^{2}] =
		\numfullcopies{Q^{2}}$. The function $\funcadd{\phi_j}{v}{j}$
		is an imprint of $\alpha^{2}_j$ which implies that
		$\alpha^{2}_j(u) = v$.
	
		Using this fact and $X_{t'} = X_t \cup \{v\}$, it is easy to verify that
		$\mathcal{P}^{2} \coloneqq \mathcal{Q}^{2}$ is a valid partial
		packing for the node $t$. Moreover, $K^{2}$ differs from $K$
		only on the $j$'th components, i.e., $Z_j \cup \{u\}$ instead of $Z_j$
		and $\funcadd{\phi_j}{u}{v}$ instead of $\phi_j$.
		Furthermore, $\phi_j$ is the imprint of
		$\alpha^{2}_j$ on $X_t$ which follows from the fact that $\funcadd{\phi_j}{u}{v}$
		is the imprint of $\alpha^{2}_j$ on $X_{t'}$.
		Hence $\mathcal{P}^{2}$ is a
		witness for $A[t,K]$ and we have
		\begin{equation}\label{eq:forget_arb_graph_claim_2_ineq_2}
			A[t,K]  \geq \numfullcopies{ \mathcal{P}^{2} } = \numfullcopies{ Q^{2} } = 
			A[t', K^{2}] = \Gamma_2.
		\end{equation}

		Finally, to prove that $A[t,K] \geq \Gamma_3$, let $K^{3}
		\coloneqq \left( \funcadd{\partition}{v}{0}, Z_1, \phi_1,
		\ldots, Z_q, \phi_q\right)$ and $Q^{3}$ be a witness for $A[t',
		K^{3}]$ such that $A[t', K^{3}] = \numfullcopies{ Q^{3} } $.
		Since $v$ is assigned 0 by $\funcadd{\partition}{v}{0}$, no
		full or partial copy contains $v$. Therefore, $\mathcal{P}^{3}
		\coloneqq Q^{3}$ is also a witness for $A[t,K]$ since $v
		\not\in  X_t$. Hence, we have that
		\begin{equation}\label{eq:forget_arb_graph_claim_2_ineq_3}
			A[t,K] \geq \numfullcopies{ \mathcal{P}^{3} } = \numfullcopies{ Q^{3} }  = A[t', K^{3}] = \Gamma_3.
		\end{equation}

		All in all, it holds that $A[t,K] \geq \max \{\Gamma_1
		\Gamma_2, \Gamma_3\}$ by \cref{eq:forget_arb_graph_claim_2_ineq_1}, \cref{eq:forget_arb_graph_claim_2_ineq_2} and \cref{eq:forget_arb_graph_claim_2_ineq_3}.
		This proves the claim.
	\end{claimproof}
	
	The lemma follows from \cref{claim:forget_arb_graph_claim_1,claim:forget_arb_graph_claim_2}.
\end{proof}

\subsubsection*{Join Node}
Let $t \in \mathcal{T}$ be a join node with two children $t_1, t_2 \in \mathcal{T}$ such that $X_t = X_{t_1} = X_{t_2}$.
Recall that by the properties of the tree decomposition, there exists no edges in $G$ between a vertex in $V_{t_1} \setminus X_{t}$ and a vertex in $V_{t_2} \setminus X_t$.
Next, we define a notion of compatible types, which formalizes the relationship between types
for the child nodes $t_1$ and $t_2$.
The following definition specifies the conditions under which two types are considered compatible.
\begin{definition}[Compatible Types]\label{definition:compatible_types}
Let 
\begin{align*}
	K^1 &= \left( \partition^{1}, Z^{1}_1, \phi^{1}_1, \ldots, Z^{1}_q, \phi^{1}_q \right) \in \generaltype{t_1} \quad \text{and} \\
	K^2 &= \left( \partition^{2}, Z^{2}_1, \phi^{2}_1, \ldots, Z^{2}_q, \phi^{2}_q \right) \in \generaltype{t_2}.
\end{align*}

The types $K^{1}$ and $K^{2}$ are compatible if
\begin{enumerate}
	\item $\partition^{1} = \partition^{2}$,
	\item $Z^{1}_i = Z^{2}_i$ for $1 \leq i \leq q$,
	\item The functions $\phi^{1}_i$ and $\phi^{2}_i$ agree on $Z^{1}_i = Z^{2}_i$, i.e. $\funcrest{\phi^{1}_i}{Z^{1}_i} = \funcrest{\phi^{2}_i}{Z^{2}_i}$,
	\item $\{\phi^{1}_i(u), \phi^{2}_i(u)\} \neq \{\downvertex\}$ for all $u \in V(H)$.
\end{enumerate}

For such compatible $K^1$ and $K^2$, we define
\begin{equation}\label{equation:K1_K2_oplus}
	K^1 \oplus K^2 \coloneqq \left( \partition, Z_1, \phi^{1}_1 \oplus \phi^{2}_1, \ldots, Z_q, \phi^{1}_{q} \oplus \phi^{2}_{q} \right),
\end{equation}
where $Z_i \coloneqq Z^{1}_i = Z^{2}_i$ for $i \in [q]$, and $\partition = \partition^{1} = \partition^{2}$.
For each $i \in [q]$, the function $\phi^{1}_i \oplus \phi^{2}_i$ is defined on $u \in V(H)$ as
\begin{equation}\label{eq:phi_oplus}
	\left( \phi^{1}_i \oplus \phi^{2}_i \right)(u) \coloneqq \begin{cases}
		\upvertex &\text{if } \phi^{1}_i(u) = \phi^{2}_i(u) = \upvertex  \\
			 x &\text{if } \phi^{1}_i(u) = \phi^{2}_i(u) = x\\
			 \downvertex &\text{if } \{ \phi^{1}_i(u), \phi^{2}_i(u) \}  = \{\upvertex, \downvertex\}.
	\end{cases}
\end{equation}
\end{definition}
Next, we show that $K^{1} \oplus K^{2}$ is a type for the node $t$.
\begin{lemma}
	Let $t \in \mathcal{T}$ be a join node with two children $t_1, t_2 \in \mathcal{T}$ such that $X_t = X_{t_1} = X_{t_2}$.
	Moreover, let
	\begin{align*}
		K^1 &= \left( \partition^{1}, Z^{1}_1, \phi^{1}_1, \ldots, Z^{1}_q, \phi^{1}_q \right) \in \generaltype{t_1}\\
		K^2 &= \left( \partition^{2}, Z^{2}_1, \phi^{2}_1, \ldots, Z^{2}_q, \phi^{2}_q \right) \in \generaltype{t_2}
	\end{align*}
	be two compatible types.
	Then, it holds that
	\begin{equation*}
		\left( K^{1} \oplus K^{2} \right) \in \generaltype{t}.
	\end{equation*}
\end{lemma}

\begin{proof}
	Let
	\begin{equation*}
		K^1 \oplus K^2 = \left( \partition, Z_1, \phi^{1}_1 \oplus \phi^{2}_1, \ldots, Z_q, \phi^{1}_{q} \oplus \phi^{2}_{q} \right)
	\end{equation*}
	where $Z_i \coloneqq Z^{1}_i = Z^{2}_i$ and $\partition = \partition^{1} = \partition^{2}$.
	Next, we will show that $K^{1} \oplus K^{2}$ satisfies the conditions in \cref{definition:type_gen_graph}.
	First, by definition, $\partition = \partition^{1} = \partition^{2}$
	is a function from $X_{t_1} = X_{t_2} = X_t$ to $[0,q]$.
	Moreover, we have that
	$Z_i = Z^{1}_i = Z^{2}_i$ is a non-empty subset of $H$ where
	\begin{equation*}
		\abs{Z_i} = \abs{Z^{1}_i} = \left( \partition^{1} \right)^{-1}(i) = \partition^{-1}(i).
	\end{equation*}
	Since $K^{1}$ and $K^{2}$ are compatible, for $i \in [q]$ we have that
	\begin{equation}\label{eq:phi_1_2_rest_eq}
		\funcrest{\phi^{1}_i}{Z^{1}_i} = \funcrest{\phi^{2}_i}{Z^{2}_i}.
	\end{equation}
	Then, by \cref{definition:type_gen_graph}, $\funcrest{\phi^{1}_i}{Z^{1}_i}$
	(and therefore $\funcrest{\phi^{2}_i}{Z^{2}_i}$) is an injective homomorphism
	from $Z^{1}_i = Z^{2}_i$ to $G[X_{t_1}(i)] = G[X_{t_2}(i)] = G[X_{t}(i)]$.

	Now let $\zeta_i$ denote the function $\phi^{1}_i \oplus \phi^{2}_i$ and let
	$u \in Z_i = Z^{1}_i = Z^{2}_i$. We have that $\phi^{1}_i(u) = \phi^{2}_i(u) \in X_t$
	by \cref{eq:phi_1_2_rest_eq}, which implies that $\zeta_i(u) = \phi^{1}_i(u) = \phi^{2}_i(u)$
	by the definition of $\phi^{1}_i \oplus \phi^{2}_i$. Therefore, 
	$\funcrest{\zeta_i}{Z_i} = \funcrest{\phi^{1}_i}{Z_i} = \funcrest{\phi^{2}_i}{Z_i}$,
	which further implies that $\funcrest{\zeta_i}{Z_i}$ is an injective
	homomorphism from $Z_i$ to $G[X_t]$. Suppose now for a contradiction
	that there exists $u_1 \in \zeta_i^{-1}(\upvertex)$, $u_2 \in \zeta_i^{-1}(\downvertex)$
	and $\{u_1,u_2\} \in E(H)$. We have that $\phi^{1}_i(u_1) = \phi^{2}_i(u_1) = \upvertex$
	and $\{\phi^{1}_i(u_2), \phi^{2}_i(u_2)\} = \{\upvertex, \downvertex\}$.
	This implies that there exists $\ell \in \{1,2\}$ such that $\phi^{\ell}_i(u_1) = \upvertex$
	and $\phi^{\ell}_i(u_2) = \downvertex$, which contradicts the third condition in
	\cref{definition:type_gen_graph}.

	Therefore, $K^{1} \oplus K^{2}$ satisfies the conditions in \cref{definition:type_gen_graph}
	and it holds that $K^{1} \oplus K^{2} \in \generaltype{t}$.
\end{proof}

Next, we demonstrate how the value of $A[t,K]$ can be determined recursively
from the children's values based on compatible types.
Given a join node $t \in \mathcal{T}$ with two children $t_1,t_2 \in \mathcal{T}$,
let us define
\begin{equation*}
	\Omega(t,K) \coloneqq \max_{\substack{K^1 \in \generaltype{t_1}, K^2 \in \generaltype{t_2} \\ \text{ s.t. } K = K^1 \oplus K^2}} A[t_1, K^1] + A[t_2, K^2].
\end{equation*}

\begin{lemma}\label{lemma:join_arb_graph}
	Let $t \in \mathcal{T}$ be a join node with two children $t_1, t_2 \in \mathcal{T}$ such that $X_t = X_{t_1} = X_{t_2}$.
	For each $K \in \generaltype{t}$, we have that
	\begin{equation*}
		A[t,K] = \Omega(t,K).
	\end{equation*}
\end{lemma}

In the following, using two separate lemmas, we will show that $A[t,K] \leq
\Omega(t,K)$ and $A[t,K] \geq \Omega(t,K)$, starting with the former, and thereby prove
\cref{lemma:join_arb_graph}.

\begin{lemma}\label{lemma:join_arb_graph_ineq_1}
	Let $t \in \mathcal{T}$ be a join node and $K = \left( \partition, Z_1, \phi_1, \ldots, Z_q, \phi_q \right)  \in \generaltype{t}$. Then we have that
	\begin{equation*}
		A[t,K] \leq \Omega(t,K).
	\end{equation*}
\end{lemma}

\begin{proof}
	Let $\mathcal{P} = \left( h_1, \ldots, h_p,
	S_1, \overline{h}_1, \ldots, S_q, \overline{h}_q
	\right)$ be a witness for $A[t,K]$ such that $A[t,K] = \numfullcopies{ \mathcal{P} 
	}$.
	For each $1 \leq i \leq p$, it holds that either $\im(h_i) \subseteq \left( V_{t_1} \setminus V_{t_2} \right) $
	or $\im(h_i) \subseteq \left( V_{t_2} \setminus V_{t_1} \right)$.
	Otherwise, since there can be no edges between vertices in $ V_{t_1} \setminus V_{t_2} $
	and $ V_{t_2} \setminus V_{t_1} $, it holds that $\im(h_i) \cap X_t \neq \emptyset$, which is a contradiction
	to the fact that $h_i$ is a full copy of $H$ in $\mathcal{P}$.
	Hence we can define
	\begin{equation*}
		B_1 \coloneqq \Bigl\{ h_i \mid 1 \leq i \leq p,\,  \im(h_i) \subseteq \left( V_{t_1} \setminus V_{t_2} \right)  \Bigr\} = \{h^{1}_1, \ldots, h^{1}_{p_1}\} 
	\end{equation*}
	and $B_2 \coloneqq \Bigl(  \{h_1, \ldots, h_p\}  \setminus B_1  \Bigr) = \{h^{2}_1, \ldots, h^{2}_{p_2}\} $
	where $p = p_1 + p_2$.

	Next, for each $\ell \in \{1,2\}$ and $i \in [q]$, we define
	\begin{align*}
		\eta^{\ell}_i &\coloneqq  \left( \overline{h}_i \right)^{-1}\left( V_{t_\ell} \right), \text{ and}\\
		\alpha^{\ell}_i &\coloneqq \funcrest{\overline{h}_i}{\eta^{\ell}_i}.
	\end{align*}
	Note that by \cref{definition:part_packing_gen_graph}, $(\overline{h}_i)^{-1}(X_t) \neq \emptyset$,
	therefore $ \Bigl( \left( \overline{h}_i \right)^{-1}\left( V_{t_\ell} \right) \Bigr) \supseteq \Bigl( (\overline{h}_i)^{-1}(X_t) \Bigr) \neq \emptyset$.
	Hence $\eta^{\ell}_i$ is not an empty set.
	Similarly, for $\ell \in \{1,2\}$, we define the partial packing
	\begin{align*}
		Q^{\ell} &\coloneqq \left( h^{\ell}_1,\, \ldots,\, h^{\ell}_{p_1},\, \eta^{\ell}_1,\, \alpha^{\ell}_1,\, \ldots,\, \eta^{\ell}_{q},\, \alpha^{\ell}_{q} \right)
	\end{align*}
	and the type
	\begin{equation}\label{eq:K_1_2_definition}
		K^{\ell} \coloneqq \left( \partition, Z_1, \phi^{\ell}_1, \ldots, Z_{q}, \phi^{\ell}_{q} \right),
	\end{equation}
	where the function $\phi^{\ell}_i$ is defined for $i \in [q]$ as
	
	\begin{equation}\label{eq:phi_j_defn}
		\phi^{\ell}_i(x) \coloneqq \begin{cases}
			\overline{h}_i(x) &\text{if } x \in (\overline{h}_i)^{-1}(X_t)\\
			\downvertex &\text{if } x \in (\overline{h}_i)^{-1}(V_{t_{\ell}} \setminus X_{t})\\			
			\upvertex &\text{if } x \not\in (\overline{h}_i)^{-1}(V_{t_{\ell}}),\text{ i.e. } x \not\in \eta_i^{\ell}. 
		\end{cases}
	\end{equation}
	for $x \in V(H)$.

	Next, we prove that $K^1,K^2$ are compatible types, and that $Q^1$ and $Q^2$
	are witnesses for $A[t_1, K^1]$ and $A[t_2, K^2]$, respectively.

	\begin{claim}
		It holds that $K^1 \in \generaltype{t_1}$ and $K^2 \in \generaltype{t_2}$.
	\end{claim}

	\begin{claimproof}
		For the sake of presentation, we will prove that $K^1 \in \generaltype{t_1}$,
		however the same arguments work for proving $K^2 \in \generaltype{t_2}$ as well.
		Since the function $\partition$ and the sets $\{Z_i\}_{i \in [q]}$ in $K_1$
		are inherited from $K$, and $K \in \generaltype{t}$,
		the first two conditions in \cref{definition:type_gen_graph}
		hold for $K_1$ as well.
		To demonstrate that the third condition in \cref{definition:type_gen_graph} holds,
		fix an $i \in [q]$.
		Observe that $\phi^{1}_i$ and $\phi_i$ agree on
		$(\overline{h}_i)^{-1}(X_t)$, because for $x \in (\overline{h}_i)^{-1}(X_t)$
		we have $\phi_i(x) = \overline{h}_i(x)$ by \cref{eq:phi_j_packing_to_type_arb}
		and $\phi^{1}_i(x) = \overline{h}_i(x)$ by definition.
		In particular, this implies that
		\begin{equation}\label{eq:phi_restr_H_i}
			\funcrest{\phi^{1}_i}{Z_i} = \funcrest{\phi_i}{Z_i},
		\end{equation}
		which is an injective homomorphism from $H[Z_i]$ to $G[X_t(i)] = G[X_{t_1}(i)]$.

		Finally, define the sets
		\begin{align*}
			S_1 &\coloneqq  \left( \phi^{1}_i \right)^{-1}(X_{t_1}),\\
			S_2 &\coloneqq \Bigl( \phi^{1}_i \Bigr)^{-1}\Bigl( \{\downvertex\}  \Bigr) = \biggl(\Bigl( \overline{h}_i \Bigr)^{-1}\left( V_{t_1} \setminus X_{t_1} \right)\biggr) \subseteq \phi^{-1}_i\left( \{\downvertex\}  \right), \\
			S_3 &\coloneqq \Bigl( \phi^{1}_i \Bigr)^{-1}\Bigl( \{\upvertex\}  \Bigr) = \biggl( V(H) \setminus \Bigl[ (\overline{h}_i)^{-1}(V_{t_1})\Bigr] \biggr) \\
			    &=  \Biggl( \biggl( V(H) \setminus \left( \overline{h}_i \right) ^{-1}(V_t) \biggr) \cup \left( \overline{h}_i \right) ^{-1}\left(  V_{t_2} \setminus X_{t_2} \right)  \Biggr)\\
			    &= \phi_i^{-1}\Bigl( \{\upvertex\}  \Bigr) \cup \left( \overline{h}_i \right) ^{-1}\Bigl(  V_{t_2} \setminus X_{t_2} \Bigr).
		\end{align*}
		Note that $\left( S_1 \cup S_2 \cup S_3 \right) = V(H)$.
		The third condition in \cref{definition:type_gen_graph} is equivalent to
		showing that $S_1$ is an $(S_2, S_3)$ separator in $H$.
		Observe that there are no edges between $S_2 \subseteq \phi^{-1}_i\left( \{\downvertex\}  \right)$
		and $\phi_i^{-1}\left( \{\upvertex\}  \right)$, because $\phi_i$
		already satisfies the third condition in \cref{definition:type_gen_graph}.
		Moreover, there are no edges between $S_2 = \Bigl( \left( \overline{h}_i \right)^{-1}\left( V_{t_1} \setminus X_{t_1} \right)\Bigr)$ and $\left( \overline{h}_i \right) ^{-1}\Bigl(  V_{t_2} \setminus X_{t_2} \Bigr)$,
		because otherwise there would be an edge between $G[V_{t_2} \setminus X_{t}]$
		and $G[V_{t_1} \setminus X_t]$ which follows from $\overline{h}_i$ being an
		injective homomorphism. However, this is a contradiction to the fact that 
		$t_1,t_2$ are children of a join node $t$. This implies that there exist
		no edges between $S_2$ and $S_3$, and hence \cref{item:phi_i} in
		\cref{definition:type_gen_graph} holds for $K^{1}$. Therefore $K^{1} \in \generaltype{t_1}$.
	\end{claimproof}

	\begin{claim}
		It holds that $K^{1}$ and $K^{2}$ are compatible.
	\end{claim}

	\begin{claimproof}
		We prove the claim by going over the conditions in \cref{definition:compatible_types}.
		The first two
		conditions hold by the definition of $K^1$ and $K^2$.

		To prove that $\phi^{1}_i$ and $\phi^{2}_i$ agree on $Z_i$, let $x \in Z_i$.
		Since $\mathcal{P}$ has type $K$, it holds that
		\begin{equation*}
			Z_i = \left( \overline{h}_i\right)^{-1}\left( X_t \right).
		\end{equation*}

		Therefore,
		by definition of $\phi^{1}_i$ and $\phi^{2}_i$, we get
		\begin{equation*}
			\phi^{1}_i(x) = \overline{h}_i(x) = \phi^{2}_i(x),
		\end{equation*}
		and hence the third condition holds.

		Finally, to prove the last condition, suppose that $\phi^{1}_i(x) = \phi^{2}_i(x) = \,\downvertex\,$
		for some $i \in [q]$ and $x \in V(H)$. Then, we have $x \in (\overline{h}_i)^{-1}\left( V_{t_1} \setminus X_{t_1} \right)$ and $x \in (\overline{h}_i)^{-1}\left( V_{t_2} \setminus X_{t_2} \right)$, which implies
		\begin{equation*}
			\overline{h}_i(x) \in \biggl(\Bigl( V_{t_1} \setminus X_t \Bigr) \cap \Bigl( V_{t_2} \setminus X_t \Bigr)\biggr) = \emptyset
		\end{equation*}
		which is a contradiction. Hence the last condition in \cref{definition:compatible_types}
		holds as well, and therefore $K^1$ and $K^2$ are compatible.
	\end{claimproof}
	Since $K^1$ and $K^2$ are compatible, $K^{1} \oplus K^{2}$ is well-defined.
	Next, we show that $K^{1} \oplus K^{2}$ is equal to $K$, as required.
	\begin{claim}\label{claim:K_12_compatible}
		It holds that
		\begin{equation*}
			K = K^1 \oplus K^2.
		\end{equation*}
	\end{claim}

	\begin{claimproof}
		By the definition of $K^{1}$ and $K^{2}$, the claim is equivalent
		to showing that $\phi_i = \left( \phi^{1}_i \oplus \phi^{2}_i \right) $
		for all $1 \leq i \leq k$. Therefore we let $1 \leq i \leq q$, $u \in V(H)$ and show that 
		$\phi_i(u) = \left( \phi^{1}_i \oplus \phi^{2}_i \right)(u)$.

		If $\phi_i(u) = \upvertex$, then $u \not\in  \left( \overline{h}_i \right)^{-1}(V_t)$
		by \cref{eq:phi_j_packing_to_type_arb}, which implies that
		$u \not\in \left( \overline{h}_i \right)^{-1}\left( V_{t_1} \right)$
		and $u \not\in \left( \overline{h}_i \right)^{-1}\left( V_{t_2} \right)$.
		This in turn implies that $\phi^{1}_i(u) = \phi^{2}_i(u) = \,\upvertex\,$
		by \cref{eq:phi_j_defn}, and $\left( \phi^{1}_i \oplus \phi^{2}_i \right)(u) = \,\upvertex\, = \phi_i(u)$.

		If $\phi_i(u) = \,\downvertex\,$, then it holds that $u \in \left(
		\overline{h}_i \right)^{-1}\left( V_t \setminus X_t \right) $
		by \cref{eq:phi_j_packing_to_type_arb}. This implies that either
		$u \in \left(
		\overline{h}_i \right)^{-1}\left( V_{t_1} \setminus X_{t_1} \right)$
		or $u \in \left(
		\overline{h}_i \right)^{-1}\left( V_{t_2} \setminus X_{t_2} \right)$,
		since $V_t \setminus X_t = \left( V_{t_1} \setminus X_t \right) \cup \left( V_{t_2} \setminus X_t \right)$
		and $\left( V_{t_1} \setminus X_t \right) \cap \left( V_{t_2} \setminus X_t \right) = \emptyset$.
		Therefore, we have either $\phi^{1}_i(u) = \,\downvertex\,$ or $\phi^{2}_i(u) = \,\downvertex\,$.
		By \cref{eq:phi_oplus}, we have that $\left( \phi^{1}_i \oplus \phi^{2}_i \right)(u) = \downvertex = \phi_i(u)$.

		Finally, if $\phi_i(u) = x \in X_t$, then $u \in \left( \overline{h}_i \right)^{-1}(X_t)$
		and $\phi^{1}_i(u) = \phi^{2}_i(u) = x$ by \cref{eq:phi_j_defn}.
		Once again, we have that $\left( \phi^{1}_i \oplus \phi^{2}_i \right)(u) = x = \phi_i(u)$.
		Therefore, $\left( \phi^{1}_i \oplus \phi^{2}_i \right) = \phi_i$ and $K = K^{1} \oplus K^{2}$.
		This proves the claim.
	\end{claimproof}
		
	Next step is to show that $Q^{1}$ and $Q^{2}$ are valid partial
	packings for $t_1$ and $t_2$, respectively.

	\begin{claim}
		It holds that $Q^{\ell}$ is a valid partial packing for the node $t_\ell$
		for $\ell \in \{1,2\}$.
	\end{claim}

	\begin{claimproof}
		For clarity of presentation, we will demonstrate the proof for $\ell = 1$.
		The same reasoning can then be applied to the case of $\ell = 2$.

		By definition, $\im\left( h^{1}_i \right) \subseteq V_{t_1} \setminus X_t$
		for $i \in [p_1]$, and $h_1$ is a full copy of $H$.
		For each $1 \leq i \leq q$, by definition, $\eta^{1}_i$ is a non-empty subset
		of $V(H)$, and $\alpha^{1}_i$ is an injective homomorphism from
		$H[\eta^{1}_i]$ to $G_{t_1}$, since $\alpha^{1}_i$ is a restriction of $\overline{h}_i$.
		Moreover, $\left( \alpha^{1}_i \cap X_t \right) \subseteq \left( \im( \overline{h}_i) \cap X_t \right) = \emptyset$.
		Since the functions in $\{\alpha^{1}_i\}_{1 \leq i \leq q}$
		are restrictions of the functions in $\{\overline{h}_i\}_{1 \leq i \leq q}$, and
		$\{h^{1}_i\}_{1 \leq i \leq q} \subseteq \{h_i\}_{1 \leq i \leq q}$
		the images of those functions are disjoint.

		Finally, observe that we have $\border{i}{\mathcal{P}} = \border{i}{Q^{1}_i}$.
		Therefore, for $u \in \border{i}{Q^{1}_i} = \border{i}{\mathcal{P}}$,
		we have that $\overline{h}_i(u) \in X_t$. Then $u \in (\overline{h}_i)^{-1}\left( V_{t_1} \right)$
		and by definition of $\alpha^{1}_i$ we have $\alpha^{1}_i(u) = \overline{h}_i(u) \in X_t$.
		Therefore, $Q^{1}$ is a valid partial packing for the node $t_1$.
	\end{claimproof}

	\begin{claim}\label{claim:Q_i_witness}
		$Q^{\ell}$ is a witness for $A[t_\ell, K_\ell]$ for $\ell \in \{1,2\}$.
	\end{claim}

	\begin{claimproof}
		We will prove the claim for the case $\ell = 1$ for the sake of
		presentation, as the argument for $\ell = 2$ case follows
		analogously.

		First, let us focus on \cref{item:tp_part}.
		Let $x \in X_t$ and suppose that $\partition(x) = 0$.
		This implies that for $j \in [1,q]$ we have that
		$x \not\in \im(\overline{h}_j)$.
		Since $\alpha^1_{j}$ is a restriction of $\overline{h}_j$,
		for $j \in [1,q]$ we immediately get that $x \not\in \im\left( \alpha^1_{j} \right) $.
		Hence \cref{item:tp_part} holds in this case.
		Now suppose that $j = \partition(x) \in [1,q]$. Note that
		\begin{equation*}
			\im(\overline{h}_j) \cap X_t = \im(\alpha^{1}) \cap X_t.
		\end{equation*}
		Since $x \in X_t$, we get that
		\begin{equation*}
			x \in \im(\overline{h}_j) \implies x \in \im(\alpha^{1}_j).
		\end{equation*}
		Hence \cref{item:tp_part} holds in this case as well.

		Now let us focus on \cref{item:tp_phi_j}.
		For $1 \leq j \leq q$ and $u \in V(H)$, this translates to

		\begin{equation}\label{eq:phi_j_prove}
			\phi^{1}_j(u) \coloneqq \begin{cases}
				\alpha^{1}_j(u) &\text{if } u \in (\alpha^{1}_j)^{-1}(X_t)\\
				\downvertex &\text{if } u \in (\alpha^{1}_j)^{-1}(V_{t_1} \setminus X_{t})\\			
				\upvertex &\text{if } u \not\in (\alpha^{1}_j)^{-1}(V_{t_1}).
			\end{cases}
		\end{equation}

		Since $\alpha^{1}_j$ is a restriction of the function $\overline{h}_j$,
		it turns out \cref{eq:phi_j_prove} holds immediately by definition
		of $\phi^{1}_j$ in \cref{eq:phi_j_defn}. Hence \cref{item:tp_phi_j}
		holds as well.

		Finally, to show that \cref{item:tp_Z_j} holds, observe that
		\begin{equation*}
			 Z_i  = \phi^{-1}_i(X_t) = \left( \overline{h}_i \right) ^{-1}(X_t) = \left( \phi^{1}_i \right)^{-1}(X_t).
		\end{equation*}
	\end{claimproof}

	By \cref{claim:K_12_compatible,claim:Q_i_witness}, and the fact that
	$p = p_1 + p_2$, we get that
	\begin{equation}\label{equation:join_ineq_1}
		\begin{aligned}
		A[t,K] = \numfullcopies{\mathcal{P} }  = p = (p_1 + p_2) &= \left( \numfullcopies{Q^{1}} + \numfullcopies{Q^{2}} \right)\\
									    &\leq A[t_1, K^1] + A[t_2,K^2]\\
									    &\leq \Omega(k,T).	
		\end{aligned}
	\end{equation}
	Therefore the lemma holds.
\end{proof}

\begin{lemma}\label{lemma:join_arb_graph_ineq_2}
	Let $t \in \mathcal{T}$ be a join node and $K = \left( \partition, Z_1, \phi_1, \ldots, Z_q, \phi_q \right)  \in \generaltype{t}$. Then we have that
	\begin{equation*}
		A[t,K] \geq \Omega(t,K).
	\end{equation*}
\end{lemma}

\begin{proof}
	Let
	\begin{align*}
		K^1 &= \left( \partition^{1}, Z_1, \phi^{1}_1, \ldots, Z_q, \phi^{1}_q \right) \in \generaltype{t_1},\\
		K^2 &= \left( \partition^{2}, Z_1, \phi^{2}_1, \ldots, Z_q, \phi^{2}_q \right) \in \generaltype{t_2}
	\end{align*}
	be two compatible types
	such that $K = K^{1} \oplus K^{2}$ and
	\begin{equation*}
		\Omega(t,K) =  A[t_1, K^{1}] + A[t_2, K^{2}].
	\end{equation*}
	Moreover, let 
	\begin{align*}
		Q_1 &= \left( h^{1}_1,\, \ldots,\, h^{1}_{p_1},\, \eta^{1}_1,\, \alpha^{1}_1,\, \ldots,\, \eta^{1}_q,\, \alpha^{1}_{q} \right),\\
		Q_2 &= \left( h^{2}_1,\, \ldots,\, h^{2}_{p_2},\, \eta^{2}_1,\, \alpha^{2}_1,\, \ldots,\, \eta^{2}_q,\, \alpha^{2}_{q} \right)
	\end{align*}
	be witnesses for $A[t_1, K^{1}]$ and $A[t_2, K^{2}]$ respectively.
	Since $Q^{1}$ has type $K^{1}$ and $Q^{2}$ has type $K^{2}$, we have
	\begin{equation*}
		(\alpha^{1}_j)^{-1}(X_t) = (\alpha^{2}_j)^{-1}(X_t) = Z_j.
	\end{equation*}
	For each $1 \leq j \leq q$, define
	\begin{equation*}
		\overline{h}_j \from \Bigl( \left( \alpha^{1}_j \right) ^{-1}\left( V_{t_1} \right)  \cup \left( \alpha^{2}_j \right) ^{-1}\left( V_{t_2} \right) \Bigr) \to V_t
	\end{equation*}
	by
	\begin{equation*}
		\overline{h}_j(u) \coloneqq \begin{cases}
			\alpha^{1}_j(u) &\text{if } u \in (\alpha^{1}_j)^{-1}(V_{t_1})\\
			\alpha^{2}_j(u) &\text{if } u \in (\alpha^{2}_j)^{-1}(V_{t_2}). 
		\end{cases}
	\end{equation*}
	
	Note that $\overline{h}_j(u)$ is well defined because if $u \in
	(\alpha^{1}_j)^{-1}(V_{t_1}) \cap (\alpha^{2}_j)^{-1}(V_{t_2})$, then
	$\alpha^{1}_j(u) \in X_t$ because otherwise $\{\alpha^{1}_j(u),
	\alpha^{2}_j(u)\} = \{\downvertex\}$, which leads to a contradiction
	by \cref{eq:phi_oplus}. Hence, we get that
	\begin{equation*}
		\alpha^{1}_j(u) = \phi^{1}_j(u) = \phi^{2}_j(u) = \alpha^{2}_j(u).
	\end{equation*}
	For $j \in [q]$, we also define
	\begin{equation}\label{eq:S_i_defn}
		S_i \coloneqq \eta^{1}_i \cup \eta^{2}_i,
	\end{equation}
	and let
	\begin{equation*}
		\mathcal{P} \coloneqq \left( h^{1}_1, \ldots, h^{1}_{p_1}, h^{2}_1, \ldots, h^{2}_{p_2}, S_1, \overline{h}_1, , \ldots, S_q, \overline{h}_q \right).
	\end{equation*}
	Next, we will show that $\mathcal{P}$ is a valid partial packing and has type $K$.
	\begin{claim}\label{claim:join_Pcal_valid_arb}
		It holds that $\mathcal{P}$ is a valid partial packing for the node $t$.
	\end{claim}

	\begin{claimproof}
		To prove the claim, we show that $\mathcal{P}$ satisfies
		\cref{item:h_i_inj,item:h_bar_j,item:im_disj} in \cref{definition:part_packing_gen_graph}.
		First, it is easy to verify that $\{h^{1}_i\}_{i \in [p_1]}$ and $\{h^{2}_i\}_{i \in [p_2]}$
		satisfies \cref{item:h_i_inj}. To show that $\mathcal{P}$ satisfies
		\cref{item:h_bar_j}, observe that $\eta^{1}_i$ and $\eta^{2}_i$
		are both nonempty subsets of $V(H)$, for $i \in [q]$.
		Moreover, by definition, $\overline{h}_i$ is an injective
		homomorphism from
		\begin{equation*}
			\left( \dom{\alpha^{1}_i} \cup \dom{\alpha^{2}_i} \right) = \left( \eta^{1}_i \cup \eta^{2}_i \right) =
			S_i
		\end{equation*}
		to $\left( G_{t_1} \cup G_{t_2} \right) \subseteq G_t$.

		Next, let $u \in \border{i}{\mathcal{P}}$ for some $i \in [q]$.
		Then, there exists $\ell \in \{1,2\}$ such that $u \in \border{i}{Q_1}$,
		which further implies that $\overline{h}_i(u) = \alpha^{\ell}_{i}(u) \in X_t$.
		Finally, it is easy to verify that $\mathcal{P}$ satisfies \cref{item:im_disj}.
		Therefore, $\mathcal{P}$ is a valid partial packing for the node $t$.
	\end{claimproof}

	\begin{claim}\label{claim:join_Pcal_type_K}
		It holds that $\mathcal{P}$ has type $K$.
	\end{claim}

	\begin{claimproof}
		To prove the claim, we show that $(\mathcal{P},K)$ satisfies
		\cref{item:tp_part,item:tp_phi_j,item:tp_Z_j}.
		Let $x \in X_t$ such that $\partition(x) = 0$.
		Then, since $Q^{1}$ has type $K^{1}$, $Q^{2}$ has type $K^{2}$
		and $\partition(x) = \partition^{1}(x) = \partition^{2}(x) = 0$,
		for all $j \in [q]$ it holds that $x \not\in  \Bigl( \im(\overline{h}^{1}_j) \cup \im(\overline{h}^{2}_j) \Bigr) = \im(\alpha_j)$.
		Similarly, if $\partition(x) = \partition^{1}(x) = \partition^{2}(x) = j \neq 0$,
		then $x \in \Bigl( \im(\alpha^{1}_j) \cup \im(\alpha^{2}_j) \Bigr)$, hence
		$x \in \im(\overline{h}_j)$. Thus, $\left( \mathcal{P}, K \right)$
		satisfies \cref{item:tp_part}.

		Next, we show that $\left( \mathcal{P}, K \right)$ satisfies \cref{item:tp_phi_j}.
		\begin{enumerate}
			\item Let $1 \leq i \leq q$ and $u \in
				\overline{h}_i^{-1}\left( V_t \setminus X_t
				\right)$. This implies that either $u \in
				\overline{h}_i^{-1}\left( V_{t_1} \setminus
				X_{t_1} \right) $ or $u \in
				\overline{h}_i^{-1}\left( V_{t_2} \setminus
				X_{t_2} \right)$.
				Therefore, we have $\phi^{1}_i(u) = \,\downvertex\,$ or
				$\phi^{2}_i(u) = \,\downvertex\,$,
				respectively. However, since $K = K^{1} \oplus
				K^{2}$, by \cref{eq:phi_oplus} we get that
				$\phi_i(u) = \,\downvertex\,$.
			\item  If $u \not\in  S_i$, then
				$u \not\in  \eta^{1}_i$ and $u \not\in 
				\eta^{2}_i$ by \cref{eq:S_i_defn}. Then, since
				$Q^{1}$ and $Q^{2}$ are witnesses for $K^{1}$
				and $K^{2}$ respectively, and by
				\cref{eq:phi_j_packing_to_type_arb}, we have
				that $\alpha^{1}_i(u) = \,\upvertex\,$ and
				$\alpha^{2}_i(u) = \,\upvertex\,$. Therefore,
				\cref{eq:phi_oplus} implies that $\phi_i(u) =
				\,\upvertex\,$.
			\item  If $u \in \overline{h}_j^{-1}\left( X_t \right)
				$, then
				\begin{equation*}
					\phi_j(u) = \alpha^{1}_j(u) = \alpha^{2}_j(u) = \overline{h}_j(u),
				\end{equation*}
		\end{enumerate}
		Therefore, $(\mathcal{P}, K)$ satisfies \cref{item:tp_phi_j}.

		Finally, we have $Z_j = (\alpha^{1}_j)^{-1}(X_t) = (\phi^{1}_j)^{-1}(X_t) = \phi_j^{-1}(X_t)$.
		Therefore, $(\mathcal{P}, K)$ satisfies \cref{item:tp_Z_j}
		and hence $\mathcal{P}$ has type $K$.
	\end{claimproof}

	By \cref{claim:join_Pcal_valid_arb,claim:join_Pcal_type_K} we get that
	\begin{equation}\label{equation:join_ineq_2}
		\Omega(t,K) = \Bigl( A[t_1, K^{1}] + A[t_2, K^{2}] \Bigr)  = \Bigl( \numfullcopies{Q^{1}} + \numfullcopies{Q^{2}} \Bigr)  = \bigl(p_1 + p_2\bigr) = \numfullcopies{\mathcal{P}} \leq A[t,K].		
	\end{equation}
	This proves the lemma.
\end{proof}

The proof of \cref{lemma:join_arb_graph} follows from \cref{lemma:join_arb_graph_ineq_1,lemma:join_arb_graph_ineq_2}.

\subsection{Correctness and Running Time}
	Let $\mathcal{A}$ denote the dynamic programming algorithm that fills the table $A[t,f]$
	for each $t \in \mathcal{T}$ and $f \in \cliquetype{t}$,
	following the update rules in \cref{sec:update_dp_clique_packing}.
	Next, we give the proof of \cref{theorem:graph_packing_arbitrary_algo}.
	
	\begin{proof}[Proof of \cref{theorem:graph_packing_arbitrary_algo}]
		Let $G$ be a graph and $\mathcal{T}$ be a tree decomposition of $G$
		of width at most $w$.
		We can compute a nice tree decomposition of $G$
		of the same width as $\mathcal{T}$ with $\mathcal{O}\left(w \cdot n\right)$
		nodes, in time $\mathcal{O}\left(w^{2} \cdot n\right)$ \cite{cyganParameterizedAlgorithms2015}.
		Therefore without loss of generality we can assume that $\mathcal{T}$
		is a nice tree decomposition.

		The update rules for each node of $\mathcal{T}$, and their correctness
		were described in \cref{section:dp_arb_graph}.
		Recall that for the root node $t_{\text{R}}$  of $\mathcal{T}$ we have
		that $X_{t_\text{R}} = \emptyset$ and $V_{t_{\text{R}}} = G$, therefore
		$A[t_{\text{R}}, \emptyset]$ is the solution for the instance $G$.

		Now recall that for each $t \in \mathcal{T}$, and type
		$K = \left( \partition, Z_1, \phi_1, \ldots, Z_q, \phi_q \right)$
		we have $q \leq \abs{X_t} \leq w$ and $Z_i \subseteq V(H)$,
		hence the number of sets $Z_i$ is at most $2^{\abs{H}}$.
		Moreover, the number of different functions $\phi_i \colon V(H) \to X_t \cup \{\upvertex, \downvertex\}$
		is at most $(w+2)^{\abs{H}}$. Hence the number of all tuples of the form
		$(Z_1, \phi_1, \ldots, Z_w, \phi_w)$ is at most
		\begin{equation*}
			\Bigl(2\cdot(w+2)\Bigr)^{\abs{H} \cdot w} = 2^{\mathcal{O}\left(w \cdot \log(w)\right)}.
		\end{equation*}
		Similarly, the number of functions $\partition \colon X_t \to [0,q]$ is at most
		\begin{equation*}
			(q+1)^{\abs{X_t}} \leq (w+1)^{w} = 2^{\mathcal{O}\left(w \cdot \log(w)\right)}.
		\end{equation*}
		Therefore,
		\begin{equation*}
			\abs{\generaltype{t}} = 2^{\mathcal{O}\left(w \cdot \log(w)\right)}.
		\end{equation*}

		If we let $t$ be a node of the tree decomposition,
		then we can compute the value  $A[t,f]$ for a single $f \in
		\generaltype{t}$ in time $2^{\mathcal{O}\left(w \cdot \log(w)\right)} \cdot n^{\mathcal{O}\left(1\right)}$,
		using the recursive rules introduced
		in the previous section. Hence, we can compute all the
		entries $A[t,f]$ for $f \in \generaltype{t}$ in time $2^{\mathcal{O}\left(w \cdot \log(w)\right)} \cdot n^{\mathcal{O}\left(1\right)}$.
		Since $\mathcal{T}$ has $\mathcal{O}\left(w \cdot n\right)$ nodes, this
		proves the theorem.
	\end{proof}